\documentclass[referee, sn-basic, pdflatex]{sn-jnl}

\usepackage{graphicx}%
\usepackage{multirow}%
\usepackage{amsmath,amssymb,amsfonts}%
\usepackage{amsthm}%
\usepackage{mathrsfs}%
\usepackage[title]{appendix}%
\usepackage{xcolor}%
\usepackage{textcomp}%
\usepackage{manyfoot}%
\usepackage{booktabs}%
\usepackage{algorithm}%
\usepackage{algorithmicx}%
\usepackage{algpseudocode}%
\usepackage{listings}%
\usepackage{comment}%
\usepackage{natbib}%
\usepackage{mdwmath}%
\usepackage{mdwtab}%
\usepackage{dsfont}%
\usepackage{caption}
\usepackage{subcaption}

\newcommand{\revision}[1]{{\textcolor{black} {#1}}}

\theoremstyle{thmstyleone}%
\newtheorem{theorem}{Theorem}
\theoremstyle{thmstylethree}%
\newtheorem{definition}{Definition}%

\begin{document}

\title[Article Title]{Modeling Self-Propagating Malware with Epidemiological Models}


\author*[1]{\fnm{Alesia} \sur{Chernikova}}\email{chernikova.a@northeastern.edu}

\author[2]{\fnm{Nicol\`{o}} \sur{Gozzi}} \nomail

\author[3]{\fnm{Nicola} \sur{Perra}} \nomail

\author[1]{\fnm{Simona} \sur{Boboila}} \nomail

\author[1]{\fnm{Tina} \sur{Eliassi-Rad}} \nomail

\author[1]{\fnm{Alina} \sur{Oprea}} \nomail

\affil*[1]{\orgname{Northeastern University}, \orgaddress{\city{Boston}, \state{MA}, \country{USA}}}

\affil[2]{\orgname{ISI Foundation}, \orgaddress{\city{Turin}, \country{Italy}}}

\affil[3]{\orgdiv{School of Mathematical Sciences}, \orgname{Queen Mary University of London}, \orgaddress{\country{UK}}}


\abstract{Self-propagating malware (SPM) is responsible for large financial losses and major data breaches with devastating social impacts that cannot be understated. Well-known campaigns such as WannaCry and Colonial Pipeline have been able to propagate rapidly on the Internet and cause widespread service disruptions. To date, the propagation behavior of SPM is still not well understood. As result, our ability to defend against these cyber threats is still limited. Here, we address this gap by performing a comprehensive analysis of a newly proposed epidemiological-inspired model for SPM propagation, the \emph{Susceptible-Infected-Infected Dormant-Recovered} (SIIDR) model.
We perform a theoretical analysis of the SIIDR model by deriving its basic reproduction number and studying the stability of its disease-free equilibrium points in a homogeneous mixed system. We also characterize the SIIDR model on arbitrary graphs and discuss the conditions for stability of disease-free equilibrium points. We obtain access to 15 WannaCry attack traces generated under various conditions, derive the model's transition rates, and show that SIIDR fits the real data well. We find that the SIIDR model outperforms more established compartmental models from epidemiology, such as SI, SIS, and SIR, at modeling SPM propagation.}

\keywords{Self-propagating Malware, Compartmental Models, Epidemiology, Modeling, Dynamical Systems}

\maketitle

\section{Introduction}
\label{sec:intro}
Self-propagating malware (SPM) is one of today's most concerning cybersecurity threats. Over past years, SPM resulted in huge financial losses and data breaches with high economic and societal impacts. For instance, the infamous WannaCry~\citep{wannacry_2017} attack, first discovered in 2017 and still actively used by attackers nowadays, was estimated to have affected more than $200,000$ computers across $150$ countries worldwide, with economic damages ranging from hundreds of millions to billions of dollars. In May 2021, the Colonial Pipeline~\citep{colonial_pipeline} cyber-attack caused the shut down of the entirety of the Colonial gasoline pipeline system for several days. It affected consumers and airlines along the East Coast of the United States and was deemed a national security threat. Another remarkable worldwide SPM attack is Petya~\citep{enwiki:1090033619}, first discovered in 2016 when it started spreading through phishing emails. Petya represents a family of various types of ransomware responsible for estimated economic damages of over $10$ million dollars~\citep{enwiki:1090033619}. 

Given the current cyber-crime landscape, with new threats emerging daily, tools designed for modeling SPM behavior become crucial. Indeed, a deep understanding of self-propagating malware characteristics provides us opportunities to identify threats, test control strategies, and design proactive defenses against attacks. 
A large body of research on the subject so far has been devoted to the design of methods to detect and mitigate self-propagating malware. Proposed techniques include network traffic signatures~\citep{Kim2004,kumar2020,portfiler2021,Newsome2005} and host-level binary analysis~\citep{chen2017,said2018} used to identify anomalous behavior, software-defined networking (SDN) for ransomware threat detection and mitigation~\citep{Akbanov2019,Alotaibi2021}, as well as evasion-resilient methods for detecting adaptive worms~\citep{Li2014,Newsome2005,portfiler2021}.
However, less attention was dedicated to comparing and finding the most suitable models to capture SPM behavior. Additionally, the majority of existing works on SPM modeling focus on theoretical analyses of infection spreading~\citep{guillen2017study, guillen2018modeling, mishra2007seirs,martinez2021malseirs}, lacking a thorough real-world evaluation of these models.

In this paper, we model the behavior of a well-known SPM attack, WannaCry, based on real-world attack traces. The similarities between the behavior of biological and computer viruses enable us to leverage compartmental models from epidemiology. We adopt a novel compartmental epidemic model called SIIDR~\citep{chernikova2022cyber}, and conduct a thorough analysis to show that it can be used to accurately model SPM spreading dynamics.

\begin{table}[!htpb]
\caption{Terminology and abbreviations used in the paper.}
\begin{tabular}{|l||l|}
\hline
\textbf{Notation} & \textbf{Meaning}\\ 
\hline
$S$&Number of Susceptible Individuals\\
$I$&Number of Infected Individuals\\
$I_D$&Number of Infected Dormant Individuals\\
$R$&Number of Recovered Individuals\\
SPM& Self-Propagating Malware\\
ODE&Ordinary Differential Equation\\
AIC& Akaike Information Criterion\\
ABC&Approximate Bayesian Computation\\
ABC-SMC (SMC)& Sequential Monte-Carlo Approach\\
ABC-SMC-MNN (SMC)&SMC when Covariance Matrix is Calculated\\
&using M Nearest Neighbors of the Particle\\
SI & Susceptible-Infected Model\\
SIS & Susceptible-Infected-Susceptible Model\\
SIR & Susceptible-Infected-Recovered Model\\
SEIR & Susceptible-Exposed-Infected-Recovered Model\\
SIIDR & Susceptible-Infected-Infected Dormant-Recovered Model\\
\hline
\end{tabular}
\label{tab:terms}
\end{table}

First, we study the model assuming a homogeneous mixing of hosts and analytically derive its basic reproduction number $R_0$~\citep{dietz1993estimation, kephart1993measuring,van2008further}. $R_0$ is the number of secondary cases generated by an infectious seed in a fully susceptible population. It describes the epidemic threshold, thus, the conditions necessary for a macroscopic outbreak ($R_0 > 1$)~\citep{fraser2009pandemic, van2008further}. We also investigate equilibrium or fixed points of SIIDR as they provide insights on how to contain or suppress the spreading.

Additionally, computer networks are often represented as graphs, where nodes denote the hosts in the network and edges represent the communication links between them. In any static graph, the propagation of contagion processes depends not only on the transition rates of SPM but also on the spectral properties of the graph~\citep{newman2018networks}. To discuss the important characteristics of SIIDR that illustrate the ability of SPM to successfully propagate through the network in these settings, we represent SIIDR model as a Non-Linear Dynamical System (NLDS) and relaxing the homogeneous mixing assumption. 

Finally, we reconstruct the dynamics of WannaCry spreading analysing real traffic logs. We use the Akaike Information Criterion (AIC)~\citep{akaike1974} to compare how different compartmental models fit the derived epidemic traces. We show that SIIDR captures malware spreading better than classical epidemic models such as SI, SIS, SIR. Indeed, the investigation of real WannaCry attacks showed that consecutive infection attempts originating from the same host are delayed by a variable time interval. This finding suggests the existence of ``dormant" infected state, in which infected hosts temporarily cease to pass infection to their neighbors. Furthermore, calibrating the model to the real data via an Approximate Bayesian Computation technique we determine the transition rates (i.e., model parameters) that characterize WannaCry propagation. 

To summarize, our contributions are the following:
\begin{itemize}
\item We derive the basic reproduction number of the SIIDR model~\citep{chernikova2022cyber} and discuss the stability conditions of the disease-free equilibrium points of the system of ODEs that represents SIIDR under a homogeneous mixing assumption.
\item We derive the conditions for stability of the SIIDR disease-free equilibrium points on arbitrary graphs thus relaxing the homogeneous mixing assumption.
\item We reconstruct the spreading dynamics of an actual SPM (WannaCry) using real-world traces obtained by running a vulnerable version of Windows in a virtual environment.
\item We show that SIIDR outperforms several classical models in terms of capturing WannaCry behavior, and derive the model's transition rates from actual attacks. \end{itemize}

We organize the rest of the paper as follows: first, we provide background information about the WannaCry malware and the most common compartmental models of epidemiology. We also define the threat model and problem statement. Then we introduce the SIIDR model, discuss the derivation of $R_0$ and the stability of the disease-free equilibrium points.
In addition, we present the experimental results that support the findings of the paper. Table~\ref{tab:terms} includes common terminology used in the paper.

\section{Background and Problem Statement}
\label{sec:background}

\subsection{WannaCry Malware}
WannaCry is a self-propagating malware attack, which targets computers running the Microsoft Windows operating system by encrypting data and demanding ransom in Bitcoins. It automatically spreads through the network and scans for vulnerable systems, using the EternalBlue exploit to gain access, and the DoublePulsar backdoor tool to install and execute a copy of itself. WannaCry malware has a 'kill-switch' that appears to work like this: part of WannaCry’s infection routine involves sending a request that checks for a web domain. If its request returns showing that the domain is alive or online, it will activate the 'kill-switch', prompting WannaCry to exit the system and no longer proceed with its propagation and encryption routines. Otherwise, if the malicious program can not connect to the domain, it encrypts the computer's data, then attempts to exploit the vulnerability of \revision{Server Message Block protocol} to spread out to random computers on the Internet, and laterally to computers on the same network~\citep{enwiki:1086034703}.

\subsection{Epidemiological Models}
\label{sec:em}
Compartmental epidemiological models are used to model the spread of infectious diseases~\citep{brauer2008compartmental, keeling2011modeling}. This approach segments the population into groups (compartments) describing the various stages of infection. The compartmental structure varies according to the disease under study and the application of the model. Following disease evolution, individuals can transition at specific rates among compartments. Generally speaking, these transitions can be either spontaneous (e.g., recovery process) or resulting from interactions (e.g., infection process). In their simplest formulation, compartmental models assume homogeneous mixing. Said differently, each individual is potentially in contact with everyone else~\citep{vespignani2012}. 

The most common compartmental models are the SI, SIS, SIR and SEIR models. \revision{In Appendix~\ref{sec:models} we will briefly review the formulation of these models by neglecting demographic changes in the population (i.e., the number of individuals is assumed to be fixed).} More in detail, we represent them as systems of Ordinary Differential Equations (ODEs). This is a common approach to model epidemics in continuous time, even though it approximates the number of individuals in different compartments as continuous functions.

\subsection{Problem Statement and Threat Model}
The objective of our work is to provide a rigorous mathematical analysis of realistic SPM attacks, and thus lay down the foundation of efficient defense strategies against these prevalent threats. 
Several works propose models to capture the behavior of SPM~\citep{guillen2017study, guillen2018modeling, mishra2007seirs,martinez2021malseirs}, however, the vast majority of them have only theoretical analysis and do not incorporate the information about real-world SPM traces. Thus, they lack validation in real-world scenarios. Additionally, it is hard to perform comparative analysis to other models without presenting their performance using real-world data. Existing work that uses actual malware traces for modeling SPM~\citep{levy2020modeling} leverages minimal epidemiological models that, in their simplicity, fail to fully capture malware characteristics. To this end, here we use a more advanced compartmental model (called SIIDR) to describe epidemics resulting from SPM and apply it to real-world attack traces from a well-known malware, WannaCry.

Besides studying different epidemiological models according to their suitability to describe WannaCry epidemics, our second goal is to infer the parameters of the SIIDR epidemic model for different malware variants. Parameter inference is crucial for enabling attack simulations on real networks to measure the impact of the attack, as well as the effectiveness of defensive measures. Indeed, once the parameters of the attack are known, an analyst could estimate the basic reproduction number of the attack, and  understand whether the attack might result in a macroscopic outbreak. Similarly, a defender might configure its network topology by performing edge or node hardening~\citep{le_2015, tong_2012, torres_2021}, minimizing the leading eigenvalue of the graph to prevent the damage from self-propagating malware attacks, or using anomaly detection methods to detect the malware propagation~\citep{portfiler2021}.

In this work, our focus is on modeling SPM propagation inside a local network (e.g., enterprise network, campus network) since we do not have global visibility on SPM propagation across different networks.
We assume that the attacker gets a foothold inside the local network through a single initially infected host. From the `patient zero' victim, the attack can propagate and infect other vulnerable machines in the subnet. We initially assume a homogeneous mixing model, meaning that every machine can contact all others. This is a valid assumption because in a subnet every machine is able to scan every other internal IP within the same subnet. 
We are assuming that none of the machines is immune to the exploited vulnerability at the beginning of the attack, thus, all of them may become infected during SPM propagation. Infectious machines become recovered when the malware is successfully detected and an efficient recovery process removes it. We assume that these machines cannot be reinfected again.
We then relax the homogeneous mixing assumption and characterize the behavior of the model on arbitrary graph, 
considering that a contact between any two nodes in a network does not occur randomly with equal probabilities, but each node communicates with the particular subset of nodes in the network.

\subsection{Related Work}
\label{sec:relatedwork}
\revision{Numerous works propose to simulate and model malware propagation on different levels of fidelity and scalability~\citep{perumalla2004high}. The research on modeling malware and worms propagation includes hardware testbeds~\citep{vahdat2002scalability, white2002integrated}, emulation systems~\citep{durst1999testing, wei2010tools}, packet-level simulations~\citep{riley2004federated, szymanski2003parallel}, fully-virtualized environments~\citep{perumalla2004high}, mixed abstraction simulations~\citep{guo2000time, kiddle2003hybrid}, and epidemic models. In our work we focus on this last line of research. Similarly,} Mishra and Jha~\citep{mishra2010seiqrs} introduce the SEIQRS (Susceptible – Exposed – Infectious – Quarantined - Recovered - Susceptible) model for viruses and study the effect of the quarantined compartment on the number of recovered nodes. In their paper, the authors focus on the analysis of the threshold that determines the outcome of the disease.
Mishra and Pandey~\citep{mishra2014dynamic} introduce the SEIS-V model for viruses with a vaccinated state, while~\citep{mishra2007seirs} study the SEIRS model to characterize the malicious objects' free equilibrium, formulating the stability of the results in terms of the threshold parameter.
Toutonji et al.~\citep{toutonji2012stability} propose a VEISV (Vulnerable – Exposed – Infectious – Secured – Vulnerable) model and use the reproduction rate to derive global and local stability. With the help of simulation, they show the positive impact of increasing security countermeasures in the vulnerable state on worm-exposed and infectious propagation waves.
Guillen et al.~\citep{guillen2019security} introduce a SCIRAS (Susceptible - Carrier - Infectious - Recovered - Attacked - Susceptible) model. Authors study the local and global stability of its equilibrium points and compute the basic reproductive number.
Ojha et al. ~\citep{ojha2021improved} develop a new SEIQRV (Susceptible - Exposed - Infected - Quarantined - Recovered - Vaccinated) model to capture the behavior of malware attacks in wireless sensor networks. In their work, authors obtain the equilibrium points of the proposed model, analyze the system stability under different conditions, and verify the performance of the model through simulations.
Zheng et al.~\citep{zheng2020seiqr} introduce the SLBQR (Susceptible - Latent - Breaking out - Quarantined - Recovered) model considering vaccination strategies with temporary immunity as well as quarantined strategies. The authors study the stability of the model, investigate a strategy based on quarantines aimed at suppressing the spread of the virus, and discuss the effect of the vaccination on permanent immunity. In order to verify their findings, the authors simulate the model exploring a range of temporary immune times and quarantine rates.

Recently, several attempts have been made to enhance the realism of the epidemic models. For instance, Guillen et al.~\citep{guillen2017study} study the SEIRS model with an improved incidence rate (i.e., new infected hosts per time unit). Additionally, the equilibrium points are computed and their local and global stability are studied. Finally, the authors derive the explicit expression of the basic reproductive number and propose efficient measures to control the epidemics.
Martinez et al.~\citep{martinez2021malseirs} introduce a dynamic version of SEIRS. The authors look at the performance of the model with different sets of parameters, propose optimal values, and discuss its applicability to model real-world malware. 
Gan et al.~\citep{gan2020dynamical} propose a dynamical SIP (Susceptible - Infected - Protected) model, find an equilibrium point, and discuss its local and global stability. Additionally, the authors perform the numerical simulations of the model to demonstrate the dependency on parameter values.
Yao et al.~\citep{yao2018epidemic} present a time-delayed worm propagation model with variable infection rate. They analyze the stability of equilibrium and the threshold of Hopf bifurcation. The authors carry out the numerical analysis and simulation of the model.

Some papers explore malware propagation on networks comprised of different types of devices.
For instance, Guillen et al.~\citep{guillen2018modeling} considers the special class of carrier devices whose operative systems are not targeted by malware (for example, iOS devices for Android malware); the authors introduce a new compartment (Carrier) to account for these devices, and analyze efficient control measures based on the basic reproductive number.
Zhu et al.~\citep{zhu2012modeling} take into consideration the ability of viruses to infect not only computers, but also many kinds of external removable devices; in their model, internal devices can be in Susceptible, Infected, and Recovered states, while removable devices can be in Susceptible and Infected states.

None of these previous works perform model fitting to real-world malware scenarios, but only consider theoretical analyses of the proposed models. The closest to our work is Levy et al.~\citep{levy2020modeling}; the authors use real traces to fit malware propagation with SIR, a simplistic model that, as we have shown, performs poorly compared to SIIDR and fails to capture self-propagating malware dynamics.

\section{Analysis of the SIIDR model}
\label{sec:methodology_ode}
In this section, we introduce the main characteristics of WannaCry propagation dynamics, the proposed modeling framework (the SIIDR model), we discuss its basic reproduction number and the stability of disease-free equilibrium points. Table~\ref{tab:terms_analysis} includes common terminology used in this section.

\subsection{SPM Modeling with the SIIDR model}
A detailed analysis of the WannaCry traces~\citep{chernikova2022cyber} revealed the following characteristics:
\begin{itemize}
    \item The time interval $\Delta t$ between two consecutive malicious attempts from the same infected IP is not constant and has high variability. This intuition is supported by the results in Figure~\ref{fig:QCoD} where we show the quartile coefficient of dispersion (QCoD) of these $\Delta t$ for different Wannacry variants. The $QCoD$ is defined as $(Q_3 - Q_1) / (Q_3 + Q_1)$. As benchmark we show the hypothetical QCoD of exponentially distributed $\Delta_t$ with the same mean observed in the data. We chose the exponential distribution since time intervals lapsing between Poisson-like events happening at constant rate follow this distribution. From the figure we see that the QCoD of $\Delta t$ obtained from the data is much higher ($\sim 50\%$ more across variants) than the one we would expect to see with constant frequency events.
    \item The time interval $\Delta t$ between the last attack from an infected IP and the end of the collected trace is large. The average values of $\Delta t$ between two consecutive malicious attempts and $\Delta t$ between the last attack attempt from an infected IP and the end of the epidemics are shown in Figure~\ref{fig:deltas}. The mean value of the $\Delta t$ in the second case are much larger then the $\Delta t$ between two consecutive attack attempts.
\end{itemize}

\begin{figure}[ht]
\centering
\includegraphics[width=0.4\textwidth]{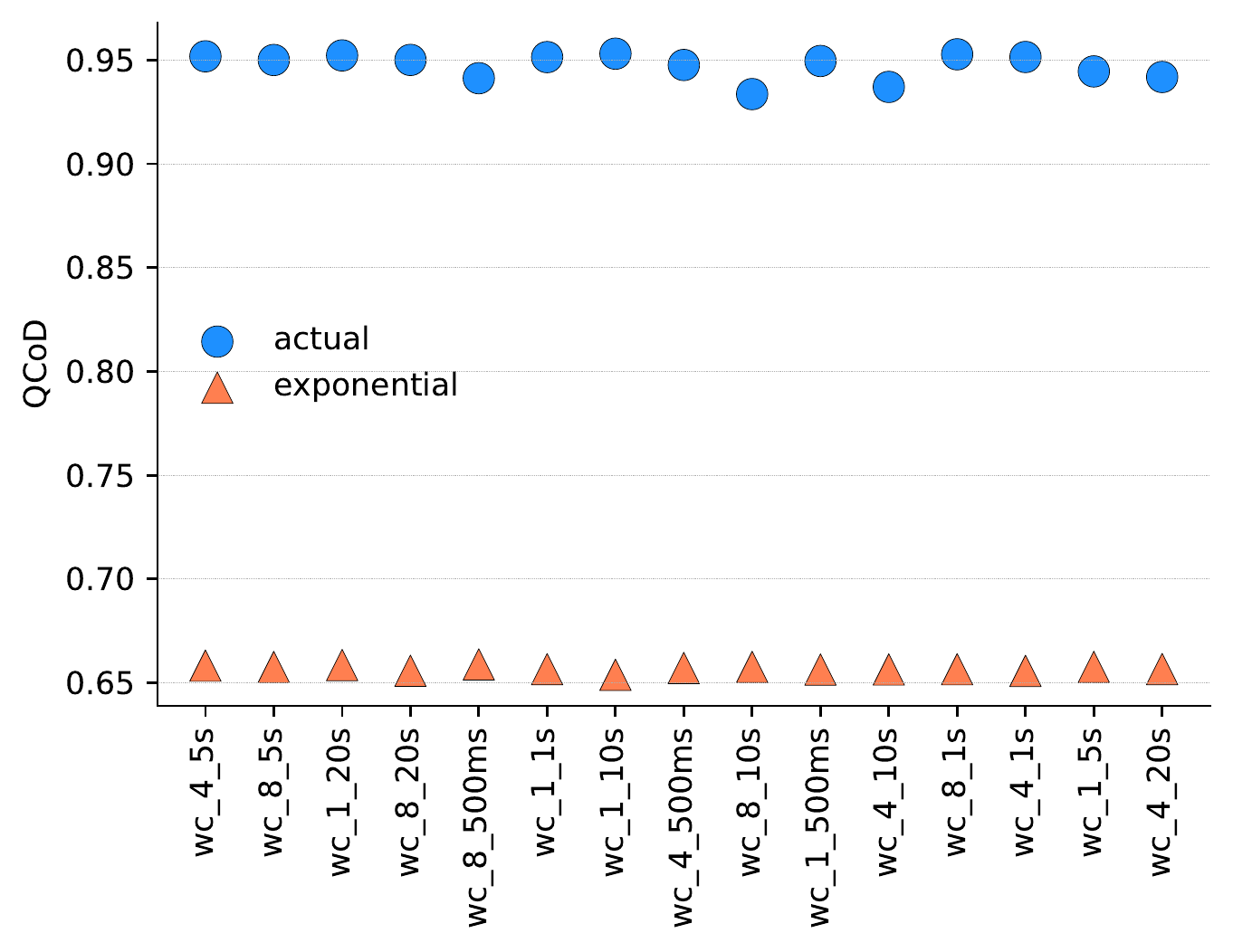}
\caption{Quartile coefficient of dispersion of $\Delta t$ between two consecutive malicious attempts from the same infected IP and of exponential distribution with same mean for different WannaCry variants.}
\label{fig:QCoD}
\end{figure}

\begin{figure}[ht]
\centering
\includegraphics[width=0.4\textwidth]{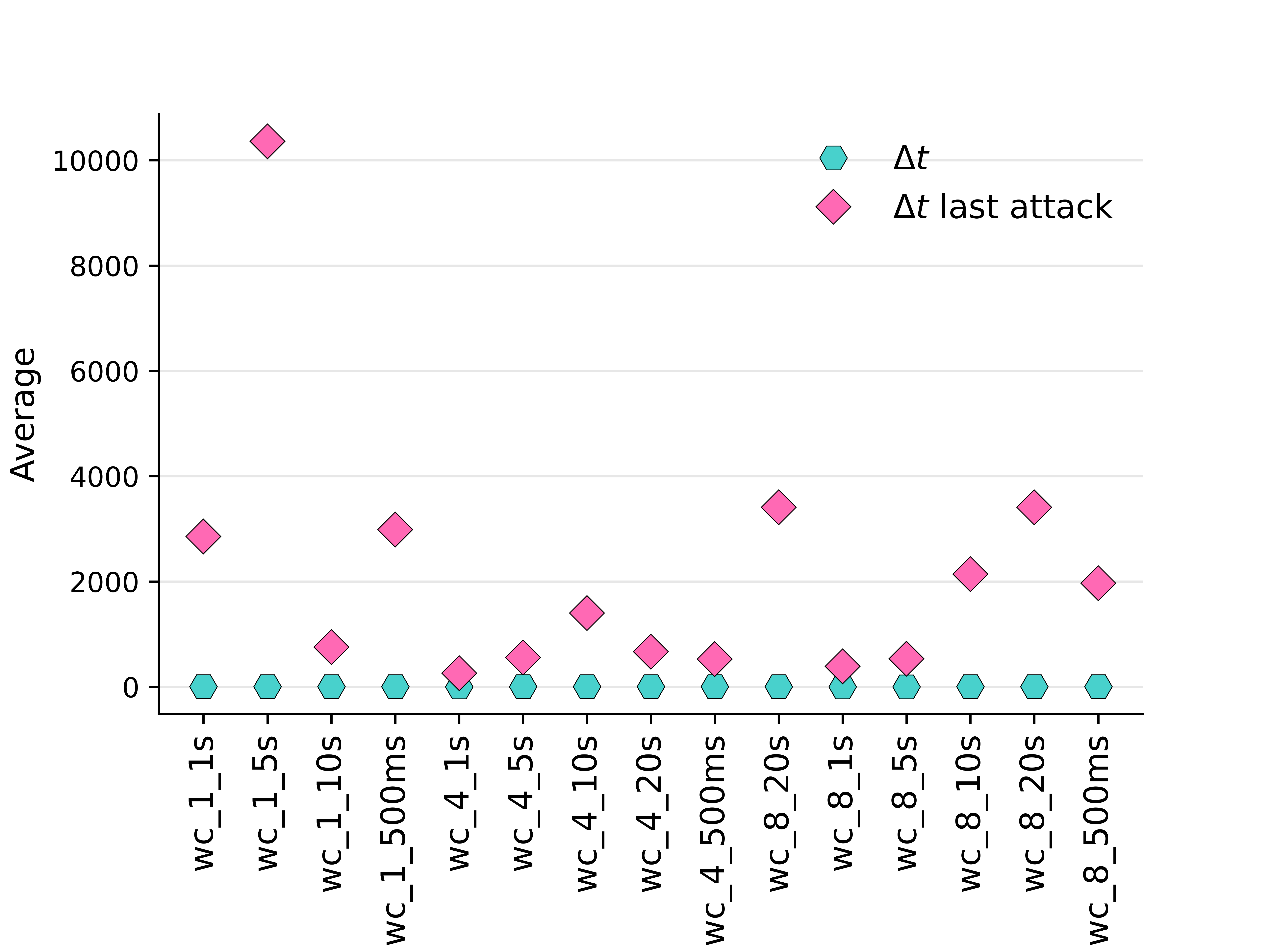}
\caption{Average $\Delta t$ between two consecutive malicious attempts from the same infected IP and Average $\Delta t$ from the last attack attempt to the end of epidemics for different WannaCry variants.}
\label{fig:deltas}
\end{figure}

Based on the first observation, an infected dormant state $I_D$ is included to capture the heterogeneous distribution of time windows between two malicious attack attempts. Therefore, an infected node can become dormant for some period of time and resume its malicious activity later. The second observation supports the presence of a Recovered state: once nodes recover, they will not become infectious or susceptible again, at least within a certain observation period.
The transition diagram corresponding to the SIIDR model is illustrated in Figure~\ref{fig:SIIDRstates}. Interacting with the infectious, a susceptible node can become infected with rate $\beta$, and afterwards, it may either recover with rate $\mu$, or move to the dormant state with rate $\gamma_1$. 
From the dormant state, it may become actively infectious again with rate $\gamma_2$.

\begin{figure}[ht]
\centering
\includegraphics[width=0.3\textwidth]{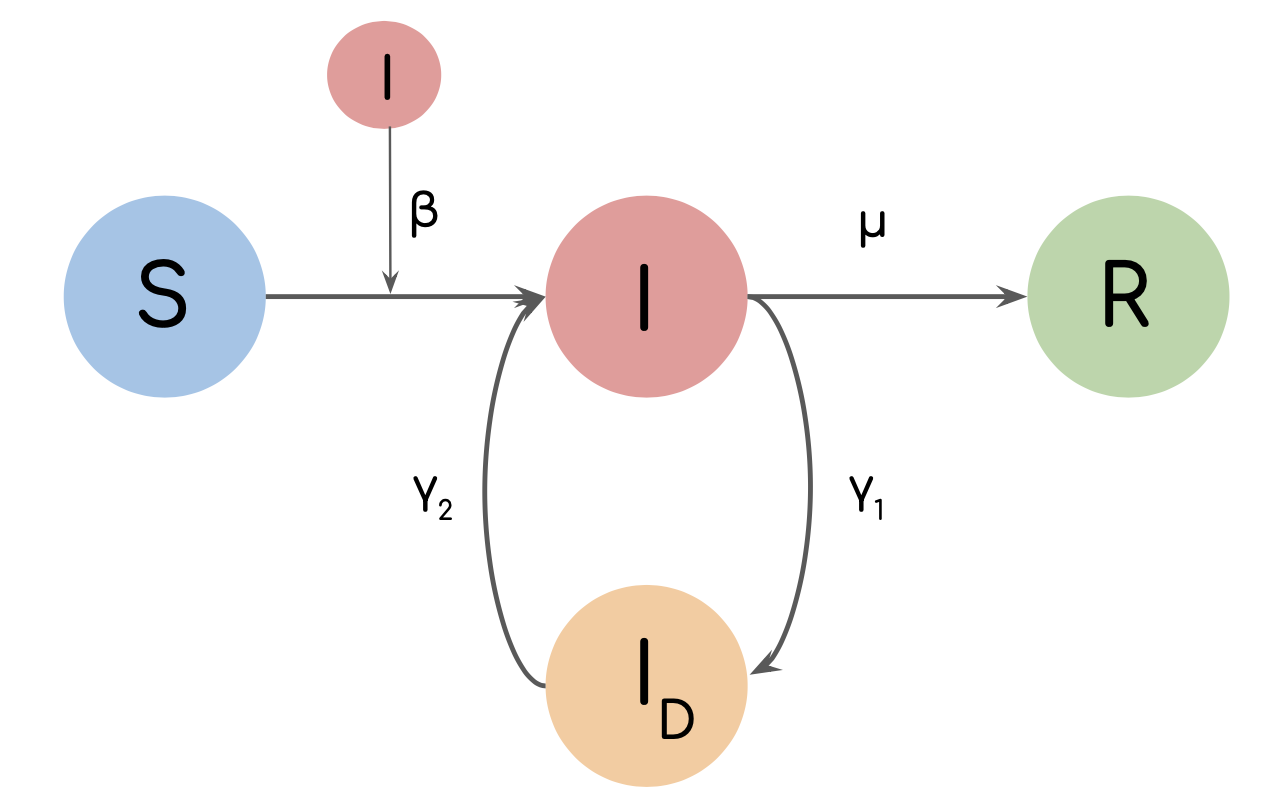}
\caption{Schematic representation of the SIIDR model.}
\label{fig:SIIDRstates}
\end{figure}

The evolution of the system can be modeled through the following ODEs system:
\begin{equation}
\begin{aligned}
    \frac{dS}{dt} &= -\beta S \frac{I}{N}\\
    \frac{dI}{dt} &= \beta S \frac{I}{N} - \mu I - \gamma_1 I + \gamma_2 I_D \\
    \frac{dI_D}{dt} &= \gamma_1 I - \gamma_2 I_D \\
    \frac{dR}{dt} &= \mu I \label{eq:ode}
\end{aligned}
\end{equation}
\noindent with $N=S(t)+I(t)+I_D(t)+R(t)$, where the total size of the population $N$ is constant. It is important to stress how the system of ODEs assumes an homogeneous mixing in the host population. 

\subsection{SIIDR Equilibrium Points}
\label{sec:eqpoints}
While modeling SPM we are interested in equilibrium states when the number of infected individuals equals to 0 and does not change over time (i.e., disease-free equilibrium points). Thus, we need to derive the constant solutions of the ODE system corresponding to SIIDR model~\citep{perko2013differential}. 

\begin{definition} An \textbf{equilibrium point} or \textbf{fixed point} of the system of ODEs $\dot{x} = f(X)$ is a solution $E^*$ that does not change with time, i.e., $f(E^*) = 0$.
\end{definition}

For the SIIDR model we can find the equilibrium points by solving the following system:
\begin{equation*}
\begin{aligned}
 -\beta I \frac{S}{N} = 0\\
 \beta I \frac{S}{N} - \mu I - \gamma_1 I + \gamma_2 I_D = 0\\
 \gamma_1 I - \gamma_2 I_D = 0\\
 \mu I  = 0
\end{aligned}
\end{equation*}
\noindent given that $S + I + I_D + R = N$.

Thus, we find disease-free equilibrium points of the SIIDR model as $E^* = (S, 0, 0, R)$ where $I = I_D  = 0$ and $S + R = N$. The particular case is the beginning of the propagation process when the number of recovered individuals is 0: $R = 0$ or $E^* = (N, 0, 0, 0)$. \revision{Therefore, we perform further analyses of SIIDR model based on this equilibrium point. There exists no endemic equilibrium point when $I \neq 0$ for SIIDR model. It is present only when $\mu = 0$ (SIID model) and is equal to $(0, I^*, \frac{\gamma_1 I^*}{\gamma_2}, 0)$.}

\subsection{The Basic Reproduction Number}
The basic reproduction number $R_0$ is the number of secondary cases generated by a single infectious seed in a fully susceptible population~\citep{keeling2011modeling}. $R_0$ defines the epidemic threshold, that is the condition for a macroscopic outbreak. If $R_0 > 1$, on average, infected individuals are able to sustain the spreading. If $R_0<1$, on average, the disease will die out before any macroscopic outbreak. 

One way to derive the basic reproduction number is to use the next-generation matrix approach~\citep{diekmann1990definition, diekmann2010construction,blackwood2018introduction}. This states that the basic reproduction number is the largest eigenvalue of the next-generation matrix. The method takes into consideration the dynamics of compartments linked to new infections. For example the number of infected individuals in compartment $i$, $i \in \{1, \dots, k\}$, where $k$ is the number of compartments with infected individuals, changes as follows:
\begin{equation*}
\begin{aligned}
\frac{df_i(X)}{dt} = F_i(X) - V_i(X)
\end{aligned}
\end{equation*}
\noindent where $F_i(X)$ is the rate of appearance of new infections in compartment $i$ by all other means, $V_i(X) = [V_i^-(X) - V_i^+(X)]$, $V_i^+(X)$ is the rate of transfer of individuals into compartment $i$ and $V_i^-(X)$ represents the rate of transfer of individuals out of compartment.
If $E^*$ is a disease-free equilibrium, then we can define a next-generation matrix:
\begin{align*}
G = FV^{-1}
\end{align*}
\noindent where:
\begin{align*}
F= \frac{\partial F_i}{\partial x_j}(E^*)\\
V=\frac{\partial V_i}{\partial x_j}(E^*)
\end{align*}
In the case of SIIDR model, the matrix $G$ can be represented at one of the disease-free equilibrium points $DFE=(N,0,0,0)$ as follows:
\begin{align}
G = \begin{bmatrix}
\beta & 0\\
0 & 0
\end{bmatrix} \begin{bmatrix}
\mu + \gamma_1 & -\gamma_2 \\
 -\gamma_1 & \gamma_2
\end{bmatrix}^{-1}=
\begin{bmatrix}
\frac{\beta}{\mu} & \frac{\beta}{\mu}\\
0 & 0
\end{bmatrix}
\label{eq:gmat}
\end{align}

Let $\Vec{v}$  be an eigenvector of the matrix $G$, and $\lambda$ its corresponding eigenvalue. The eigenvalue equation is~\citep{bhatia97}:
\begin{equation*}
    \begin{aligned}
G\Vec{v} = \lambda \Vec{v},
    \end{aligned}
\end{equation*}
\noindent where $\Vec{v}$ is a nonzero vector, therefore $det[\lambda I - G] = 0$. Using $G$ from equation~\ref{eq:gmat}, we obtain: 
\begin{equation*}
    \begin{aligned}
det[\lambda I - G] = \lambda\left(\lambda - \frac{\beta}{\mu}\right) = 0,
\end{aligned}
\end{equation*}

\noindent which results in: 1) $\lambda = 0$ or 2) $\lambda = \beta / \mu$. 
According to the next-generation matrix method~\citep{diekmann1990definition, diekmann2010construction,blackwood2018introduction}, the reproduction number $R_0$ is the largest eigenvalue of the next-generation matrix $G$, hence, $R_0 = \frac{\beta}{\mu}$, which is the same definition of $R_0$ of the SIR model. In other words, the introduction of the new compartment $I_D$ does not alter the conditions for a macroscopic outbreak. 
We note that, in general, the disease free equilibrium might contain individuals already immune to the disease, i.e., $E^* =(N-R,0,0,R)$. This might be due to wave of infections caused by previous introductions of the virus. In this more general case we have: $R_0 = \frac{\beta}{\mu}\left (1-\frac{R}{N} \right )$, where in parenthesis we have the fraction of the susceptible population. 
\begin{table*}[!htpb]
\caption{Terminology and abbreviations used for SIIDR analysis.}
\begin{tabular}{|l||l|}
\hline
\textbf{Notation} & \textbf{Meaning}\\ 
\hline
$\beta$ & Infection Rate\\
$\tilde{\beta}$& Infection Probability\\
$\mu$ & Recovery Rate\\
$\gamma_1$ & Transition Rate from Infected to Infected Dormant Compartment\\
$\gamma_2$ &Transition Rate from Infected Dormant to Infected Compartment\\
$\zeta_{i, t}(I)$&The Probability of Node $i$ of not Getting Infected at Time Step $t$\\
$\alpha_{XY}$&The Probability of a Node to Transition from State $X$ to $Y$\\
$\alpha_{XX}$&The Probability of a Node to stay in the State $X$\\
$DFE$&Disease-free equilibrium point\\
$R_0$ & The Basic Reproduction Number\\
$G$& Next-generation Matrix\\
$X$& The Vector of Individual Numbers in Each Compartment\\
$E^*$& Equilibrium Point for SIIDR as the System of ODEs\\
$L$&Lyapunov Function\\
$P$&Each Node's Vector of Probabilities to be in Each Compartment\\
$P^*$& Equilibrium Point for SIIDR as the NLDS\\
$g$& Matrix Form of SIIDR Represented as the NLDS\\
$C$&Linear Part of SIIDR as the NLDS Matrix Form\\
$B$& Non-Linear Part of SIIDR as the NLDS Matrix Form\\
$\mathds{A}$&Graph Adjacency Matrix\\
$\lambda_A$&The Largest Eigenvalue of the Adjacency Matrix $\mathds{A}$\\
$d$&Degree of the $d$-regular Graph\\
\hline
\end{tabular}
\label{tab:terms_analysis}
\end{table*}

\subsection{Stability Analysis of SIIDR Equilibrium Points}
\label{sec:oed_stab}
A particularly important characteristic of a disease-free equilibrium point is its stability~\citep{hirsch1974differential}, which indicates whether the system will be able to return to the equilibrium point after small perturbations. For example, a small perturbation can be a slight increase in the number of initially infected nodes. 

Let us consider the system of ODEs that captures the dynamics of our SIIDR model (see Equations~\ref{eq:ode}), governed by:
\begin{equation*}
    \begin{aligned}
\dot{x} = f(X), X \in R^n
\end{aligned}
\end{equation*}
\noindent Let $X = E^*$ be a fixed point of $f(X)$, that is, $f(E^*) = 0$. Furthermore, let us assume that the system's initial state at $t = 0$ is $X = X^0$. In this context, the stability of $E^*$ can be obtained answering to the following question: if the system starts near $E^*$, how close will it remain to $E^*$? Beside this intuition, stability is more formally defined as follows~\citep{hirsch1974differential}:

\begin{definition} The equilibrium point $E^*$ is \textbf{stable} if for any $\epsilon > 0$, there exists a $\delta > 0$ such that: if the system's initial state $X^0$ lies in the ball of radius $\delta$ around $E^*$ (i.e., $||X^0 - E^*|| < \delta$), then solutions $X^t$ exist for all $t > 0$, and they stay in the ball of radius $\epsilon$ around $E^*$ (i.e., $||X^t - E^*|| < \epsilon$).
\end{definition}

In addition: 
\begin{definition} We say that $E^*$ is \textbf{locally asymptotically stable} if it is stable and the solutions $X^t$ with initial state $X^0$ in the ball of radius $\delta$ converge to $E^*$ as $t \rightarrow \infty$.
\end{definition}

And:
\begin{definition}
We say that  $E^*$ is \textbf{stable in the sense of Lyapunov (i.e., Lyapunov stable)} when there exists the continuously differentiable function $L(X)$ such that:
\begin{align}
L(X) \geq 0, L(E^*) = 0 \label{eq:lyapunov1}\\
\dot{L}(X) = \frac{d}{dt}L(X)= \sum_{i}\frac{dL}{dx_i}f_i(X)\leq 0, \dot{L}(E^*) = 0
\label{eq:lyapunov2}
\end{align}
If $\dot{L}(X) <0$ and $\dot{L}(X) = 0$ only when $X=E^*$, then $E^*$ is \textbf{locally asymptotically stable}. 
\label{def:lyapunov}
\end{definition}

We next analyze the stability of the SIIDR disease-free equilibrium points and show that they are Lyapunov stable, if the reproduction number $R_0$ is smaller or equal to one. We formally state and prove it in the following theorem:
\begin{theorem} If $R_0 \leq 1$ the disease-free equilibrium point $E^*$ of the SIIDR system of ODEs is Lyapunov stable.
\end{theorem}
\begin{proof} 
Let $L(X) = I + I_D$, where $L$ is the valid Lyapunov function as long as it is non-negative continuously differentiable scalar function which equals 0 at the disease-free equilibrium point ($I = I_D = 0$). The time-derivative of $L$ is the following:
\begin{align*}
\dot{L} = \frac{dL}{dt} = \frac{d(I + I_D)}{dt} = \beta S \frac{I}{N} - \mu I, 
\end{align*}
where we used Equations~\ref{eq:ode} that describe the evolution of $I$ and $I_D$.
Therefore, $\dot{L} \leq 0$ (Equation~\ref{eq:lyapunov2}) when:
\begin{align*}
I\left(\frac{\beta S} {\mu N} - 1\right) \leq 0 
\end{align*}
Given the basic reproduction number $R_0 = \frac{\beta S} {\mu N} $, we obtain:
\begin{align}
I\left(R_0 - 1\right) \leq 0
\label{eq:th1x}
\end{align}
\noindent Equation~\ref{eq:th1x} holds when $R_0 \leq 1$. Hence, $\dot{L} \leq 0$ when $R_0 \leq 1$.
Furthermore, $\dot{L}(E^*) = 0$ (since $I = 0$ when $X = E^*$), which concludes the proof that $E^*$ is a Lyapunov stable disease-free equilibrium point.

Note that $\dot{L}(X) = 0$ when $I = 0$, even if $X \neq E^*$ (for instance, if $I_D \neq 0$). Thus, $E^*$ is not locally asymptotically stable (see Definition~\ref{def:lyapunov}).
\end{proof}

\subsection{SIIDR Analysis on Arbitrary Graphs}
Our analysis in previous sections was performed under the homogeneous-mixing assumption~\citep{bansal2007,vespignani2012}. In this limit, all hosts are well-mixed and potentially in contact. The homogeneous approximation might be a good representation of the contact dynamics in a local subnet where each machine can contact anyone else. However, the contact patterns in larger networks are complex. 
Indeed, many real networks (including the Internet) feature, among other properties, a heterogeneous connectivity distribution consisting of a few highly-connected 'hubs', while the vast majority of nodes have much lower connectivity~\citep{Albert:2002:rmp,Pastor-Satorras2015}. 
In this section, we analyze the epidemiological dynamics of the SIIDR model on arbitrary graphs that capture heterogeneity in host contact patterns. In this case, the propagation of malware can be modeled with a discrete-time Non-Linear Dynamical System~\citep{Chakrabarti2008,Prakash2011}.

A NLDS system is specified by the vector of probabilities at time step $t+1$ as $P_{t+1} = g(P_{t})$, where $g$ is non-linear continuous function operating on a vector $P_{t}$. We define the system equations based on the transition diagram of the model (Figure~\ref{fig:SIIDRstates}).

First, we are computing the probability of node $i$ of \emph{not getting infected} at time step $t$: $\zeta_{i, t}(I)$, which happens when: 1) none of its neighbors are in state $I$, or 2) a neighbor is in state $I$ but fails to infect $i$ with probability $(1-\tilde{\beta})$, where $\tilde{\beta}$ is the attack transmission probability over a contact-link. We note how $\tilde{\beta}$ is generally different than the infection rate $\beta$ introduced above. Indeed we can approximate $ \beta= \tilde{\beta} \langle k \rangle_t$ where $ \langle k \rangle_t$ is the average contact rate per unit time. Hence:
\begin{equation}
\begin{aligned}
    \zeta_{i, t}(I) = \prod\limits_{j \in \texttt{Neigh}(i)} [(1-P_{I,j,t}) +  P_{I,j,t} \cdot (1-\tilde{\beta})]
    = \prod\limits_{j \in {1..N}} (1 - \tilde{\beta} A_{i,j} P_{I, j, t})
\end{aligned}
\end{equation}

Next, we develop the equations for probabilities $P$ of node $i$ to be in each of the possible states ($S, I, I_D, R$) at time step $t+1$. 

For generality and clarity, we denote by $\alpha_{XY}$ the probability of a node to transition from state $X$ to $Y$, while $\alpha_{XX}$ is the probability of a node to remain in state $X$. With this notation, the probability equations for each state are as follows:

\noindent \textit{State $S$:} A node $i$ is in state $S$ at time $t+1$ if it was in state $S$ at time $t$ and it did not get infected:
\begin{equation}
    P_{S,i,t+1} = P_{S,i,t} \cdot \zeta_{i, t}(I)
\label{eq:ndls_S}
\end{equation}

\noindent \textit{State $I$:} A node $i$ is in state $I$ at time $t+1$ if either: 1) it was in state $S$ at time $t$ and was successfully infected, or 2) it was in state $I$ at time $t$ and it remained there (i.e., it did not transition to states $R$ or $I_D$), or 3) it was in state $I_D$ at time $t$ and transitioned to state $I$.
\begin{equation}
\begin{split}
    P_{I,i,t+1} &= P_{S,i,t} \cdot (1-\zeta_{i, t}(I)) + P_{I,i,t} \cdot \alpha_{II} + P_{I_D,i,t} \cdot \alpha_{I_DI}
\end{split}
\label{eq:ndls_I}
\end{equation}

\noindent \textit{State $I_D$:} A node $i$ is in state $I_D$ at time $t+1$ if either: 1) it was in state $I$ at time $t$ and transitioned to state $I_D$, or 2) it was in state $I_D$ at time $t$ and it remained there.
\begin{equation}
    P_{I_D,i,t+1} = P_{I,i,t} \cdot \alpha_{II_D} + P_{I_D,i,t} \cdot \alpha_{I_DI_D}
\label{eq:ndls_ID}
\end{equation}

\noindent \textit{State $R$:} We can compute $P_{R,i,t}$ using the relation: 
\begin{equation}
\forall{i,t}:  P_{S,i,t}+P_{I,i,t}+P_{I_D,i,t}+P_{R,i,t}=1
\label{eq:ndls_R}
\end{equation}

Now we can write down the system of equations for SIIDR using Equations~\ref{eq:ndls_S}--\ref{eq:ndls_R} to define $P_{t}$, the probability vector that completely describes the evolution of the system at any time step $t$:
\begin{equation}
\begin{aligned}
    P_{S,i,t+1} &= P_{S,i,t} \cdot \zeta_{i, t}(I)\\
    P_{I,i,t+1} &= P_{S,i,t} \cdot (1-\zeta_{i, t}(I)) + P_{I,i,t} \cdot \alpha_{II} + P_{I_D,i,t} \cdot \alpha_{I_DI}\\
    P_{I_D,i,t+1} &= P_{I,i,t} \cdot \alpha_{II_D} + P_{I_D,i,t} \cdot \alpha_{I_DI_D}\\
    P_{R,i,t+1} &= 1 - P_{S,i,t} - P_{I_D,i,t} \cdot (\alpha_{I_DI} + \alpha_{I_DI_D}) - P_{I,i,t} \cdot  (\alpha_{II} + \alpha_{II_D})
\end{aligned}
\label{eq:sys_nlds}
\end{equation}

\subsubsection{Stability Analysis}
\label{sec:ndls_stability}
The next step in our analysis of the SIIDR propagation on complex networks represented as arbitrary graphs is to define the disease-free equilibrium points and analyze their stability. 

\begin{definition} An \textbf{equilibrium point} of NLDS is the probability vector $P^*$ that satisfies $P_{t+1}$ = $P_t = P^*$ for any $t$~\citep{verhulst2006nonlinear}.
\end{definition}
Thus, for the SIIDR model we can define the disease-free equilibrium point as follows: 
\begin{align*}
P^* = [P_S, 0, 0, P_R]^T,\text{ where }P_R + P_S = 1
\end{align*}

One way to analyze the stability of the equilibrium point of a non-linear dynamical system is to approximate its dynamics at that point as a linear dynamical system (i.e., linearization)~\citep{sayama2015}. In this case, the system behavior in an infinitesimally small area about the equilibrium point is approximated with a Jacobian matrix.

The largest eigenvalue $\lambda_J$ of the Jacobian matrix indicates whether the equilibrium point of the system is stable or not. Since we are considering the time as discrete, if $|\lambda_J| < 1$, the equilibrium point is asymptotically stable; even if small perturbations occur, the system asymptotically goes back to the equilibrium point. If $|\lambda_J| > 1$, the system is unstable and diverges away from the equilibrium point. If $|\lambda_J| = 1$, then the system may either diverge from, or converge to the equilibrium point~\citep{bof2018lyapunov, dahleh2004lectures,Haddad2011,sayama2015}.

The Jacobian matrix of SIIDR modeled as NDLS and an analysis of its eigenvalues is presented in Appendix~\ref{sec:appB}. We show that one of the eigenvalues of the Jacobian has value 1. This result is particularly significant. Asymptotic stability requires all the eigenvalues of the Jacobian matrix to be less than 1 in absolute values. Since the Jacobian matrix has at least one eigenvalue of value 1, the equilibrium point of the NLDS system cannot be asymptotically stable. However, the equilibrium point can still be Lyapunov stable.

We show that the equilibrium points of SIIDR are indeed Lyapunov stable using Lyapunov's second stability criterion.

\begin{definition}
The equilibrium point $P^*$ of  $P_{t+1} = g(P_{t})$ NLDS is \textbf{Lyapunov stable} if there exists a continuous function $L$, such that for any $t$:
\begin{align}
L(P) > 0, L(P^*)  = 0 \\
L(P_{t+1}) - L(P_t) \leq 0 \label{eq:diff}
\end{align}
\end{definition}

\begin{theorem}
The equilibrium points of SIIDR represented as NLDS of the form~(\ref{eq:sys_nlds}) are Lyapunov stable if:
\begin{align}
\lambda_A \frac{\tilde{\beta} }{\mu} \leq 1\label{eq:thres}
\end{align}
where $\lambda_A$ is the largest eigenvalue of the adjacency matrix, $\tilde{\beta}$ and $\mu$ are probabilities of infection and recovery respectively. 
\end{theorem}
The proof of Theorem 3 is presented in Appendix~\ref{sec:app2}.

\section{Experimental Results}
\label{sec:experiments}
In this section, we present the reconstruction of WannaCry dynamics from network logs captured with Zeek monitoring tool~\citep{zeek}. Additionally, we show supporting results that confirm that the SIIDR model fits WannaCry traces best. We also present our experiments for parameter estimation, providing the statistics from the posterior distribution of SIIDR transition rates. These results expand the results presented in our previous work where we introduced SIIDR model~\citep{chernikova2022cyber}.
\revision{Moreover, we study the basic reproduction number $R_0$ of the reconstructed attacks to understand its correlation with SPM dynamics (in particular, its propagation speed). We also discuss the  issue of structural and practical identifiabiility of SIIDR parameters which is common in epidemiological modeling. Finally, we experimentally demonstrate that the condition for Lyapunov stability of the disease-free equilibrium point holds when the networks are modeled as arbitrary graphs relaxing homogeneous mixing assumption.}

\subsection{WannaCry Malware Traces}
We obtained realistic WannaCry attack traces by running the malware  in a controlled virtual environment consisting of 51 virtual machines, configured with a version of Windows vulnerable to the EternalBlue SMB exploit. The external traffic generated by the VMs was blocked to isolate the environment and prevent external malware spread. The infection started from an initial victim IP,  and then the attack propagated through the network as the infected IPs began to scan other IPs. In these experiments, WannaCry varied the number of threads used for scanning, which were set to 1, 4 or 8, and the time interval between scans, which was set to 500ms, 1s, 5s, 10s or 20s. Using the combination of these two parameters resulted in 15 WannaCry traces. While running WannaCry with this setup, the log traces were collected with the help of the open source Zeek network monitoring tool. 

\subsection{WannaCry Reconstruction}
To reconstruct the WannaCry dynamics we are using Zeek communication logs where we consider only communication between internal IPs. Since WannaCry attempts to exploit the SMB vulnerability, we label as malicious all the attempts of connections on destination port 445. The first attempt to establish the malicious connection is considered to be the start of the epidemics, and the end corresponds to the last communication event in the network. Each IP trying to establish a malicious connection for the first time at time $t$ is considered infected at time $t$. The cumulative number of infected IPs through time represents the curve of the WannaCry epidemics.

\subsection{WannaCry dynamics}
We show the dynamics of WannaCry variants characterized by different numbers of scanning threads and time between scans in Figure~\ref{fig:wc_reconstruct}. These dynamics represent the cumulative number of infected nodes during the epidemic time. The trace which corresponds to 1 thread and 20s sleeping time wc\_1\_20s has unusual behavior in the dynamics. It has a very small number of infected nodes until the end of the attack, when the infections rapidly increase to the 7 infected nodes at once.  
For all other WannaCry variants we observe that the attack reaches the maximum number of infected nodes quickly and is not able to infect any other nodes for a large time window before the end of the epidemic. These graphs confirm the fact that after an IP enters a recovered state it no longer has an opportunity to get back to susceptible or infected nodes. For modeling and parameter estimation experiments we exclude the time windows after which the number of infections does not change. Additionally, we present the number of contacted and infected IPs in Table~\ref{tab:wc_reconstruct_nums}. Interestingly, the overall percentage of infected nodes is small (around 25\% on average) for all variants. \revision{The possible reason for this is the fact that some of the machines that do not get infected may have immunity to the malware.}

\begin{figure*}[ht]
\centering
\includegraphics[width=0.95\textwidth]{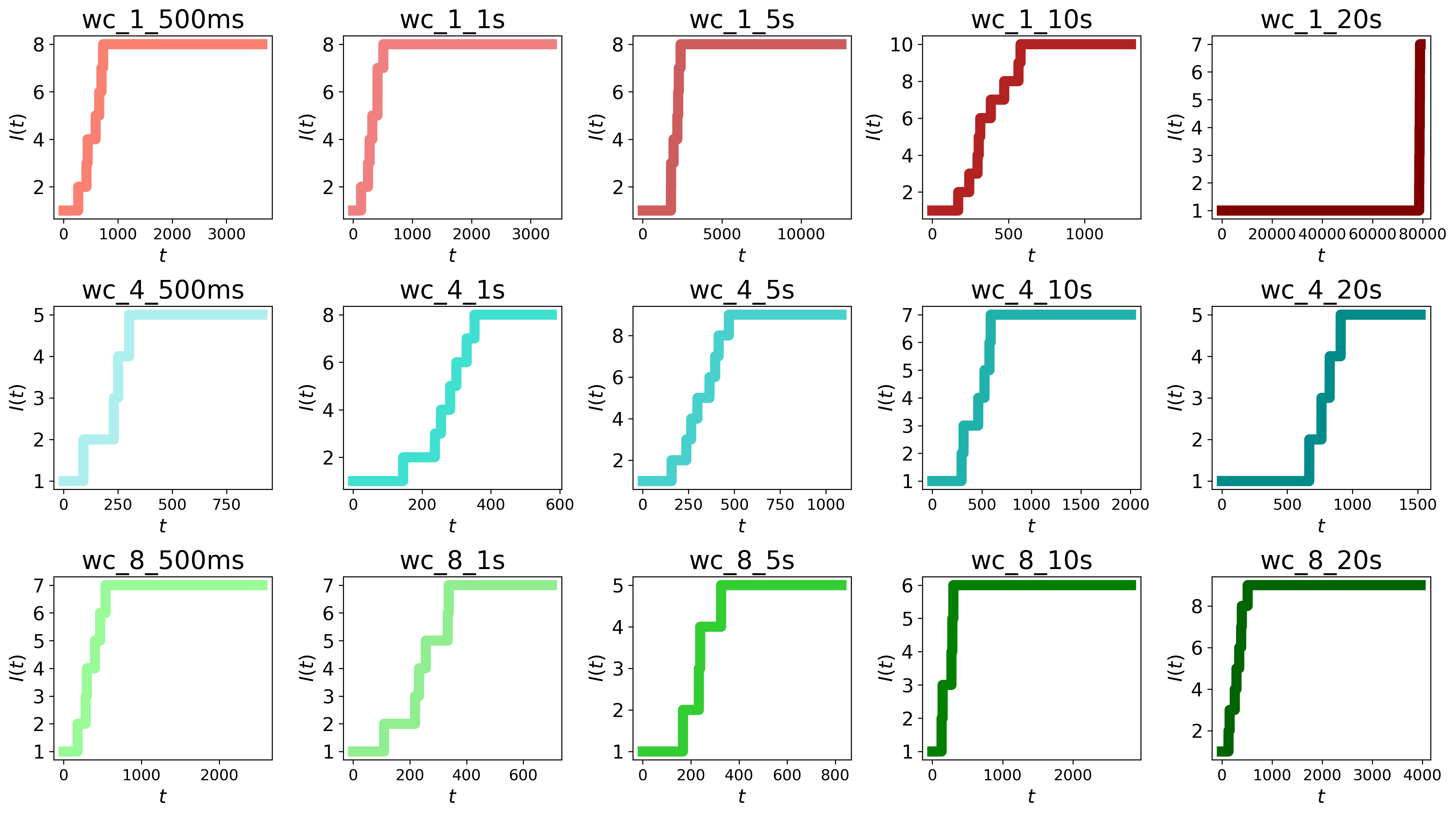}
\caption{The cumulative number of the infected nodes $I(t)$ (counting all nodes in states $I, I_D$ and $R$) at each time point $t$ of WannaCry propagation for different variants of WannaCry. Each WannaCry variant is identified by two parameters: the number of threads used for scanning and the time interval between scans (i.e., wc\_1\_500ms uses 1 thread to scan every 500 ms).}
\label{fig:wc_reconstruct}
\end{figure*}

\begin{table*}[!htpb]
\caption{Number of contacted and Infected IP adresses from communication data for WannaCry modeling.}
\begin{tabular}{|c||c|c|c|}
\hline
\textbf{WannaCry} &  \textbf{\# Contacted IPs} & \textbf{\# Infected IPs} & \textbf{Fraction of}\\ 
\textbf{Variant}&&&\textbf{Infected IPs}\\
\hline
wc\_1\_500s&37&8&0.22\\
wc\_1\_1s&37&8&0.22\\
wc\_1\_5s&37&8&0.22\\
wc\_1\_10s&34&10&0.29\\
wc\_1\_20s&35&7&0.20\\
wc\_4\_500ms&34&5&0.15\\
wc\_4\_1s&35&8&0.23\\
wc\_4\_5s&36&9&0.25\\
wc\_4\_10s&35&7&0.20\\
wc\_4\_20s&35&5&0.14\\
wc\_8\_500ms&35&7&0.20\\
wc\_8\_1s&35&7&0.20\\
wc\_8\_5s&35&5&0.14\\
wc\_8\_10s&36&6&0.17\\
wc\_8\_20s&35&9&0.26\\
\hline
\end{tabular}
\label{tab:wc_reconstruct_nums}
\end{table*}
\subsection{Model Selection}
\revision{We select the model that fits WannaCry traces best among several representative compartmental epidemiological models: SI, SIS, SIR, SEIR and SIIDR assuming an homogenous mixing of machines.} \revision{These models have different number of parameters and, therefore, different a-priori explaining power. The SIIDR model is also the one that has the largest number of parameters. To allow for a fair comparison among models, we considered the Akaike Information Criterion (AIC) as a metric to measure their performance. The AIC is calculated based on the maximum likelihood estimate and the number of free model parameters, thus, allowing comparison of models with different number of parameters.}
More information about AIC criteria can be found in Appendix~\ref{sec:fitting_aic}. 
We perform model selection for all WannaCry traces. The lowest AIC score corresponds to the best model. We run the experiments on an uniform grid of model parameter values between 0 and 1. We select the lowest AIC score for each WannaCry trace and each compartmental model. The results are illustrated in Table~\ref{tab:aic_scores}.
\revision{The SIIDR model has the lowest AIC score for all traces except for wc\_1\_20s. For instance, the AIC score associated with the SEIR model for wc\_8\_5s WannaCry trace is equal to -87, the SIS model score is 104, the SIR model score is -35, whereas for the SIIDR model the AIC is the lowest and has the value of -121. This trend is valid for all other WannaCry traces except for wc\_1\_20s where the SEIR model provides the best fit. However, this variant is an outlier. Therefore, we can conclude that, among the four epidemiological models, the SIIDR model fits the WannaCry attack traces best.} 

For each compartmental model and each WannaCry trace, we plot the reconstruction curve of the number of infected nodes using the parameters corresponding to the lowest AIC score along with the true dynamics of infected nodes.
The results are shown in Figure ~\ref{fig:AIC_all}. 
\revision{In the case of the SIS model, the orange line (representing the simulated dynamics of the number of infected nodes) is far from the blue one, which illustrates the empirical dynamic for all malware traces. In the case of the SIR and SEIR models the numbers of simulated infections are closer to the real ones, however, the SIIDR and actual dynamics curves are the closest.}

\bgroup
\setlength\tabcolsep{3pt}
\begin{table*}[t] 
\caption{\revision{AIC scores for each of the SPM models for different WannaCry variants. Lower is better.}}
\begin{tabular}[t]{|c||c|c|c|c|}
\hline
\textbf{WannaCry} &\bf SIS &\bf SIR&\bf SEIR &\bf SIIDR\\
\textbf{variant} & &&&\\
\hline
wc\_1\_500ms & 143&114&-8 &\textbf{-126}\\
wc\_1\_1s &188 & 145&-10&\textbf{-127}\\
wc\_1\_5s&163& 143&121 &\textbf{72}\\
wc\_1\_10s & 197 &53& 69&\textbf{-92}\\
wc\_1\_20s&559&696&\textbf{-63} &700 \\
wc\_4\_500ms &76&-45& -143&\textbf{-166}\\
wc\_4\_1s &160 & 107 &-17 &\textbf{-55}\\
wc\_4\_5s&186&158&28 &\textbf{-46}\\
\hline
\end{tabular}
\quad
\begin{tabular}[t]{|c||c|c|c|c|}
\hline
\textbf{WannaCry} &\bf SIS &\bf SIR&\bf SEIR &\bf SIIDR\\
\textbf{variant} & &&&\\
\hline
wc\_4\_10s &94&-36&-78&\textbf{-145}\\
wc\_4\_20s &76&11 & -26&\textbf{-117}\\
wc\_8\_500ms&101&18& -120&\textbf{-147}\\
wc\_8\_1s &91& 51&-99&\textbf{-116}\\
wc\_8\_5s &104&-35&-87 &\textbf{-121}\\
wc\_8\_10s  &74&-90&-92 &\textbf{-118}\\
wc\_8\_20s  &164&173&105&\textbf{-89}\\
\hline
\end{tabular}
\label{tab:aic_scores}
\end{table*}
\egroup

\begin{figure*}[!htbp]
\centering
\includegraphics[width=0.95\textwidth]{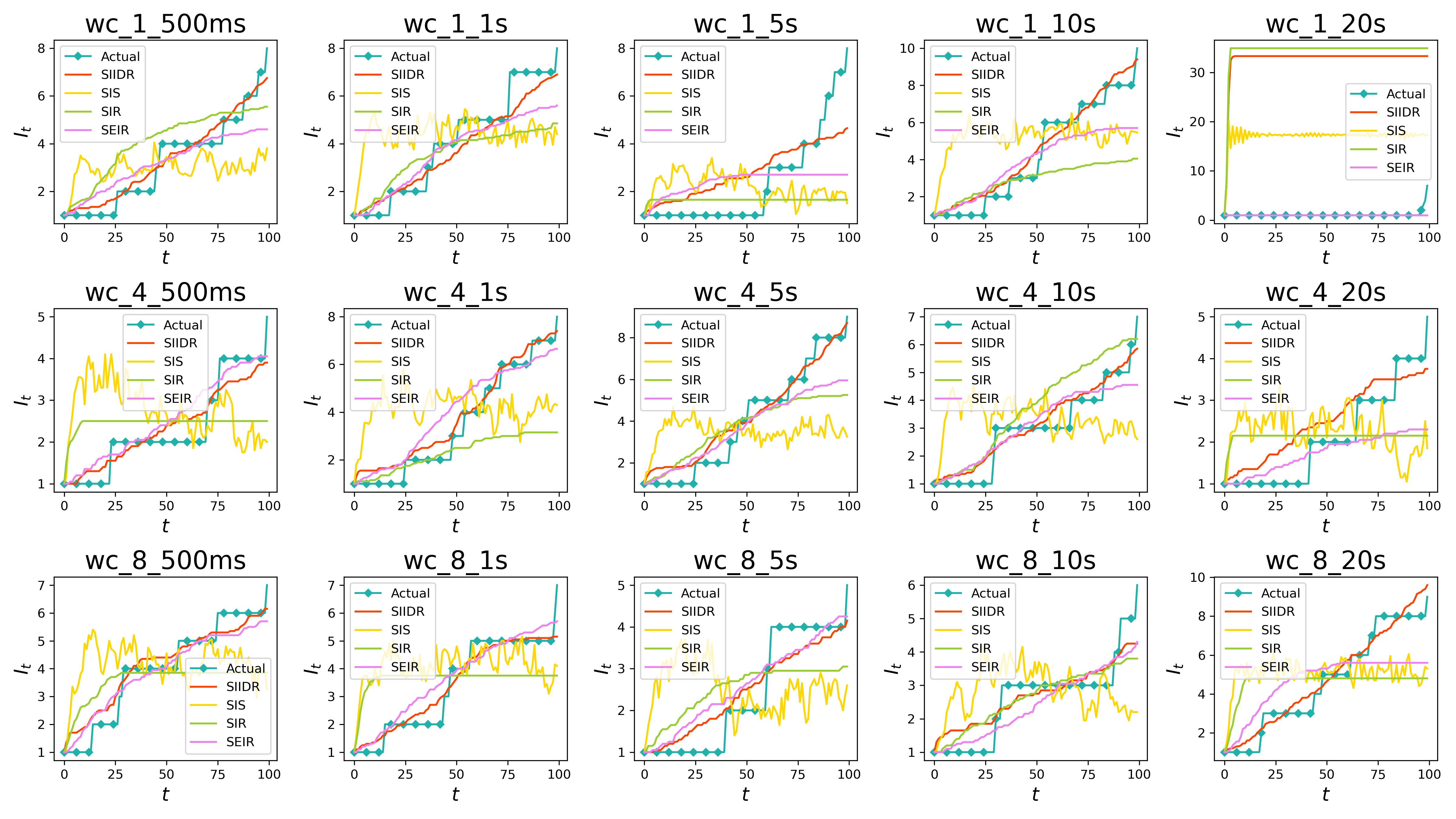}
\caption{Model fitting for different WannaCry variants.}
\label{fig:AIC_all}
\end{figure*}

\subsection{Parameter estimation}
We approximated the posterior distribution of SIIDR transition rates using the ABC-SMC-MNN technique~\citep{filippi2013optimality}. The details of this technique are described in Appendix`~\ref{sec:fitting_smc}.
The mean values and standard deviation of the posterior distribution of SIIDR transition rates ($\beta$, $\mu$, $\gamma_1$, $\gamma_2$) are represented in Table~\ref{tab:siidrparams}. The parameter $dt$ is the integration step, which is calculated as:  $dt = (t_N - t_0) /T$, where $t_N$ is the last timestamp, $t_0$ is the first timestamp, and $T$ is the number of timestamps in WannaCry traces. $dt$ differs by variant due to the different propagation speeds. The attack transmission probability $\tilde{\beta}$ is related to attack transmission rate $\beta$ as follows: $\beta = \tilde{\beta} \langle k \rangle_t$ where $\langle k \rangle_t$ is the average contact rate per unit time. In the WannaCry traces we have one communication or contact per $dt$, hence, the transmission probability $\tilde{\beta}$ over a contact-link also equals $\beta$. 

Based on estimated values of transition rates we calculated the basic reproduction number $R_0$ for all WannaCry traces. We also calculate the SPM propagation speed for all WannaCry traces as the average number of new infections per 100 seconds. The results are illustrated in Figure~\ref{fig:tracesR0}. As expected, we observe that higher SPM propagation speed corresponds to a higher basic reproduction number $R_0$.

The mean values of the parameters' posterior distribution can be further used to simulate SPM with the SIIDR model. This provides an opportunity to create synthetic, but realistic, WannaCry scenarios and evaluate whether existing defenses are successful in preventing and stopping the malware from propagation in the networks. However, we notice that some of the WannaCry attack variants affect only a small number of nodes. For example, the wc\_8\_5s trace has only 4 infected nodes at the end of the trace which constitutes 14\% of all nodes. Consequently, ABC-SMC-MNN is expected to perform worse in the estimation of transition rates for such traces. Thus, parameters estimated from the traces with higher numbers of infections are more reliable.

\begin{table*}[!t]
\caption{Statistics from posterior distribution of SIIDR parameters estimated with the ABC-SMC-MNN method.}
\begin{tabular}[t]{|c||c|c|c|c||c|}
\hline
\textbf{WannaCry} &$\beta$
&$\mu$ &$\gamma_1$&$\gamma_2$
&$dt$\\
\textbf{variant} &(mean, std)&(mean, std)&(mean, std)&(mean, std)&\\
\hline
wc\_1\_500ms &(0.16, 0.10) &(0.11, 0.11) &(0.79, 0.15) &(0.06, 0.07) &0.09\\
wc\_1\_1s & (0.16, 0.11) &(0.11, 0.10) & (0.80, 0.15) &(0.06, 0.06) &0.06\\
wc\_1\_5s &(0.05, 0.03) & (0.04, 0.03)&(0.82, 0.12) &(0.02, 0.01)&0.16\\
wc\_1\_10s &(0.13, 0.08) & (0.08, 0.07)&(0.80, 0.15) & (0.05, 0.04)&0.09\\
wc\_1\_20s &(0.22, 0.20) & (0.63, 0.24)&(0.46, 0.28)&(0.51, 0.29)& 0.99\\
wc\_4\_500ms & (0.45, 0.26) &(0.66, 0.24)& (0.53, 0.28)& (0.47, 0.29)&0.05\\
wc\_4\_1s &(0.17, 0.13) & (0.11, 0.13)&(0.79, 0.17) & (0.07, 0.07)&0.05\\
wc\_4\_5s &(0.14, 0.10) & (0.09, 0.08)& (0.76, 0.18)& (0.07, 0.07)&0.07\\
wc\_4\_10s &(0.20, 0.17) & (0.23, 0.20)& (0.74, 0.20)& (0.07, 0.08)&0.10\\
wc\_4\_20s &(0.43, 0.26) & (0.65, 0.24) &(0.50, 0.29) & (0.51, 0.29)&0.14 \\
wc\_8\_500ms &(0.17, 0.14) &(0.16, 0.16)&(0.79, 0.15)&(0.07, 0.08)&0.03\\
wc\_8\_1s &(0.17, 0.14) & (0.16, 0.16)&(0.76, 0.17)& (0.09, 0.09)&0.03\\
wc\_8\_5s &(0.45, 0.27) &(0.63, 0.24)&(0.47, 0.28)&(0.48, 0.29)&0.07 \\
wc\_8\_10s &(0.47, 0.26) & (0.63, 0.24)&(0.49, 0.29)&(0.46, 0.29)&0.06 \\
wc\_8\_20s &(0.14, 0.10) & (0.09, 0.09)& (0.80, 0.15)& (0.07, 0.06)& 0.07\\
\hline
\end{tabular}
\label{tab:siidrparams}
\end{table*}

\begin{figure}[ht]
\centering
\includegraphics[width=0.5\textwidth]{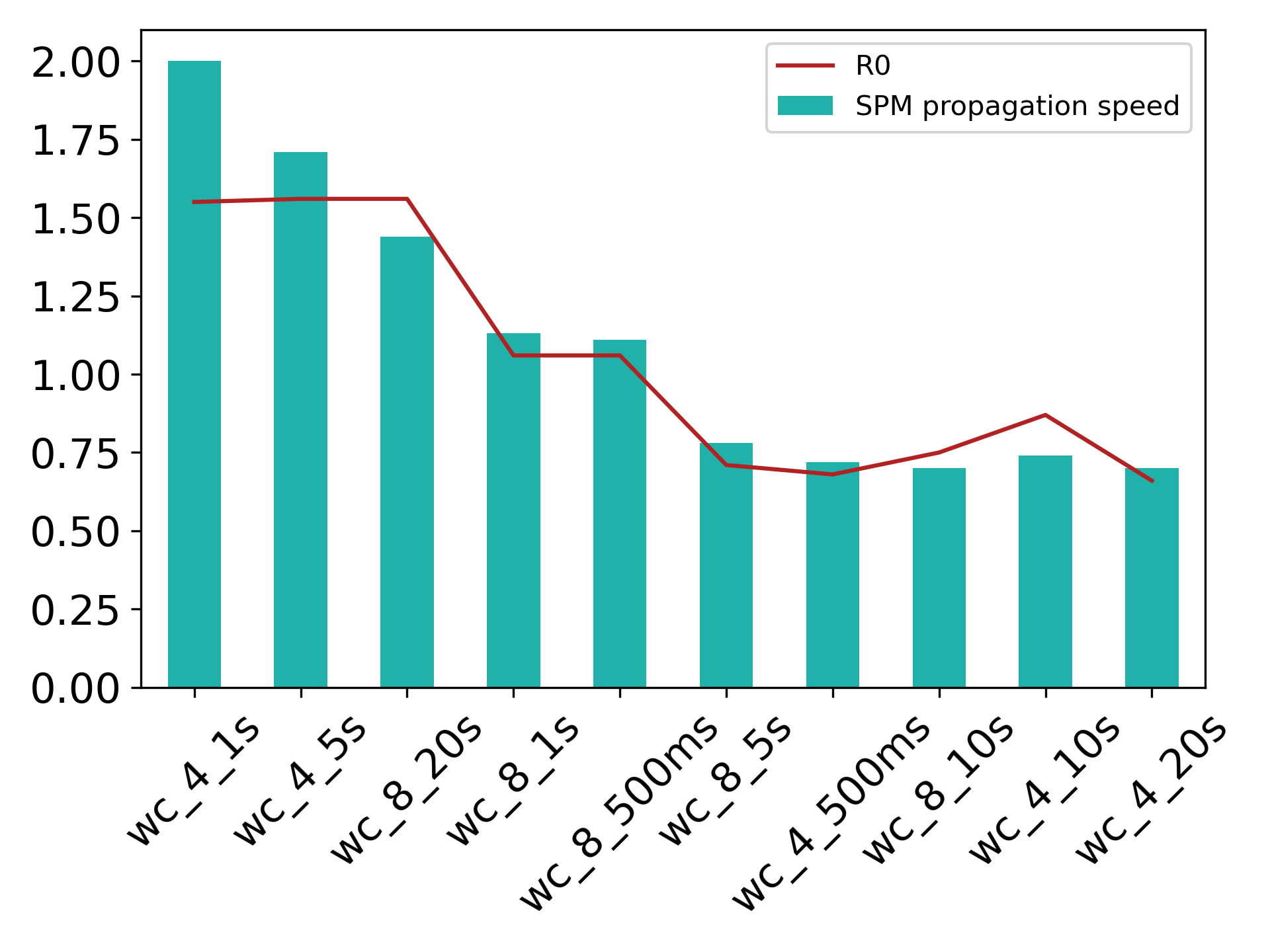}
\caption{The basic reproduction number $R_0$ calculated by using estimated values of transition rates compared to the speed of SPM propagation. Higher propagation speed corresponds to higher $R_0$. We exclude the results for the wc\_1* traces.}
\label{fig:tracesR0}
\end{figure}

\revision{
\subsection{Identifiability of SIIDR transition rates}
\label{sec:siidr_idt}
As long as the goals of modeling with SIIDR include inferences about the underlying propagation process, we are interested in the estimation of SIIDR parameter distribution corresponding to model outputs that best fit the observed data. However, parameters' estimation can only produce robust results if the model is identifiable meaning that it is possible to obtain a unique solution for all unknown parameters given the model structure and output. On the other hand, if parameters are not identifiable their similar values may yield considerably different model outputs~\citep{chis2011structural, tuncer2018structural}.}

\revision{The common problem of data uncertainty forces the issue of parameter identifiability to appear relevant in epidemiological modeling ~\citep{chowell2017fitting, gallo2022lack, weitz2015modeling, valdez2015predicting}. The lack of identifiability in the model parameters may prevent reliable predictions of the epidemic dynamics. Therefore, it becomes crucial to investigate the parameter identifiability, and its limitations and propose solutions to improve it.} 

\revision{There exist notions of structural and practical identifiability. Structural identifiability is a property of the model structure itself given that the model is error-free and the observed data has no noise. 
Practical identifiability is connected to the quality of data leveraged for parameter estimation. It measures whether there is enough information to infer the transition rates~\citep{dankwa2022structural}.}

\revision{We addressed the structural SIIDR parameters identifiability using the method of differential algebra~\citep{chis2011structural,miao2011identifiability} with the help of DAISY~\citep{bellu2007daisy} and SIAN~\citep{hong2020global, ilmer2021web} software and achieved the following result:
\begin{theorem}
    All parameters of the SIIDR model are globally structurally identifiable when incidence represents the output of the model and the size of population $N$ is known. Otherwise, parameters $N$ and $\beta$ appear to be structurally non-identifiable while $\mu, \gamma_1$ and $\gamma_2$ remain identifiable.
\end{theorem}
}
\revision{Therefore, we consider the SIIDR model to be structurally identifiable as long as the size of the computer networks is usually known. More information about SIIDR structural identifiability along with the results from DAISY software can be found in Appendix~\ref{sec:app_idtp}.}

\revision{However, even when the model parameters are structurally identifiable, they may still be non-identifiable in practice due to the limited number of observed variables, the quality of data used for estimation, and the complexity of the model (the number of parameters that are jointly estimated).}

\revision{To investigate practical identifiability we looked at the joint posterior distribution of SIIDR parameters. The plots can be found in Appendix~\ref{sec:app_idprc}. For some of the WC variants, there exists a correlation between parameters $\beta$ and $\mu$. Additionally, some of the joint posterior distributions possess multimodality. Although on average the issue of non-identifiability is not dominant, it might appear in some parts of the phase space of the SIIDR model. One reason for this behavior is that the incidence represents the output of the fitted model and appears to be insufficient to characterize the whole model's dynamic. 
On the other hand, SIIDR has four parameters estimated jointly, therefore, it may contain multiple sets of parameters that lead to the same output of the model. Hence, measuring the data about other states rather than just the number of infected nodes as a function of time to characterize the system dynamics more extensively, should improve the practical SIIDR identifiability.}


\subsection{Threshold Evaluation}
In this section, we evaluate the conditions of SIIDR model equilibrium points to satisfy the Lyapunov stability. Specifically, we are interested in the equilibrium point which corresponds to the start of epidemics, when all nodes in the network have the following probability vector to appear in all of the states of SIIDR model $P^* = \{\Vec{1},\Vec{0}, \Vec{0}, \Vec{0}\}$. We study the stability of this point after the infection of the initial node by SPM (i.e., the system initial state $P_0$ lies in the ball of radius $\delta$ around $P^*$) by looking at the density of recovered nodes w.r.t to the stability threshold $s$ and associated infection propagation dynamics $P_t$.

We evaluate stability conditions on the variety of synthetic and real-world networks described in the following subsection.

\subsubsection{Graphs Characteristics}
We consider synthetic networks generated with Barabási–Albert (BA)~\citep{barabasi1999emergence}, Erdős–Rényi (ER)~\citep{erdds1959random}, Watts–Strogatz (WS)~\citep{watts1998collective}, Configuration Model (CM)~\citep{Newman_2003}, and Scale-free (SF)~\citep{sf} models along with three real-world graphs~\citep{leskovec2005graphs, leskovec2012learning, leskovec2007graph}. Real-world graphs include networks generated using Facebook data (Facebook), autonomous systems peering information inferred from Oregon route views (Oregon), and anonymized traffic data about incoming and outgoing emails between members of the European research institution (Email).
All synthetic graphs have $1000$ nodes and different topological characteristics. Thus, ER graphs have different leading eigenvalues that range from $11$ to $999$, BA networks have the leading eigenvalue between $35$ and $508$, and WS graphs - between $10$ and $900$. ER, BA, and WS networks have only one connected component. They have a larger diameter and average path length, and smaller density and transitivity in the graphs with smaller leading eigenvalues. CM and SF networks have more connected components and the values of other topological characteristics are similar to ER, BA, and WS graphs with small leading eigenvalues.

More details about the topological characteristics of considered networks are presented in Table~\ref{tab:graphs}.

\begin{table*}[!t]
\caption{Topological Characteristics of Graphs Generated for the Stability Condition Evaluation.ER is Erdős–Rényi graph, BA is Barabási–Albert graph, WS is Watts–Strogatz graph, CM is Configuration Model graph, SF is scale-free graph. Dm, T, Dn is the diameter, transitivity, and density of the graph correspondingly.}
\begin{tabular}[t]{|c||c|c|c|c|c|c|c|}
\hline
\textbf{Graph} & \textbf{Number}& \textbf{Number} & $\lambda_A$ & \textbf{Dm} & \textbf{T} & \textbf{Dn} & \textbf{Avg. Path} \\
& \textbf{of Nodes} & \textbf{of Edges}&&&&&\textbf{Length} \\
\hline
ER& 1000& 5054 & 11& 5& 0.01&0.005&3.2\\
ER&1000&49304&100&3&0.1&0.05& 1.9\\
ER&1000&249540&500&2&0.5&0.25&1.5\\
ER&1000&499500&999&1&1&0.5&1 \\
\hline
BA&1000&9900&35&4&0.06&0.01&2.6\\
BA&1000&47500&130&3&0.17&0.05& 1.9\\
BA&1000&90000&222&2&0.27&0.09& 1.8\\
BA&1000&187500&508&2&0.5&0.19& 1.6\\
\hline
WS&1000&5000&10&7&0.48&0.005&4.4\\
WS&1000&50000&100&3&0.56&0.05&2.0\\
WS&1000&250000&500&2&0.63&0.25&1.5\\
WS&1000&299500&900&1&1&0.5&1\\
\hline
CM&1000&995&9&21&0.01&0.001&6.6\\
SF&1000&2165&22&7&0.03&0.002&3.2\\
\hline
Email &265214&365570&103&14&0.004&0.00001& 4.1\\
Facebook &4039&88234&162&8&0.52&0.005&3.7\\
Oregon &11174&23409&60&10&0.01&0.0002& 3.6\\
\hline
\end{tabular}
\label{tab:graphs}
\end{table*}

\subsubsection{Phase Transition}

To illustrate the results of Theorem 2 we plot the final number of recovered nodes in the network with respect to the threshold values $s = \lambda_A * \beta/ \mu$ in the range from 0 to 2. We achieve these results by fixing the transition rates $\mu = 0.5, \gamma_1 = 0.5, \gamma_2 = 0.5$ and changing the value of $\beta$. For ER, BA and WS graphs infection propagation starts from one initially infected node, for SF, CM and real-world networks the fraction of infected nodes at t = 1 is $0.05$. We average results over 100 stochastic realizations that we run considering 50 different seeds. Resulting phase transition plots are illustrated in Figures~\ref{fig:pt_er},~\ref{fig:pt_ba},~\ref{fig:pt_ws}, and~\ref{fig:pt_real}.

For all types of graphs, the total fraction of recovered nodes is negligible for values of $s<1$. As predicted by the theory, the epidemic threshold is $s\sim 1$. In the case of  SF, CM, and real networks (see Figure~\ref{fig:pt_real}), the threshold appears to be for $s<1$. However, we note how in order to obtain macroscopic outbreaks in these graphs, we started the simulations with $5\%$ of initially infected seeds, instead of a single one as done for the other networks. Hence, also for these networks, the phase transition takes place for $s\sim 1$.

In general, networks with larger diameters and average path lengths, smaller density, and transitivity have a smaller fraction of recovered nodes during the infection propagation. 

These results demonstrate that for all $t$ the solution $P_t$ stays in some ball of radius $\epsilon$ from the starting equilibrium point $P^* = P_0$ when $s < 1$, therefore, it is Lyapunov stable. 
Moreover, we see that SIIDR behaves the same as the SIR model in terms of the stability of equilibrium points: when the threshold $s$ is less than one the SIIDR system solution converges to DFE when $t$ tends to infinity. It can be explained by the fact that SIIDR model is very similar to a SIR model except for the particular configuration of transition rates.

\begin{figure}[ht]
\centering
\includegraphics[width=0.5\textwidth]{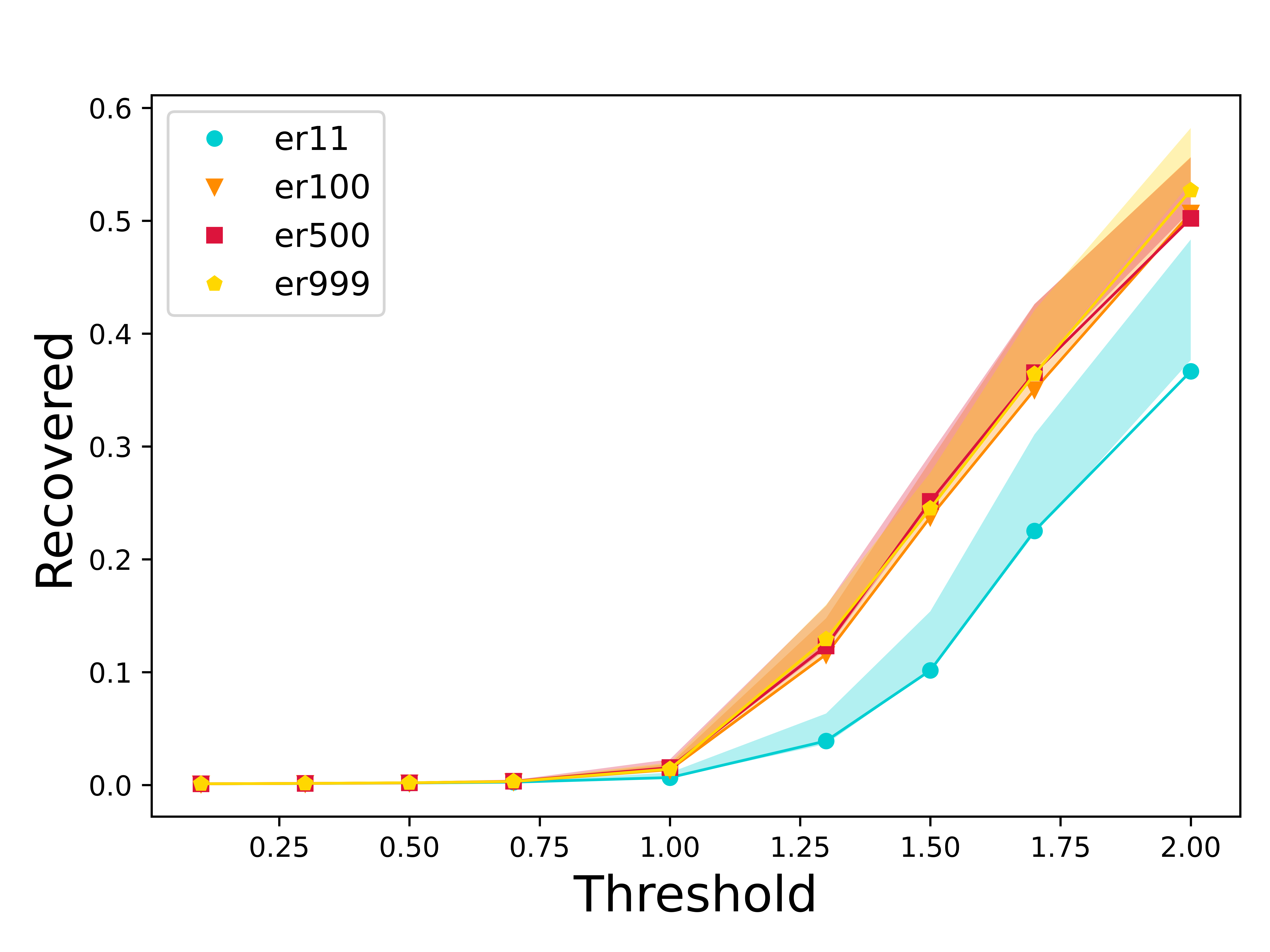}
\caption{The mean value of recovered nodes $R$ with 50\% and 95\% reference ranges obtained from numerical simulations of the SIIDR model on Erdős–Rényi networks with respect to threshold $\lambda_1 * \beta / \mu$ value.}
\label{fig:pt_er}
\end{figure}

\begin{figure}[ht]
\centering
\includegraphics[width=0.5\textwidth]{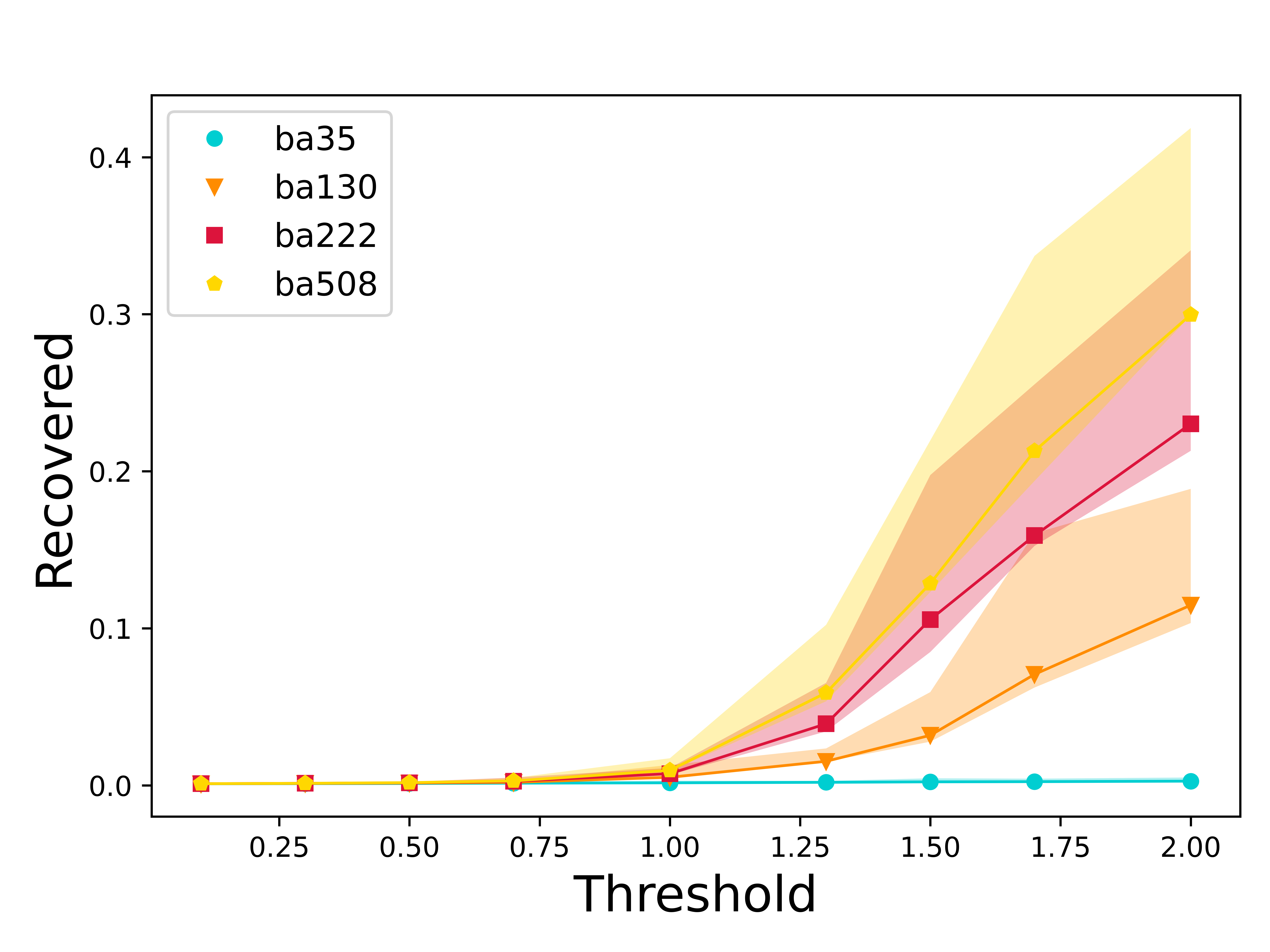}
\caption{The mean value of recovered nodes $R$ with 50\% and 95\% reference range obtained from numerical simulations of the SIIDR model on Barabási–Albert networks with respect to threshold $\lambda_1 * \beta / \mu$ value.}
\label{fig:pt_ba}
\end{figure}

\begin{figure}[ht]
\centering
\includegraphics[width=0.5\textwidth]{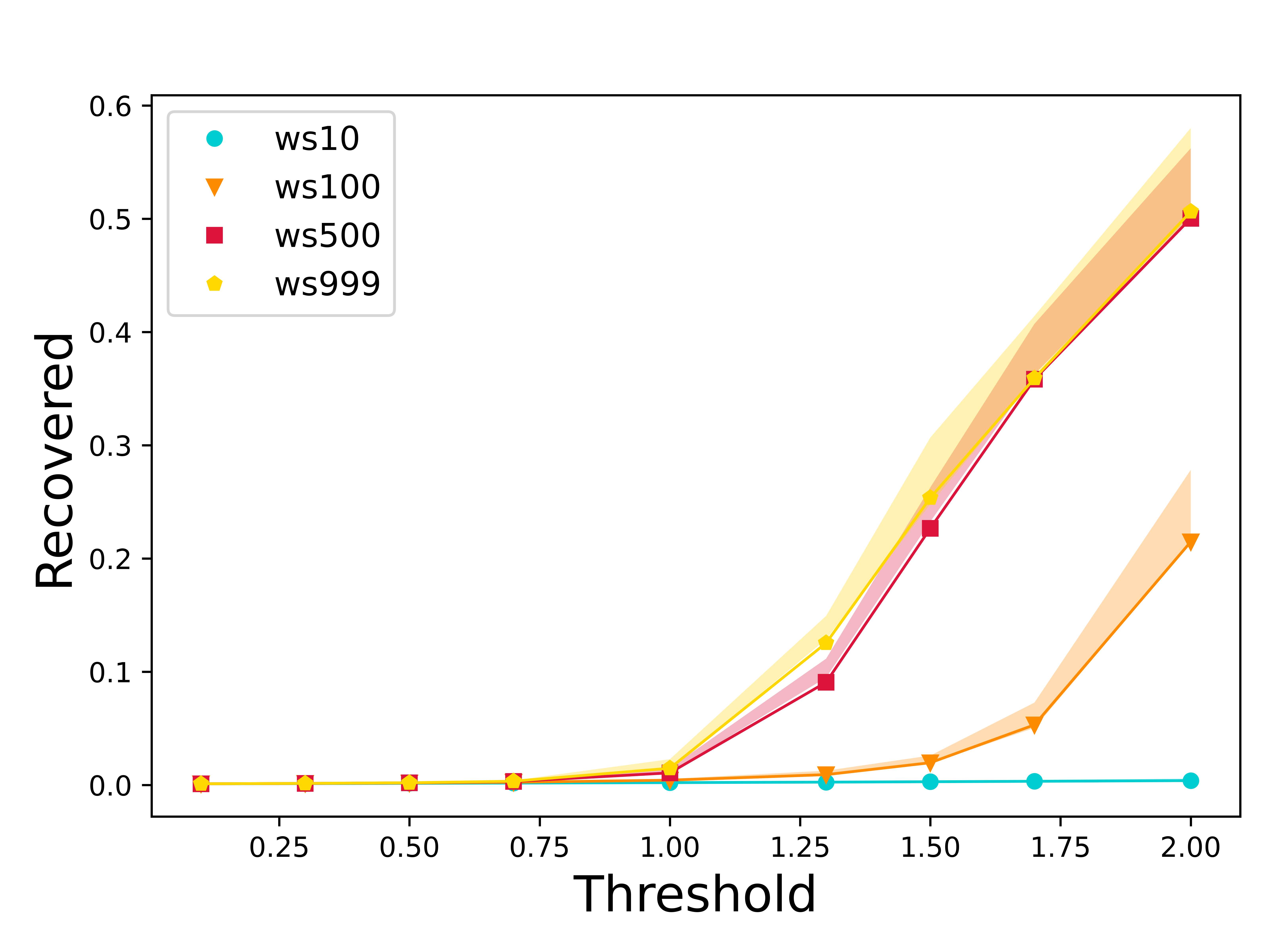}
\caption{The mean value of recovered nodes $R$  with 50\% and 95\% reference range obtained from numerical simulations of the SIIDR model on Watts–Strogatz networks with respect to threshold $\lambda_1 * \beta / \mu$ value.}
\label{fig:pt_ws}
\end{figure}

\begin{figure}[ht]
\centering
\includegraphics[width=0.5\textwidth]{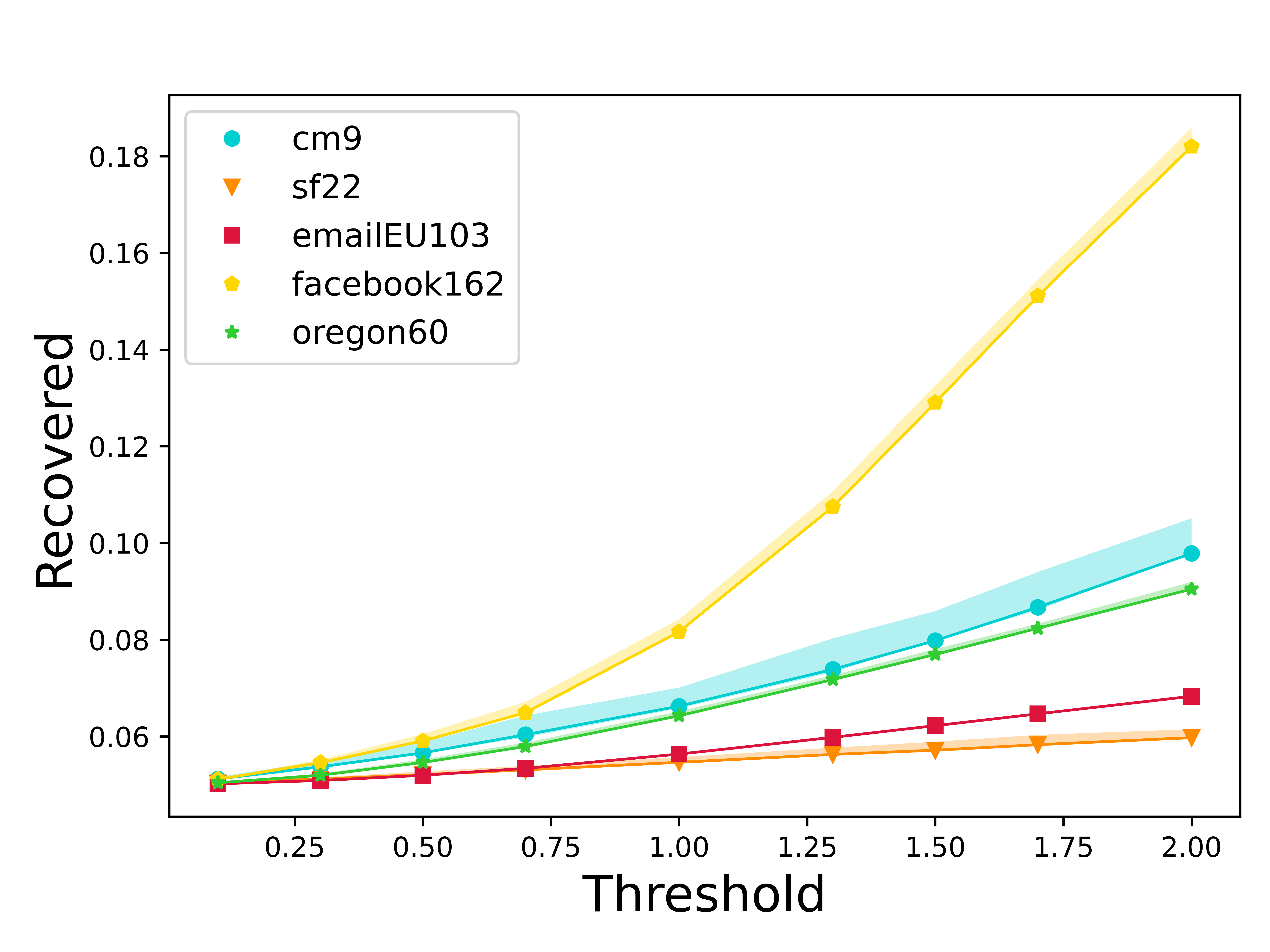}
\caption{The mean value of recovered nodes $R$ with 50\% and 95\% reference range obtained from numerical simulations of the SIIDR model on real-world networks along with scale-free and configuration model graphs with respect to threshold $\lambda_1 * \beta / \mu$ value.}
\label{fig:pt_real}
\end{figure}

\section{Conclusions}
\label{sec:conclusions}
We performed a comprehensive analysis of a new compartmental model, SIIDR, that captures the behavior of self-propagating malware. We showed that SIIDR fits real-world WannaCry traces much better than existing compartmental models such as SI, SIS, SIR, and SEIR (which were previously studied in the literature). Additionally, we estimated the posterior distribution of the model's parameters for real attack traces and showed how they characterize the WannaCry behavior. We also analytically derived the conditions when SPM is expected to become an epidemic and discussed the stability of model's disease-free equilibrium points. Our work demonstrates the impact of modeling the propagation of SPM, simulating real attacks on networks, and evaluating defensive techniques.

\backmatter


\bmhead{Acknowledgments}
We acknowledge Jason Hiser and Jack W.~Davidson from University of Virginia for providing us access to the WannaCry attack traces. 

\section*{Declarations}

\begin{itemize}
\item This research was sponsored by the U.S. Army Combat Capabilities Development Command Army Research Laboratory under Cooperative Agreement Number W911NF-13-2-0045 (ARL Cyber Security CRA). The views and conclusions contained in this document are those of the authors and should not be interpreted as representing the official policies, either expressed or implied, of the Combat Capabilities Development Command Army Research Laboratory or the U.S. Government. The U.S. Government is authorized to reproduce and distribute reprints for Government purposes notwithstanding any copyright notation here on.
\item The authors declare that they have no competing interests.
\item The datasets supporting the conclusions of this article are available in the github repository: 
\url{https://github.com/achernikova/siidr/}. WannaCry data is available from the corresponding author on reasonable request.
\item The code is available in the github repository: \url{https://github.com/achernikova/siidr/}.
\item NG and NP proposed the SIIDR model. AC, NG, and NP contributed to the methodology of the paper. AC and NG performed the experiments. All authors contributed to the discussion and writing of the paper, and approved the final manuscript.
\end{itemize}

\noindent

\begin{appendices}

\section{Compartmental Models of Epidemiology}
\label{sec:models}
\subsection{SI model}
The SI model is used to describe diseases where infection is permanent. It features two compartments and one transition. The susceptible compartment $S$ represents healthy individuals that interacting with infectious individuals in the compartment $I$ can get infected ($S + I \rightarrow 2I$). It can be translated in the following system of ODEs:
\begin{equation*}
    \begin{aligned}
    \frac{dS}{dt} &= -\beta S \frac{I}{N} \\
    \frac{dI}{dt} &= \beta S \frac{I}{N} \\
    \end{aligned}
\end{equation*}
Due to the homogeneous mixing assumption, the per capita rate at which susceptible individuals get infected can be written as the probability of interacting with an infected individual ($I/N$) times the transmission rate of the disease $\beta$. 
The state diagram for the SI model is shown in Figure~\ref{fig:SIstates}.
\begin{figure}[ht]
\centering
\includegraphics[width=0.2\textwidth]{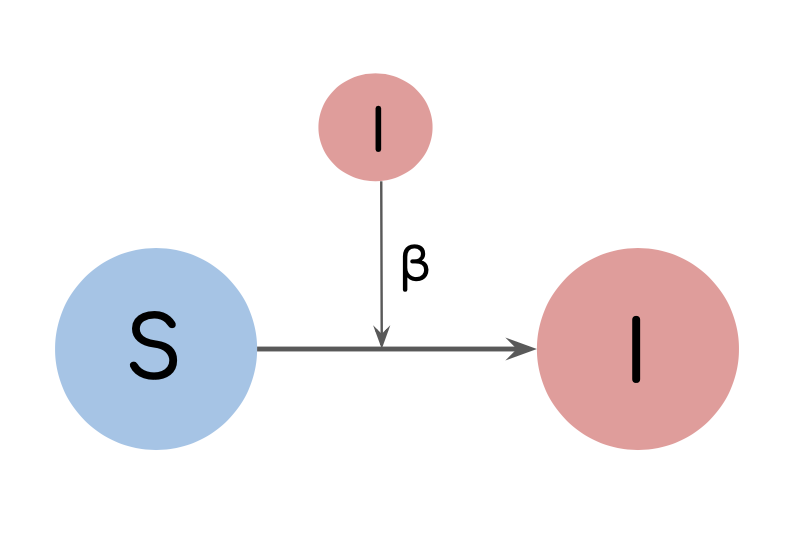}
\caption{Schematic representation of the SI model.}
\label{fig:SIstates}
\end{figure}

\subsection{SIS model}
The SIS model features two compartments and two transitions. Beside the infection process as in the SI model, SIS models have also a recovery process: infected individuals spontaneously recover at rate $\mu$ becoming susceptible to the disease again ($I \rightarrow S$). Hence SIS models are used for diseases that can infect individuals multiple times.
The system of ODEs associated with SIS model is:
\begin{equation*}
\begin{aligned}
    \frac{dS}{dt} &= -\beta S \frac{I}{N} + \mu I\\
    \frac{dI}{dt} &= \beta S \frac{I}{N} - \mu I\\
\end{aligned}
\end{equation*}
Note how, differently from infection, the recovery process is spontaneous and does not require any interaction. Hence, each infected individual has an average duration of infection of $\mu^{-1}$. The state diagram for SIS model is shown in Figure~\ref{fig:SISstates}.
\begin{figure}[ht]
\centering
\includegraphics[width=0.2\textwidth]{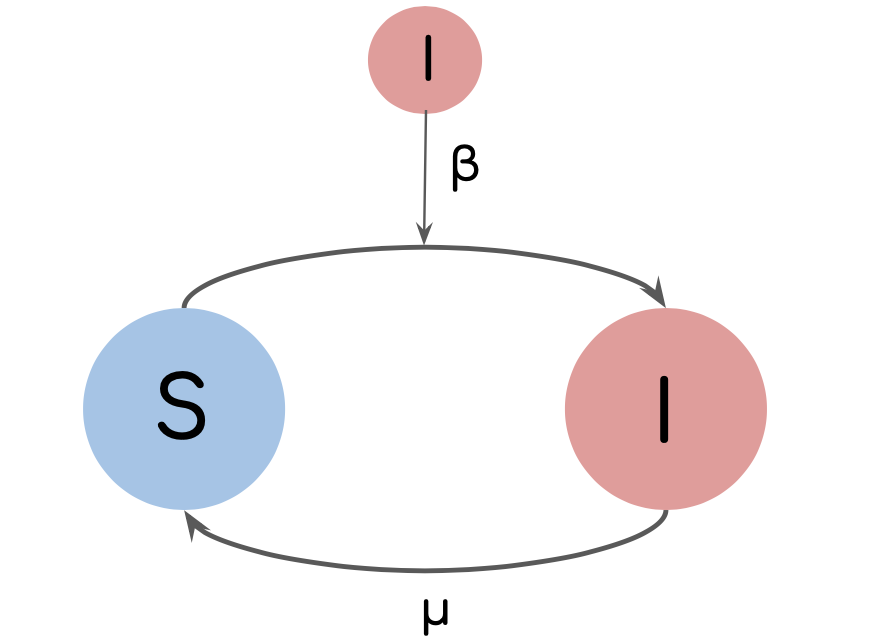}
\caption{Schematic representation of the SIS model.}
\label{fig:SISstates}
\end{figure}

\subsection{SIR model}
The SIR model describes diseases that give permanent (or long-lasting) immunity. It features three compartments and two transitions. Differently from SIS models, within the SIR framework infected individuals that are no longer infectious transition to the recovered compartment $R$. 
The system of differential equations corresponding to the SIR model is the following:
\begin{equation*}
\begin{aligned}
    \frac{dS}{dt} &= -\beta S \frac{I}{N}\\
    \frac{dI}{dt} &= \beta S \frac{I}{N} - \mu I\\
    \frac{dR}{dt} &= \mu I
\end{aligned}
\end{equation*}
The state diagram for the SIR model is represented in Figure~\ref{fig:SIRstates}.
\begin{figure}[ht]
\centering
\includegraphics[width=0.3\textwidth]{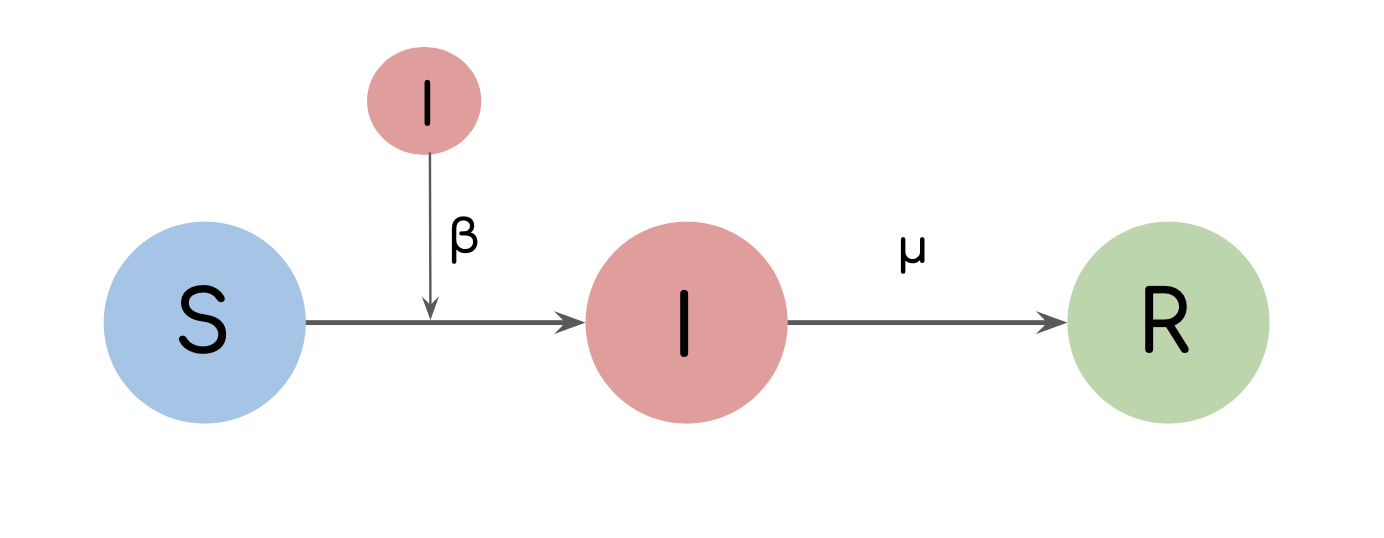}
\caption{Schematic representation of the SIR model.}
\label{fig:SIRstates}
\end{figure}

\revision{
\subsection{SEIR model}
The SEIR model describes diseases where susceptible individuals $S$ remain exposed $E$ after interaction with infected $I$ individual before becoming infectious themselves. It features four compartments and three transitions. 
The system of differential equations corresponding to the SEIR model is the following:
\begin{equation*}
\begin{aligned}
    \frac{dS}{dt} &= -\beta S \frac{I}{N}\\
    \frac{dE}{dt} & = \beta S \frac{I}{N} - \gamma E\\
    \frac{dI}{dt} &= \gamma E - \mu I\\
    \frac{dR}{dt} &= \mu I
\end{aligned}
\end{equation*}
The state diagram for the SEIR model is represented in Figure~\ref{fig:SEIRstates}.
\begin{figure}[ht]
\centering
\includegraphics[width=0.3\textwidth]{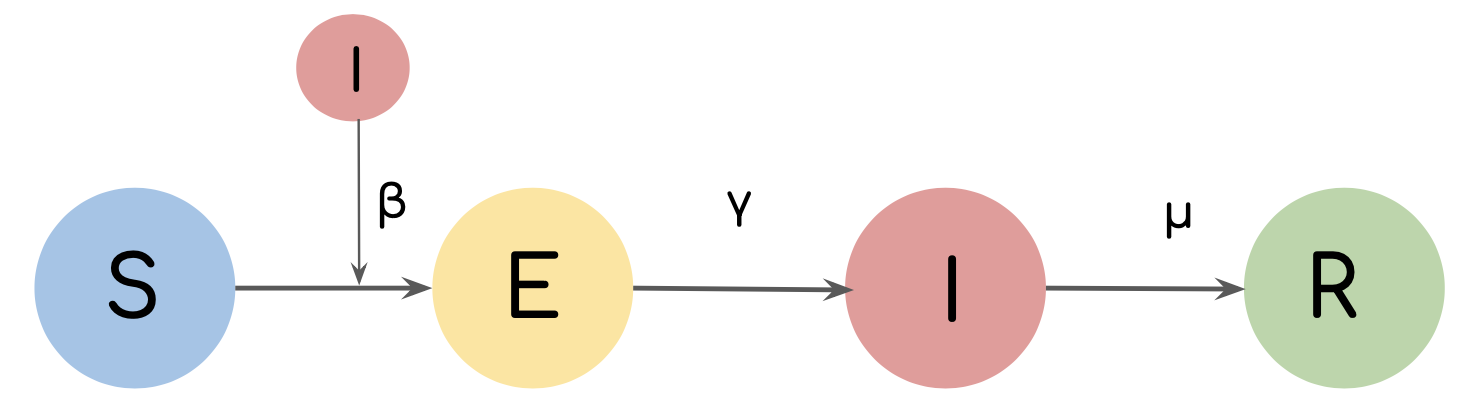}
\caption{Schematic representation of the SEIR model.}
\label{fig:SEIRstates}
\end{figure}
}

\revision{
\section{Identifiability of SIIDR transition rates}
\label{sec:app_idtp}
SIIDR model can be represented as follows:
\begin{equation}
\begin{aligned}
SIIDR :=
\begin{cases}
    \dot X(t) =f(X(t), \theta)\\
    Y(t) = g(X(t), \theta)\\
    X_{0} = X(t_0)\\
\end{cases}
\end{aligned}
\end{equation}
where $t_0 \leq t \leq T$, $\dot X(t)$ is a system of ODEs, $X(t)$ is a vector of time-varying diseases states and the unique solution to the system $\dot X(t)$, $\theta \in \Theta$ is a vector of constant unknown model parameters, $Y(t)$ is a vector of time-dependent model outputs, $g$ is the measurement equation which defines the relationship between $X(t)$, $\theta$ and $Y$, and $X_0$ is a vector of the known initial conditions.
\begin{definition}
A parameter $\theta$ is structurally globally identifiable if  $\forall$ $\theta^* \in \Theta$:
\begin{equation*}
    SIIDR(\theta^*) = SIIDR(\theta) \Rightarrow \theta^* = \theta
\end{equation*}
\end{definition}
\begin{definition}
A parameter  $\theta$ is structurally locally identifiable if $\forall$ $\theta^* \in \Theta$, there exists a neighbourhood $\Omega(\theta)$  such that
\begin{equation*}
\theta^* \in \Omega(\theta) \vee SIIDR(\theta^*) = SIIDR(\theta) \Rightarrow \theta^* = \theta
\end{equation*}
\end{definition}
A variety of methods exists to evaluate the structural and practical identifiability of parameters. In our work, we leveraged the method of differential algebra implemented in DAISY~\citep{bellu2007daisy} and SIAN~\citep{hong2020global, ilmer2021web} software to address the structural identifiability of SIIDR. We looked at the joint posterior distribution of SIIDR parameters to address the issue of practical identifiability. We discuss SIIDR identifiability results in the following subsections.
}
\revision{
\subsection{Differential Algebra Approach for Structural Identifiability}
\label{sec:app_idstr}
In this section, we show the results for structural identifiability of SIIDR parameters achieved with the differential algebra approach implemented in DAISY software. Figures~\ref{fig:SIIDR_daisy1_in},~\ref{fig:SIIDR_daisy1} represent the input and output of the DAISY software when the number of infected nodes is the output variable $Y(t)$. Figures~\ref{fig:SIIDR_daisy2_in},~\ref{fig:SIIDR_daisy2} show the results from DAISY software in the situation when the sum of infected, infected dormant, and recovered nodes is the output variable. When the size of the population $N$ is known, we can exclude it from the ODE equations and consider $\beta = \beta/N$ to be the unknown parameter. In both cases all parameters of the SIIDR model are globally structurally identifiable. Figures~\ref{fig:SIIDR_daisy1_N},~\ref{fig:SIIDR_daisy2_N} show the results when the $N$ is the unknown parameter. In this sutiation, parameters $\beta$ and $N$ are not identifiable, however, $\mu, \gamma_1, \gamma_2$ remain identifiable.}
\begin{figure}[h]
\centering
\includegraphics[width=0.5\textwidth]{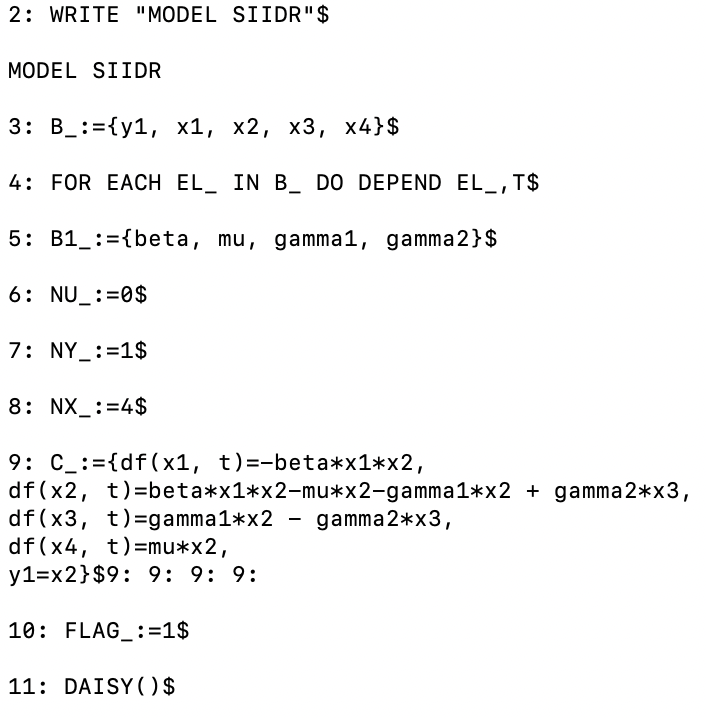}
\caption{Input for DAISY to evaluate the SIIDR structural identifiability of parameters when output is the number of infected individuals.}
\label{fig:SIIDR_daisy1_in}
\end{figure}

\begin{figure}[h]
\centering
\includegraphics[width=0.5\textwidth]{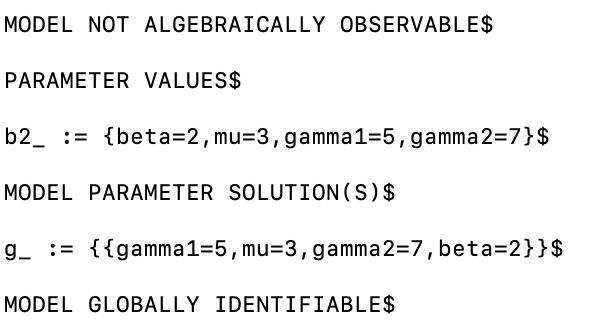}
\caption{SIIDR structural identifiability of parameters achieved with DAISY when output is the number of infected individuals.}
\label{fig:SIIDR_daisy1}
\end{figure}

\begin{figure}[h]
\centering
\includegraphics[width=0.5\textwidth]{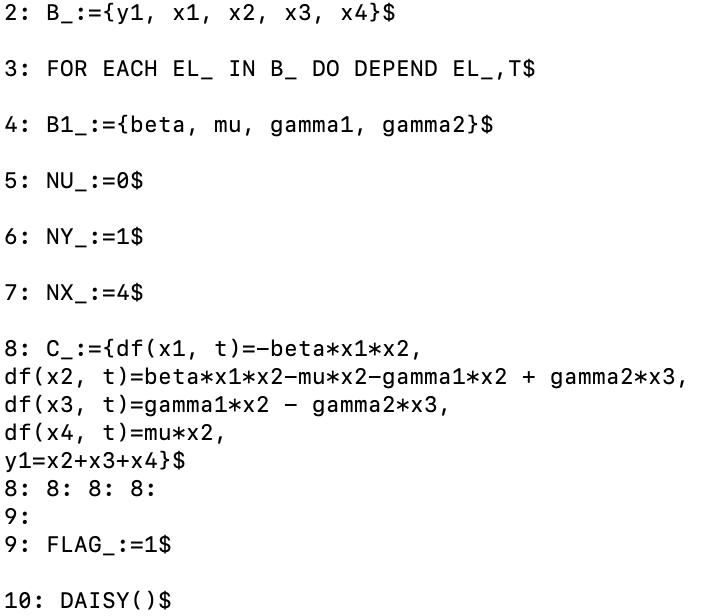}
\caption{Input for DAISY to evaluate the SIIDR structural identifiability of parameters when output is the sum of infected, infected dormant and recovered individuals.}
\label{fig:SIIDR_daisy2_in}
\end{figure}

\begin{figure}[h]
\centering
\includegraphics[width=0.5\textwidth]{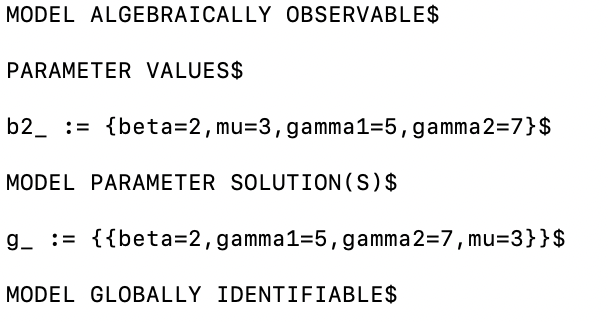}
\caption{SIIDR structural identifiability of parameters achieved with DAISY when output is the sum of infected, infected dormant and recovered individuals.}
\label{fig:SIIDR_daisy2}
\end{figure}

\begin{figure}[h]
\centering
\includegraphics[width=0.5\textwidth]{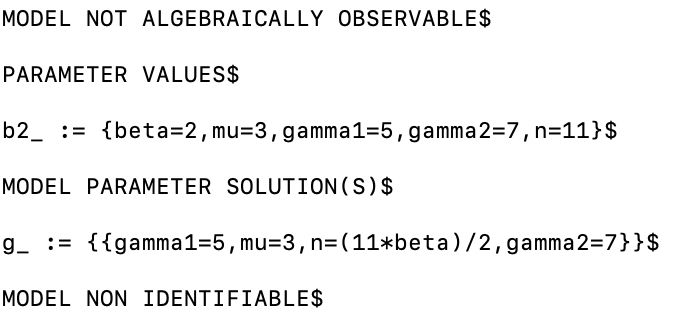}
\caption{SIIDR structural identifiability of parameters achieved with DAISY when output is the sum of infected, infected dormant and recovered individuals and the size of population $N$ is unknown.}
\label{fig:SIIDR_daisy1_N}
\end{figure}

\begin{figure}[h]
\centering
\includegraphics[width=0.5\textwidth]{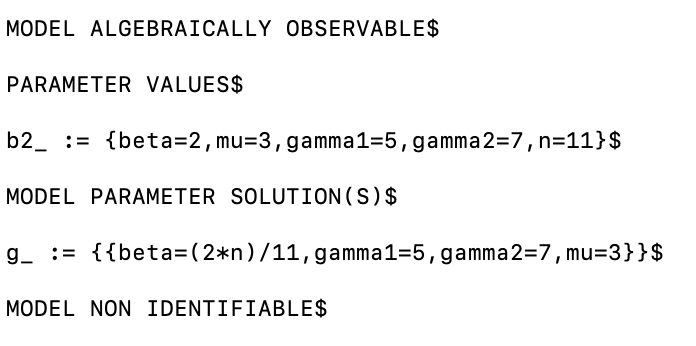}
\caption{SIIDR structural identifiability of parameters achieved with DAISY when output is the sum of infected, infected dormant and recovered individuals and the size of population $N$ is unknown.}
\label{fig:SIIDR_daisy2_N}
\end{figure}

\revision{
\subsection{Joint Posterior Distributions of SIIDR Parameters}
\label{sec:app_idprc}
}
\revision{In this subsection, we illustrate the joint posterior distribution for SIIDR parameters. The plots for wc\_4\_500ms variant are in Figures~\ref{fig:wc4500_1_param},~\ref{fig:wc4500_2_param}. In this case, joint posterior distribution has multiple modes which means that the parameters value are not uniquely identifiable. The results for wc\_8\_20s are illustrated in Figures~\ref{fig:wc820_1_param},~\ref{fig:wc820_2_param}. In this situation, $\beta$ and $\mu$ parameters are correlated. In Figures~\ref{fig:wc15_1_param},~\ref{fig:wc15_2_param} we show the results for wc\_1\_5s variant. The posterior joint distribution of $\beta$ and $\mu$ parameters are not correlated and there is no multimodality.}

\begin{figure}
     \centering
     \begin{subfigure}[b]{0.3\textwidth}
         \centering
         \includegraphics[width=\textwidth]{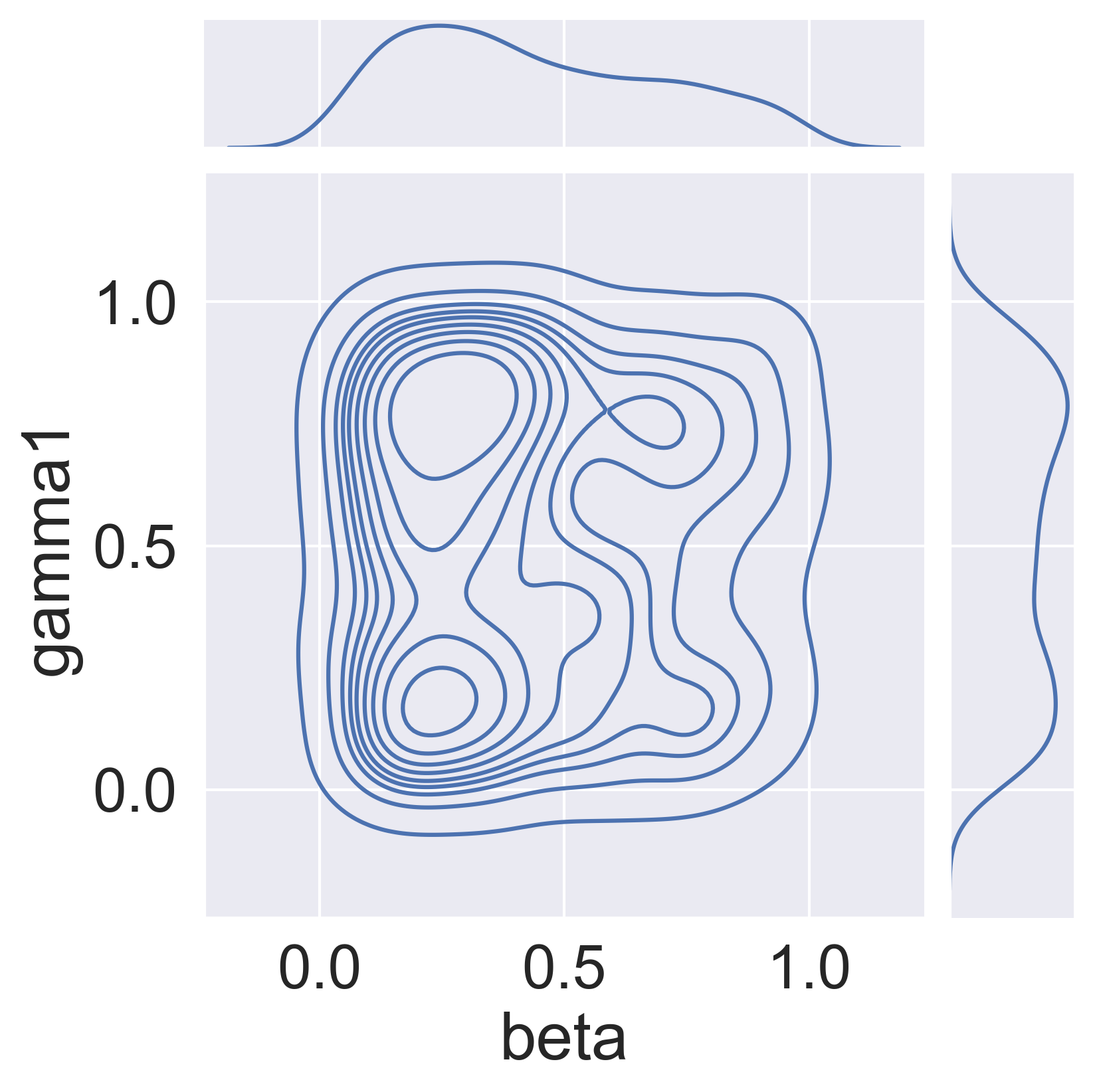}
         \caption{$\beta$ and $\gamma_1$}
     \end{subfigure}
     \hfill
     \begin{subfigure}[b]{0.3\textwidth}
         \centering
         \includegraphics[width=\textwidth]{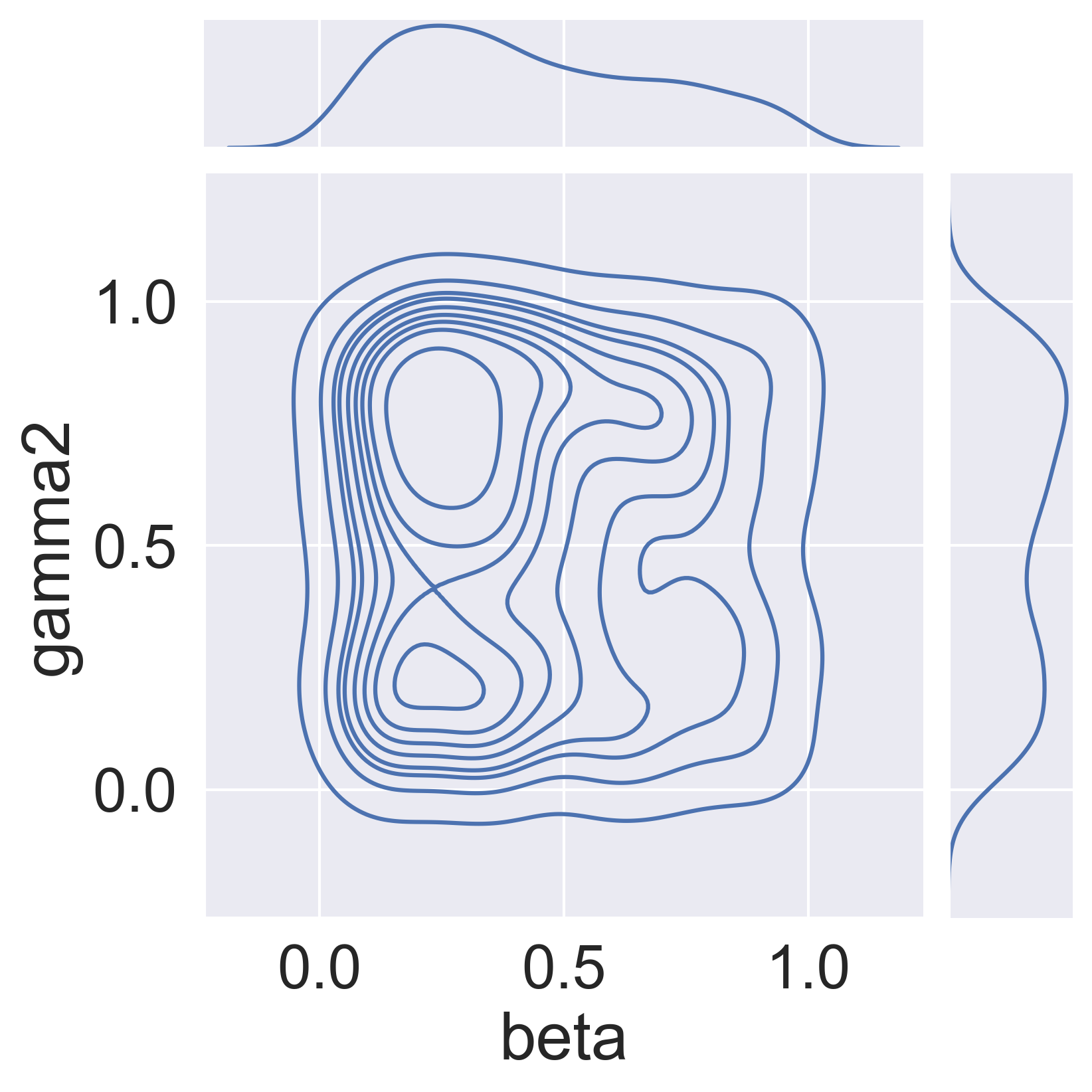}
         \caption{$\beta$ and $\gamma_2$}
     \end{subfigure}
     \hfill
     \begin{subfigure}[b]{0.3\textwidth}
         \centering
         \includegraphics[width=\textwidth]{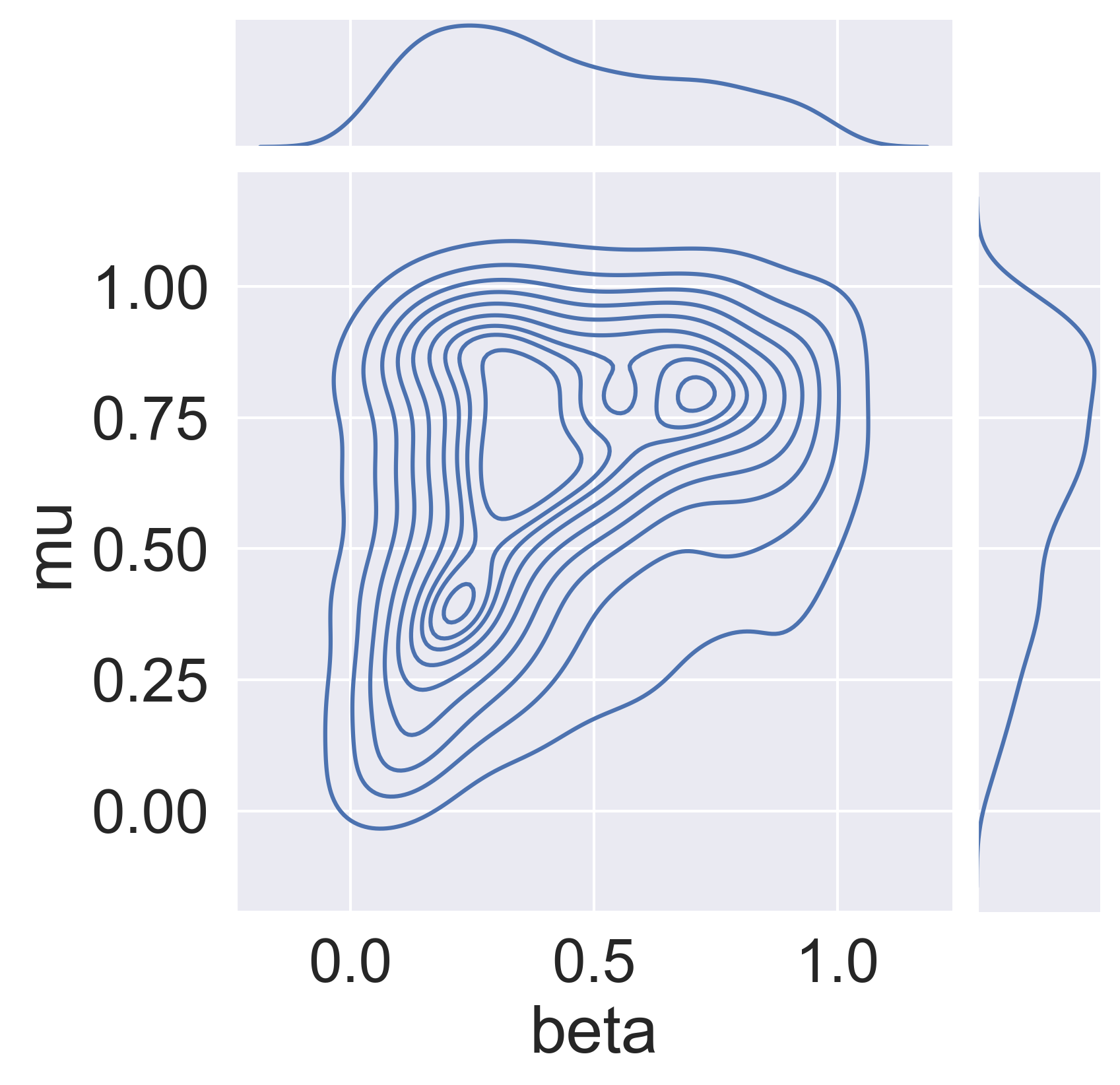}
         \caption{$\beta$ and $\mu$}
     \end{subfigure}
        \caption{Joint posterior distribution for SIIDR parameters for wc\_4\_500ms variant.}
        \label{fig:wc4500_1_param}
\end{figure}

\begin{figure}
     \centering
     \begin{subfigure}[b]{0.3\textwidth}
         \centering
         \includegraphics[width=\textwidth]{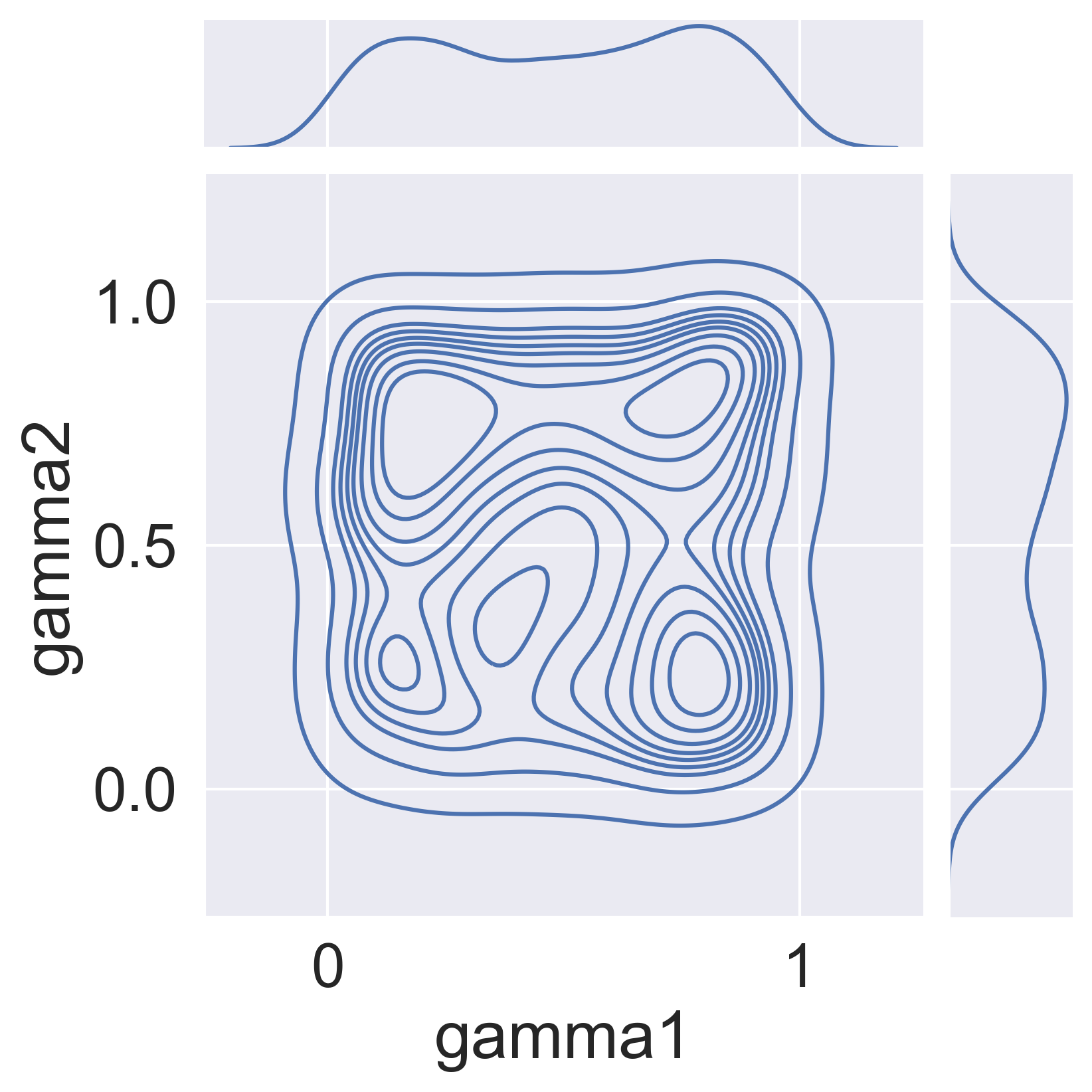}
         \caption{$\gamma_1$ and $\gamma_2$}
     \end{subfigure}
     \hfill
     \begin{subfigure}[b]{0.3\textwidth}
         \centering
         \includegraphics[width=\textwidth]{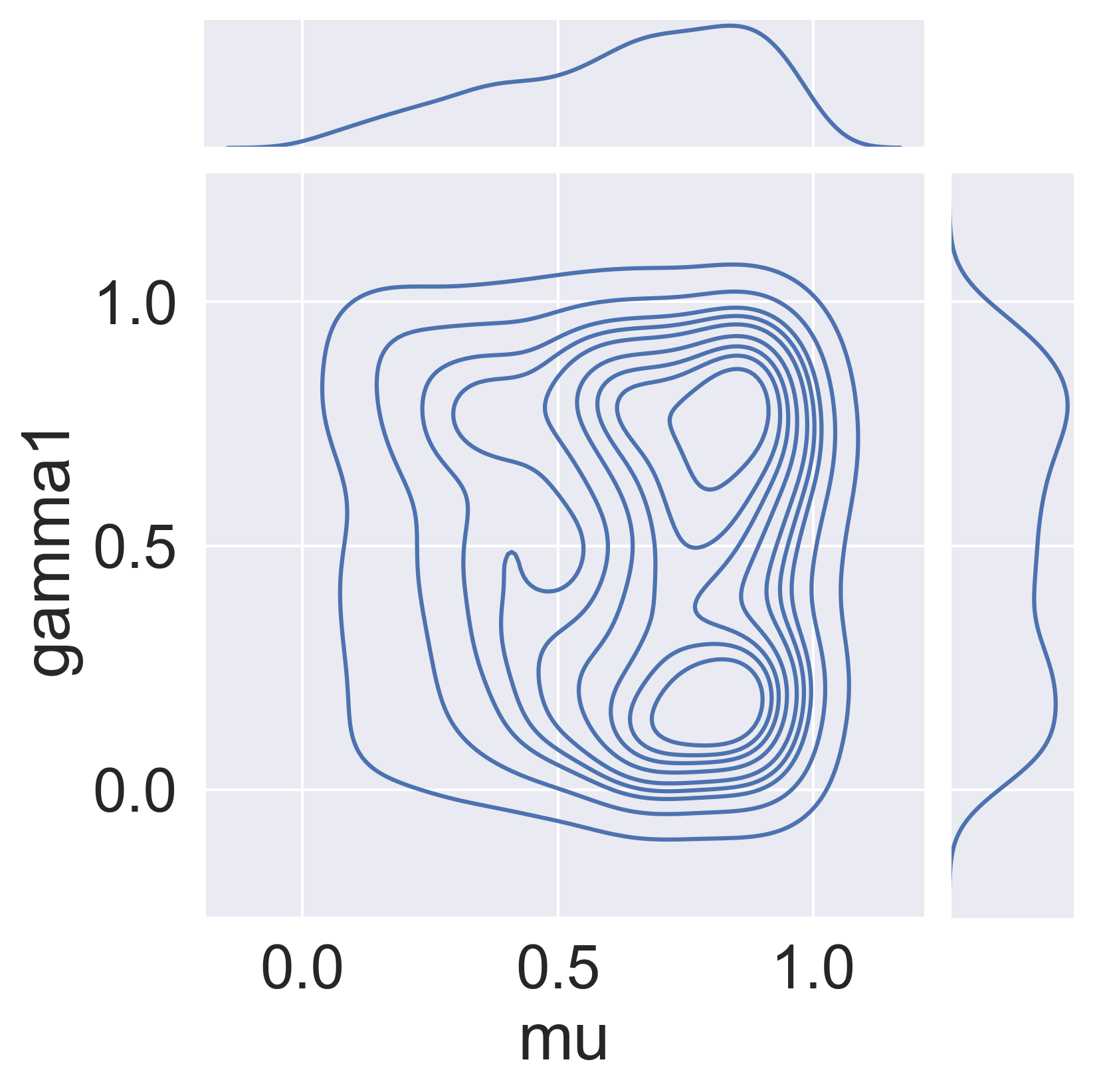}
         \caption{$\mu$ and $\gamma_1$}
     \end{subfigure}
     \hfill
     \begin{subfigure}[b]{0.3\textwidth}
         \centering
         \includegraphics[width=\textwidth]{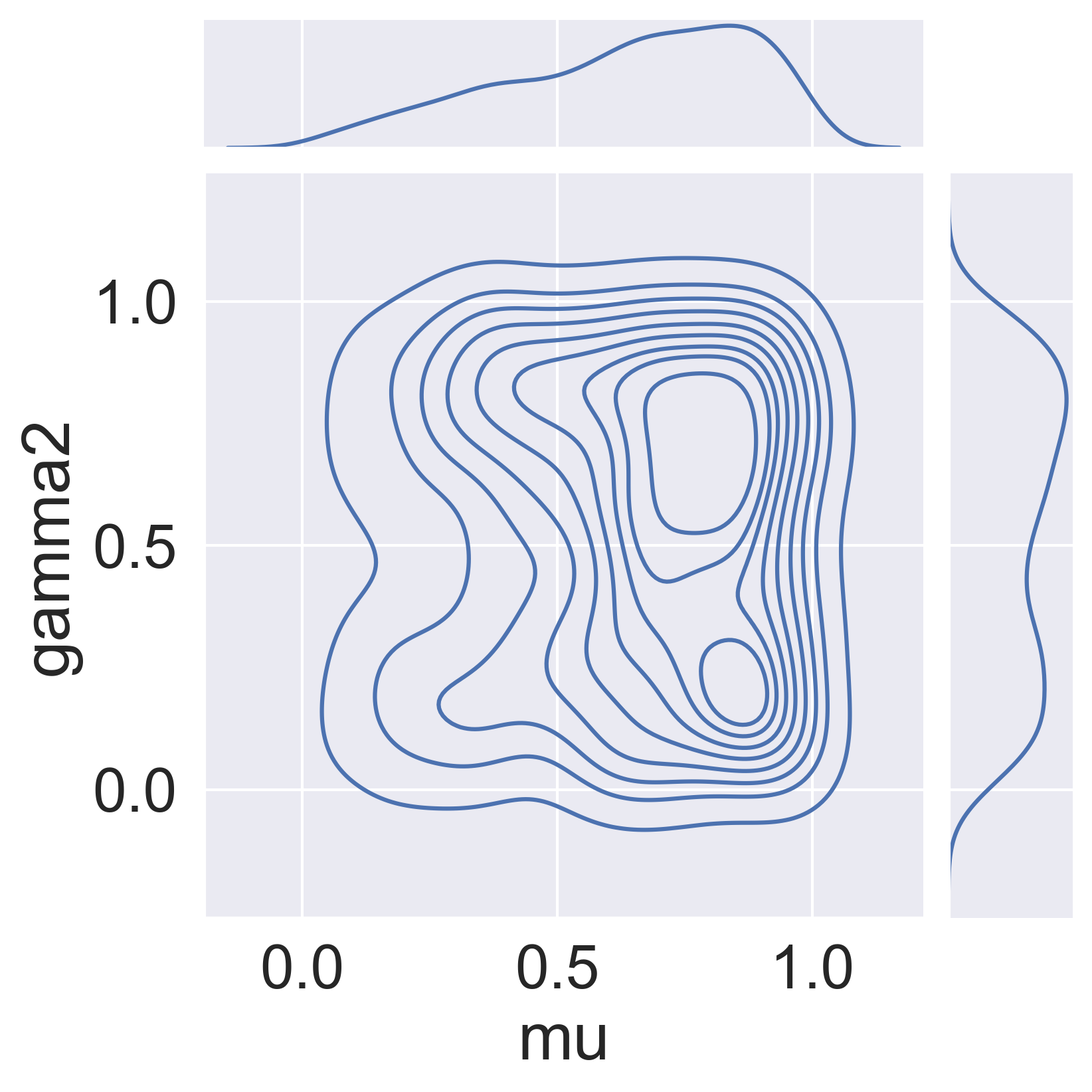}
         \caption{$\mu$ and $\gamma_2$}
     \end{subfigure}
        \caption{Joint posterior distribution for SIIDR parameters for wc\_4\_500ms variant.}
        \label{fig:wc4500_2_param}
\end{figure}

\begin{figure}
     \centering
     \begin{subfigure}[b]{0.3\textwidth}
         \centering
         \includegraphics[width=\textwidth]{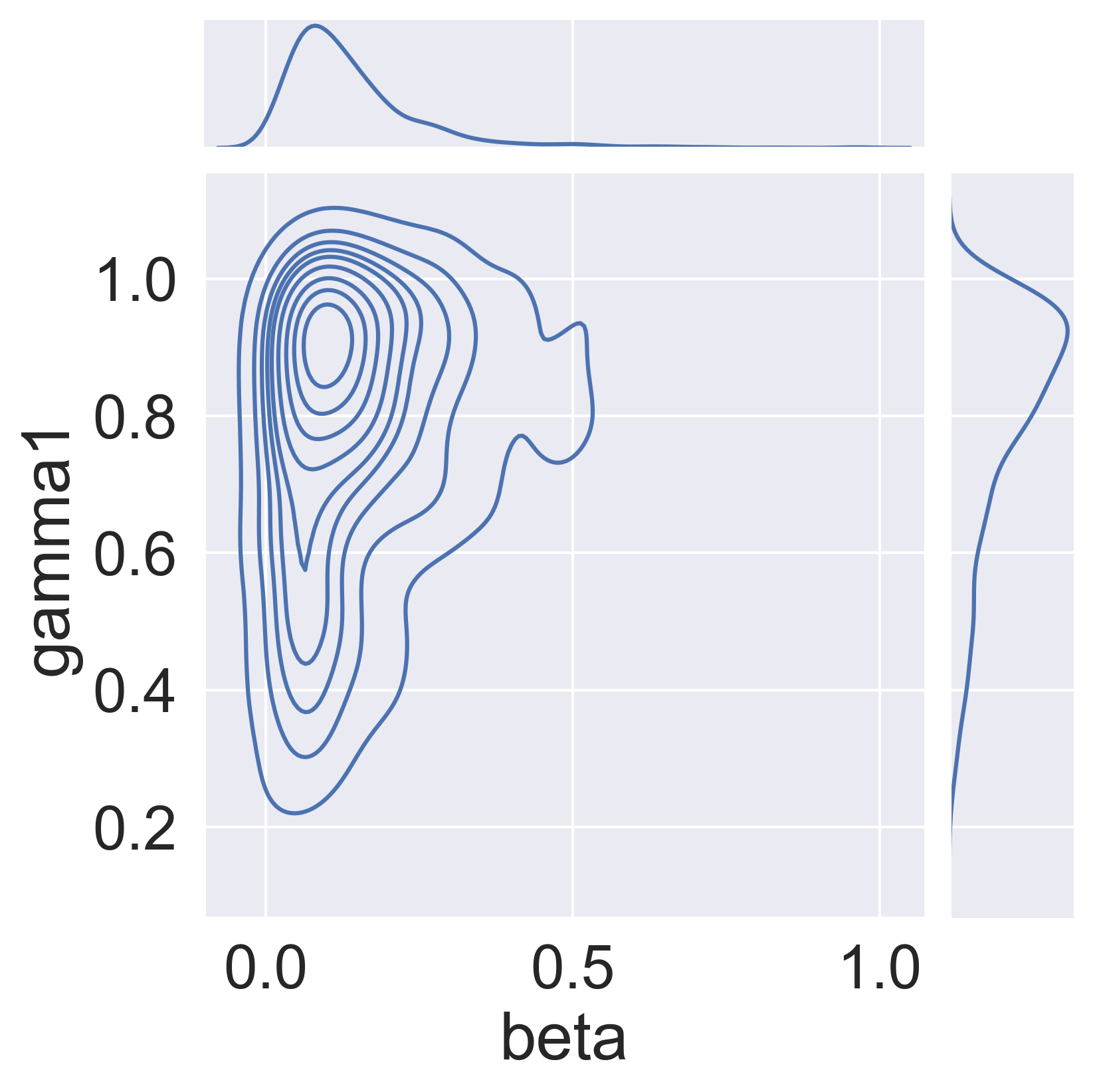}
         \caption{$\beta$ and $\gamma_1$}
     \end{subfigure}
     \hfill
     \begin{subfigure}[b]{0.3\textwidth}
         \centering
         \includegraphics[width=\textwidth]{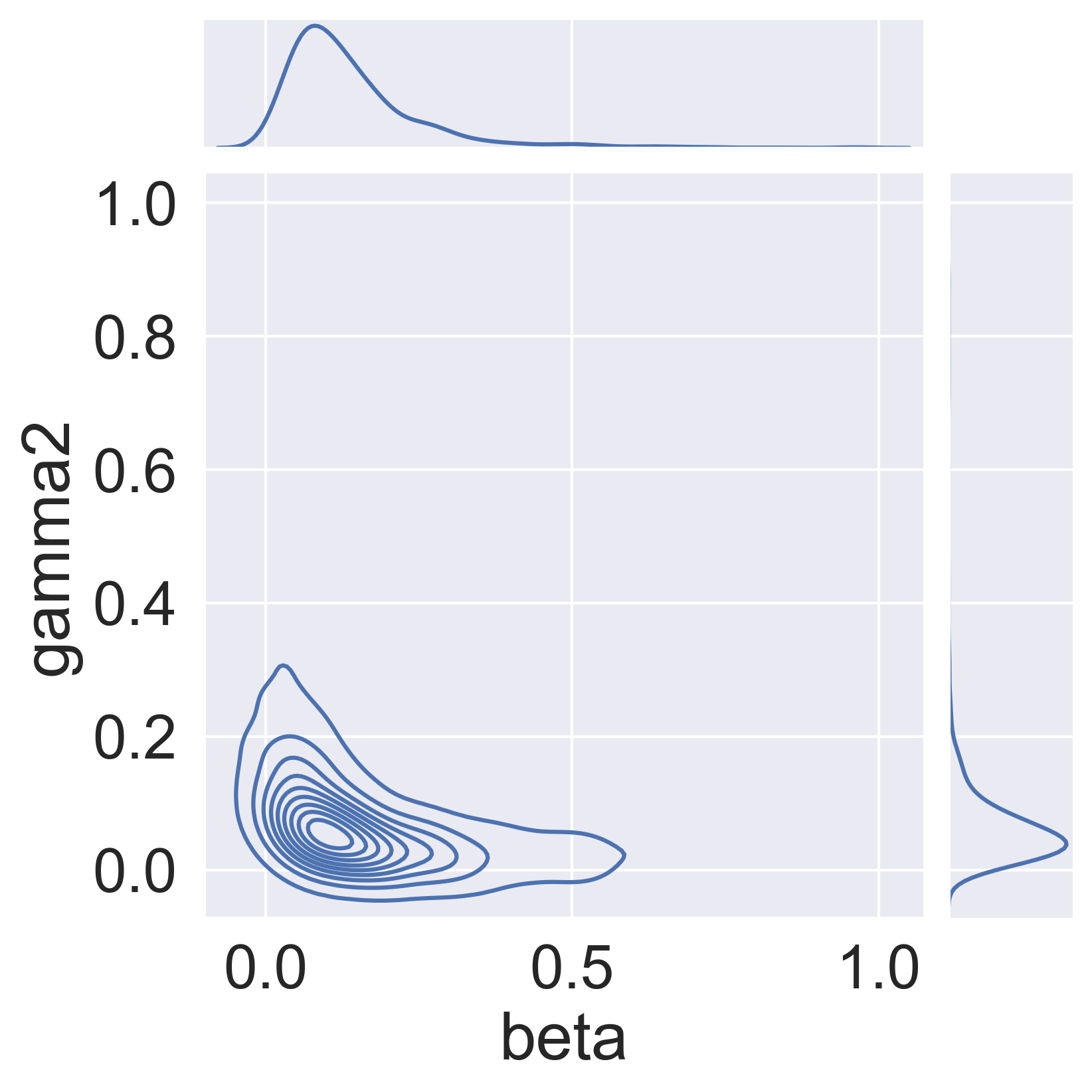}
         \caption{$\beta$ and $\gamma_2$}
     \end{subfigure}
     \hfill
     \begin{subfigure}[b]{0.3\textwidth}
         \centering
         \includegraphics[width=\textwidth]{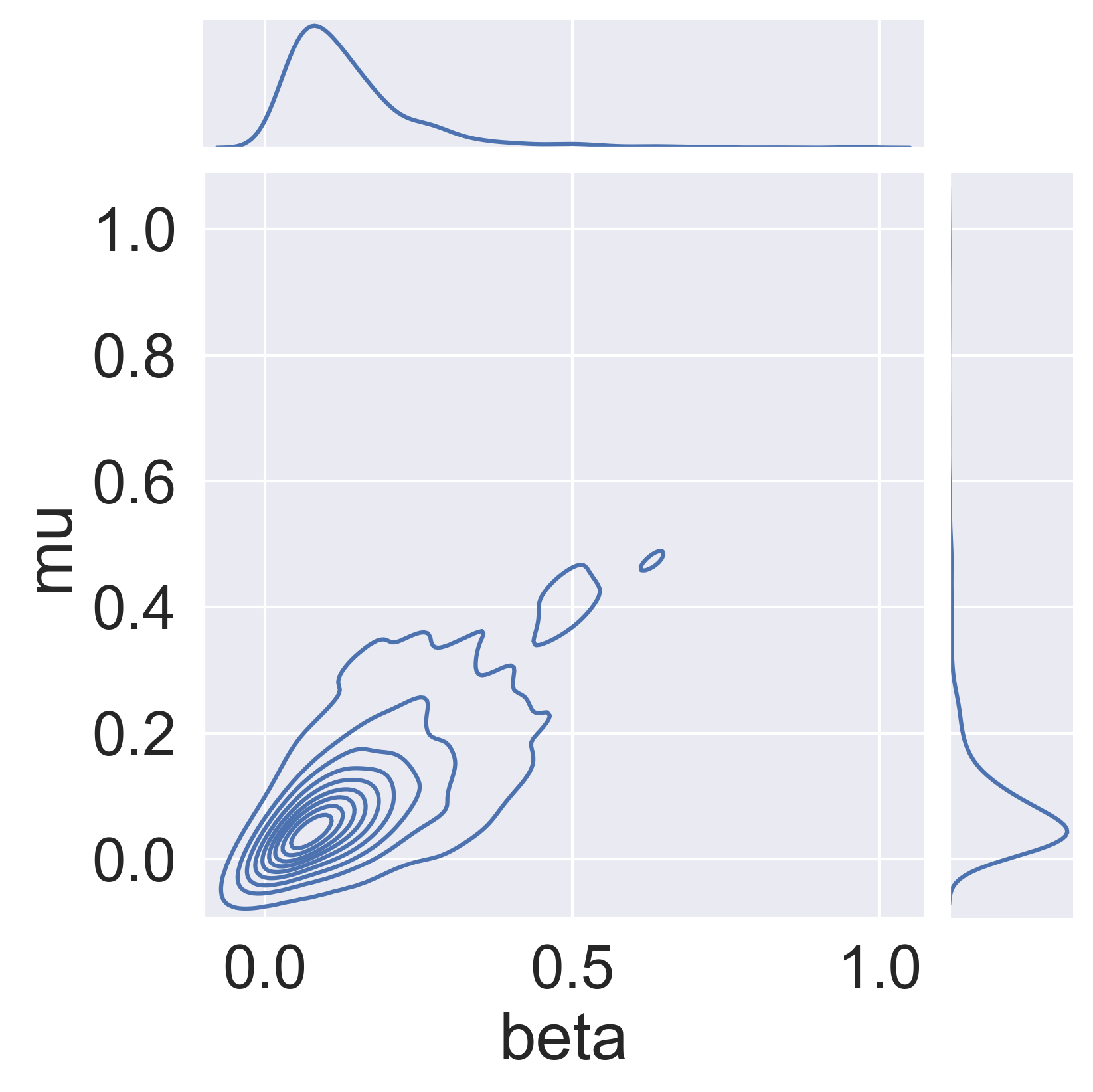}
         \caption{$\beta$ and $\mu$}
     \end{subfigure}
        \caption{Joint posterior distribution for SIIDR parameters for wc\_8\_20s variant.}
        \label{fig:wc820_1_param}
\end{figure}

\begin{figure}
     \centering
     \begin{subfigure}[b]{0.3\textwidth}
         \centering
         \includegraphics[width=\textwidth]{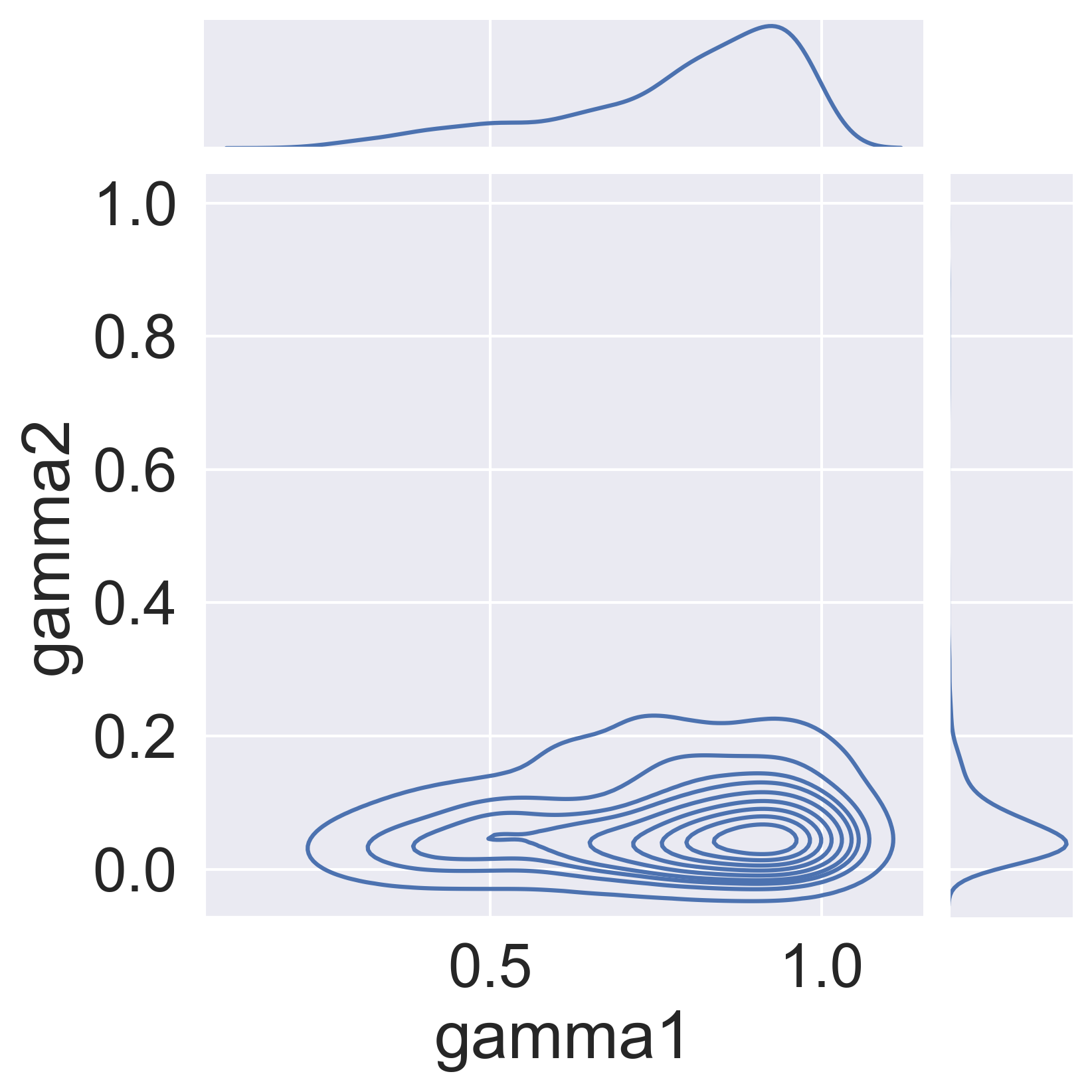}
         \caption{$\gamma_1$ and $\gamma_2$}
     \end{subfigure}
     \hfill
     \begin{subfigure}[b]{0.3\textwidth}
         \centering
         \includegraphics[width=\textwidth]{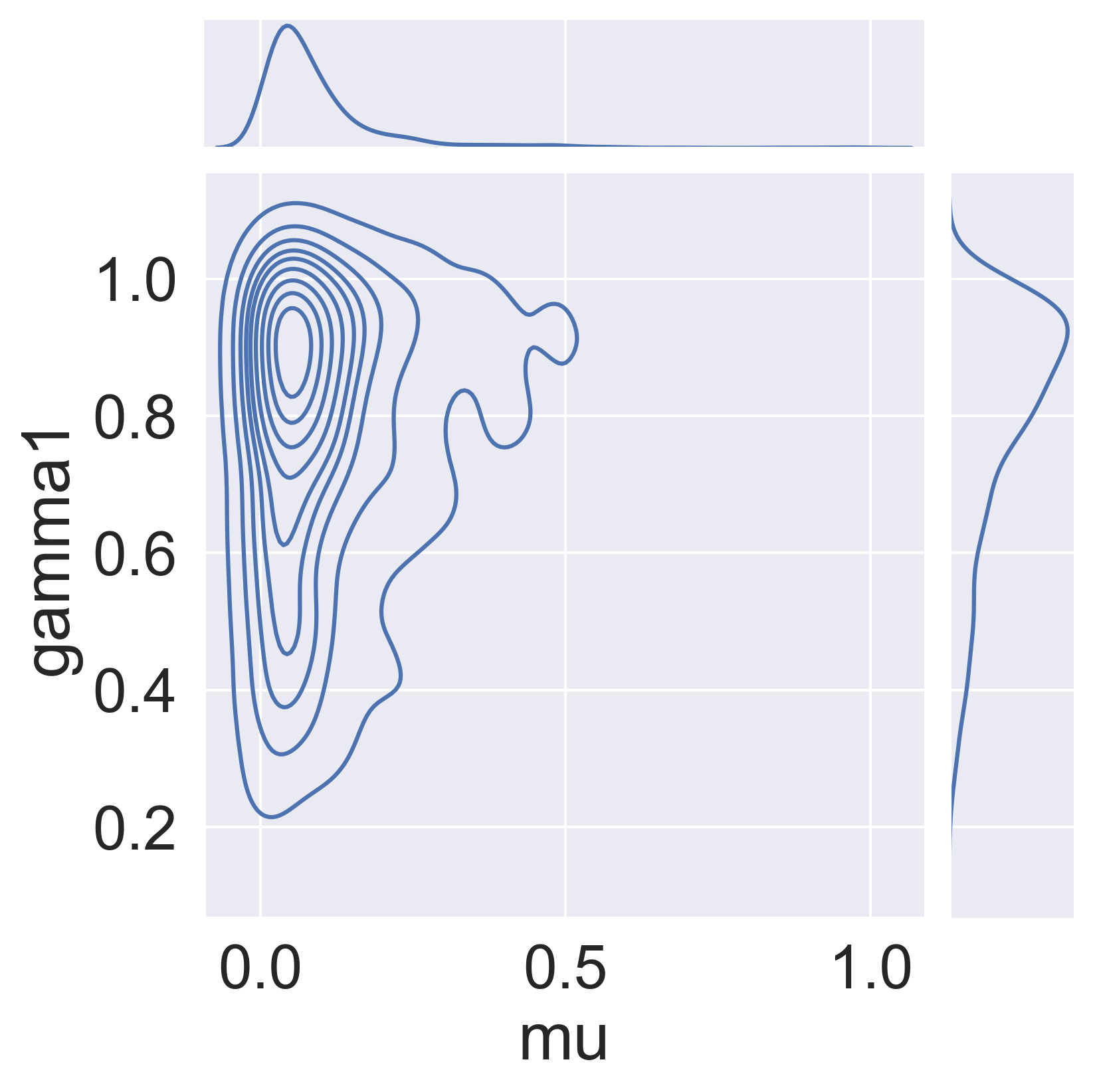}
         \caption{$\mu$ and $\gamma_1$}
     \end{subfigure}
     \hfill
     \begin{subfigure}[b]{0.3\textwidth}
         \centering
         \includegraphics[width=\textwidth]{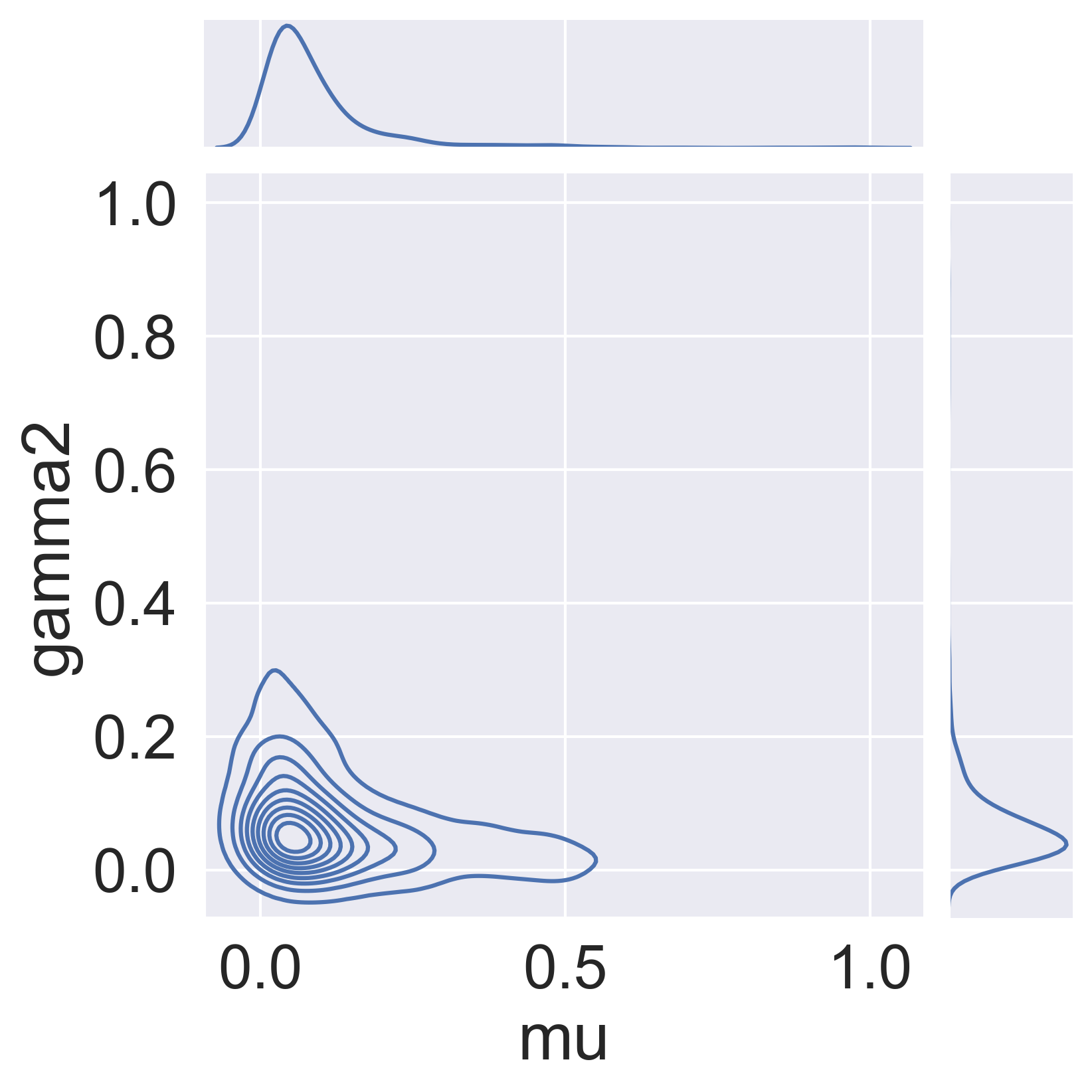}
         \caption{$\mu$ and $\gamma_2$}
     \end{subfigure}
        \caption{Joint posterior distribution for SIIDR parameters for wc\_8\_20s variant.}
        \label{fig:wc820_2_param}
\end{figure}

\begin{figure}
     \centering
     \begin{subfigure}[b]{0.3\textwidth}
         \centering
         \includegraphics[width=\textwidth]{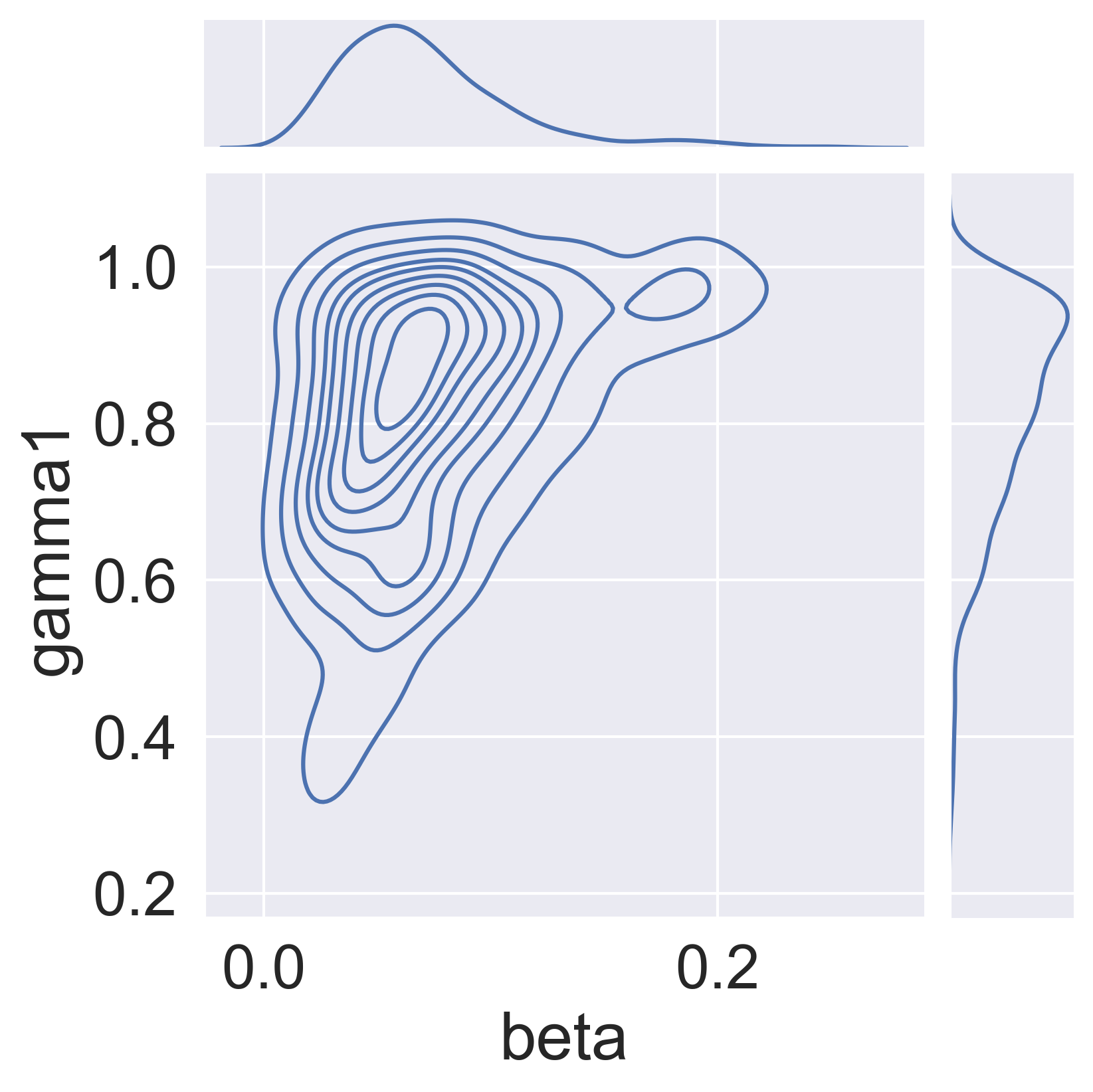}
         \caption{$\beta$ and $\gamma_1$}
     \end{subfigure}
     \hfill
     \begin{subfigure}[b]{0.3\textwidth}
         \centering
         \includegraphics[width=\textwidth]{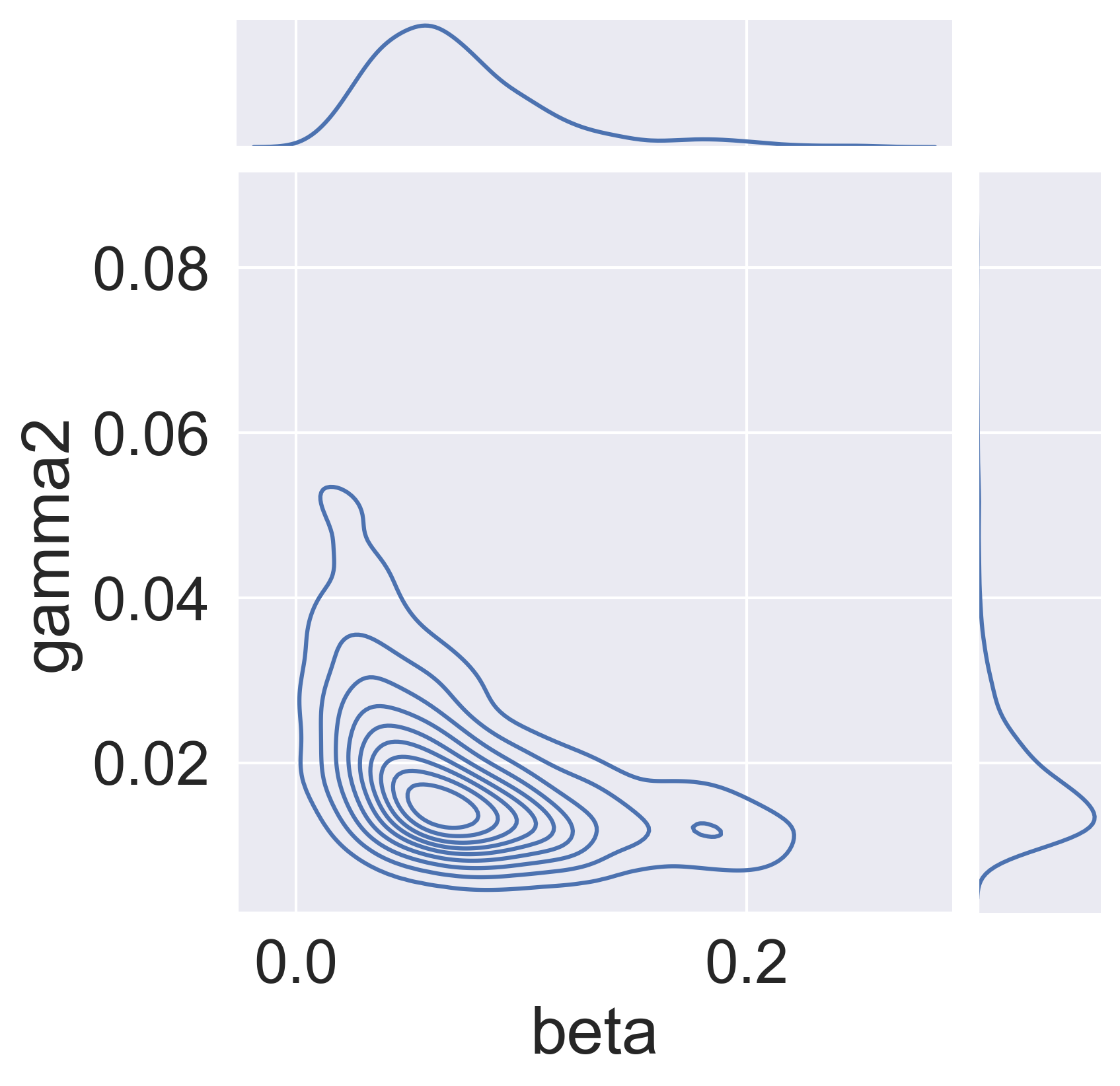}
         \caption{$\beta$ and $\gamma_2$}
     \end{subfigure}
     \hfill
     \begin{subfigure}[b]{0.3\textwidth}
         \centering
         \includegraphics[width=\textwidth]{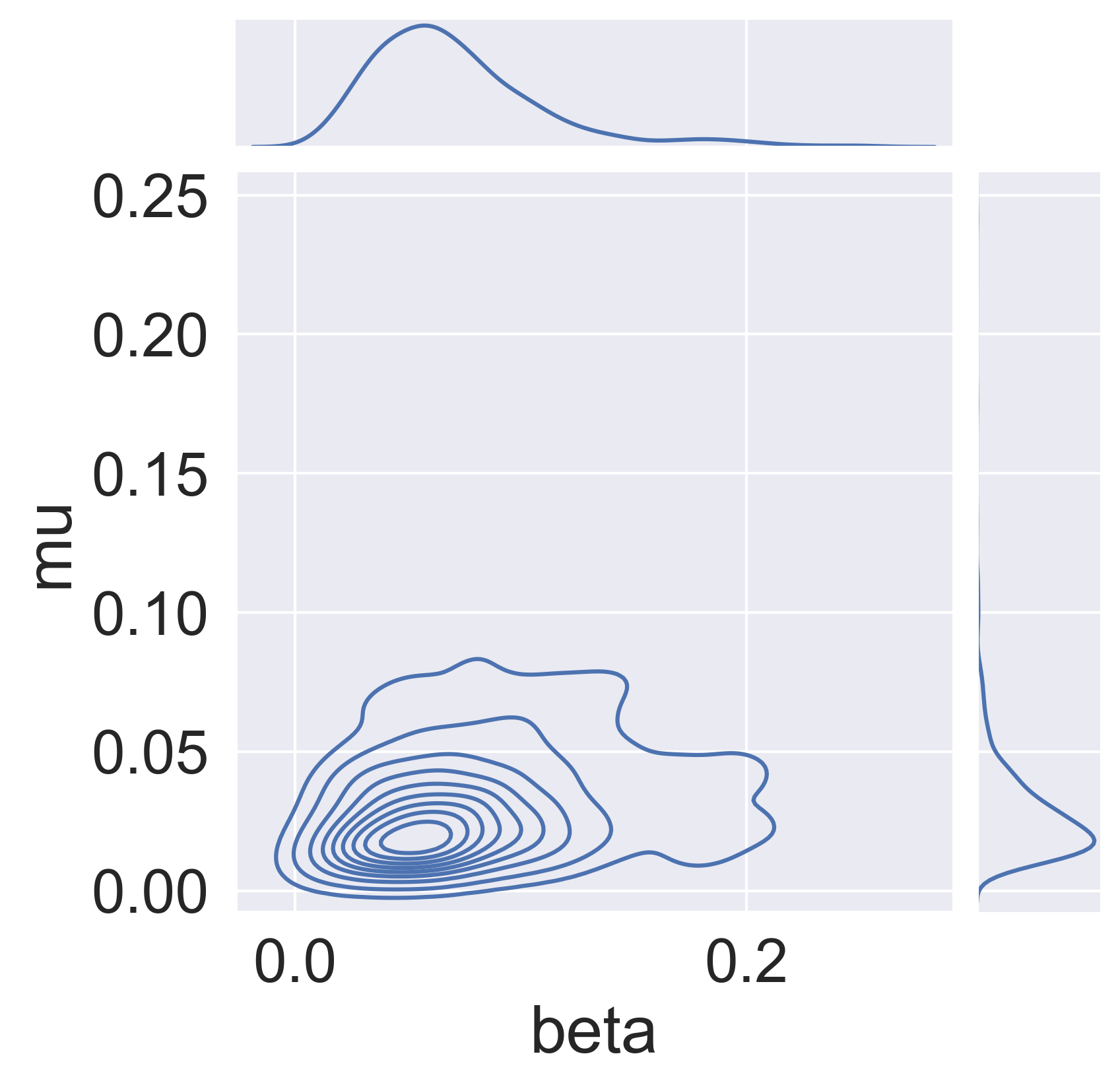}
         \caption{$\beta$ and $\mu$}
     \end{subfigure}
        \caption{Joint posterior distribution for SIIDR parameters for wc\_1\_5s variant.}
        \label{fig:wc15_1_param}
\end{figure}

\begin{figure}
     \centering
     \begin{subfigure}[b]{0.3\textwidth}
         \centering
         \includegraphics[width=\textwidth]{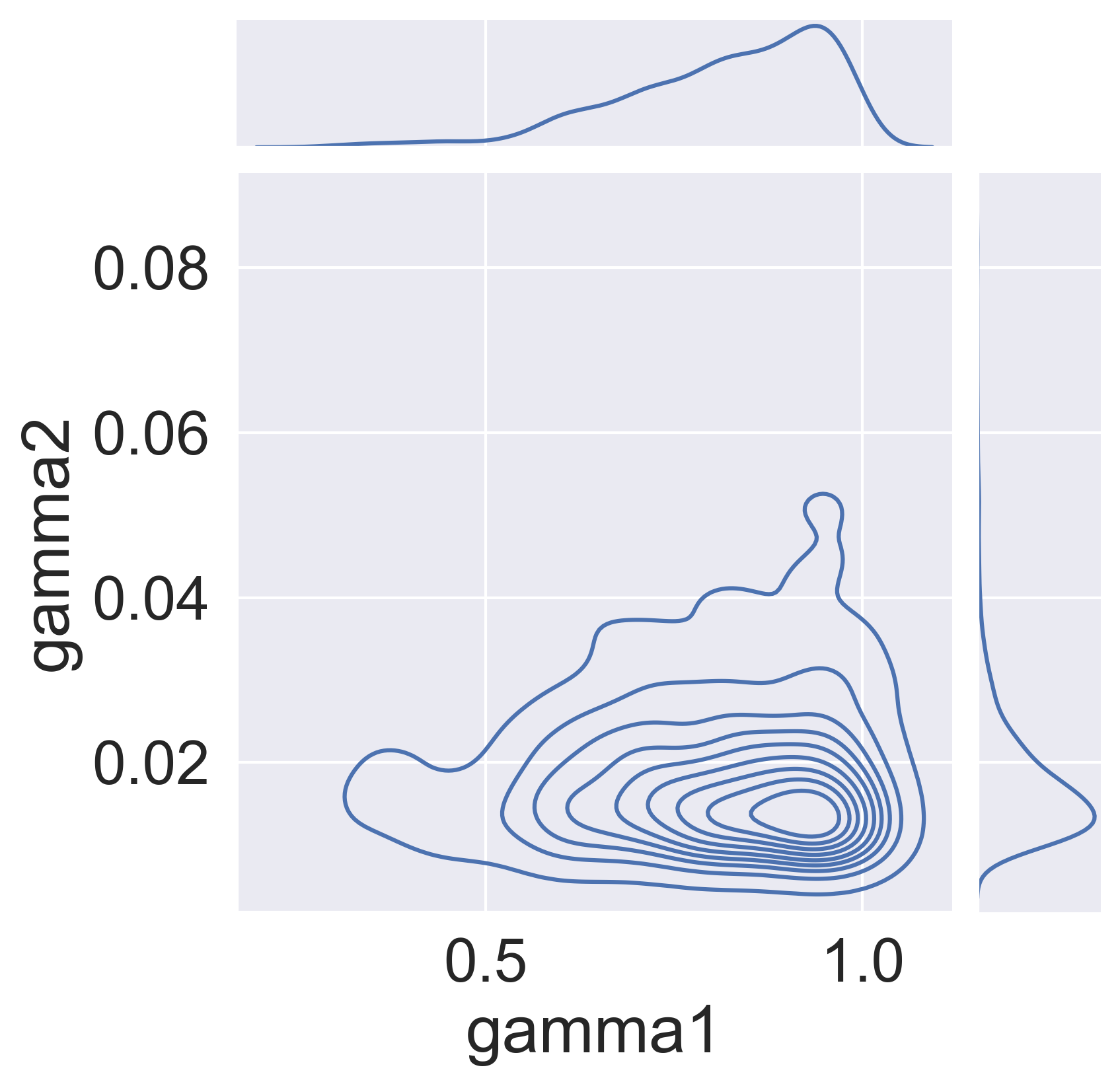}
         \caption{$\gamma_1$ and $\gamma_2$}
     \end{subfigure}
     \hfill
     \begin{subfigure}[b]{0.3\textwidth}
         \centering
         \includegraphics[width=\textwidth]{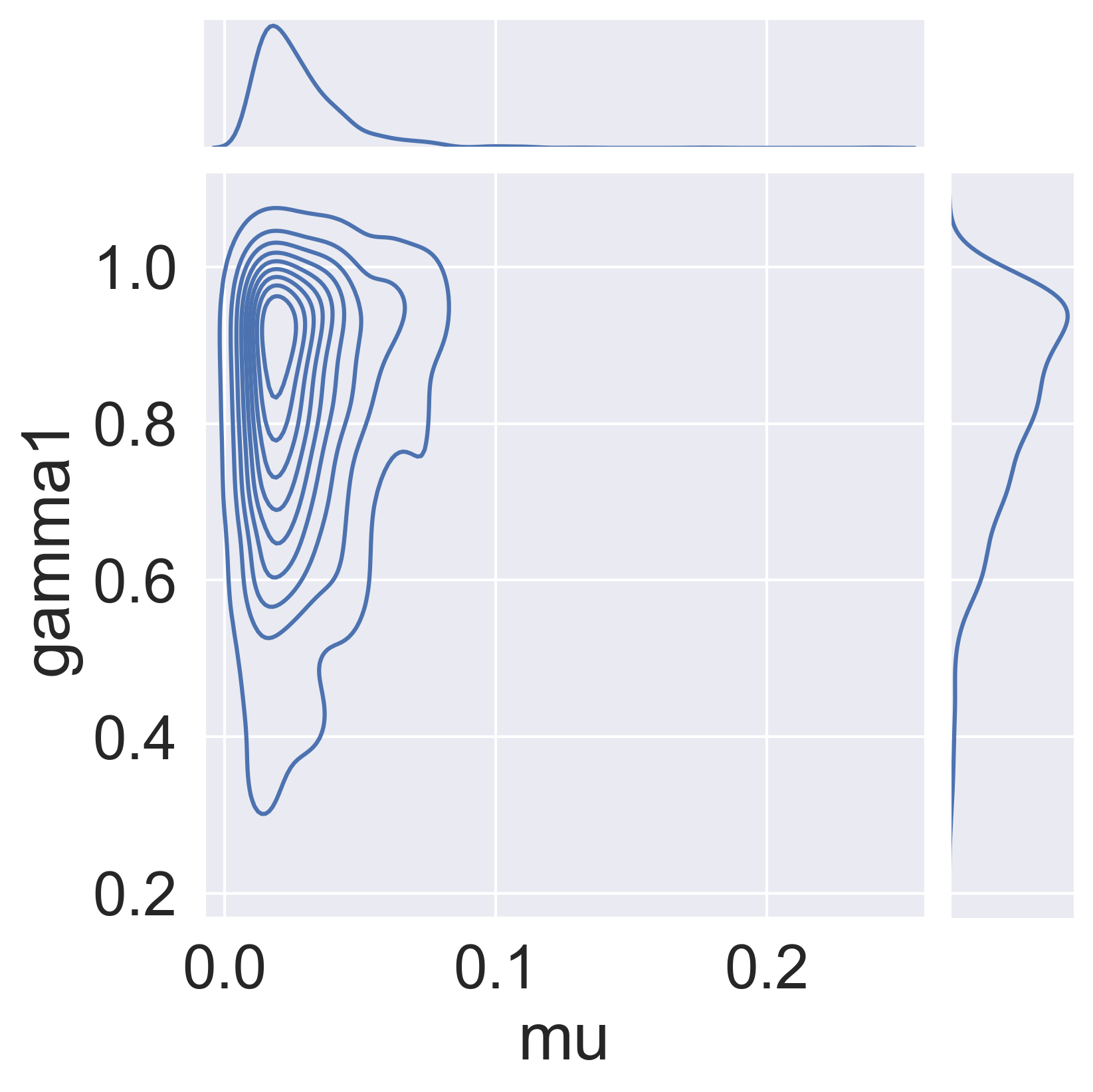}
         \caption{$\mu$ and $\gamma_1$}
     \end{subfigure}
     \hfill
     \begin{subfigure}[b]{0.3\textwidth}
         \centering
         \includegraphics[width=\textwidth]{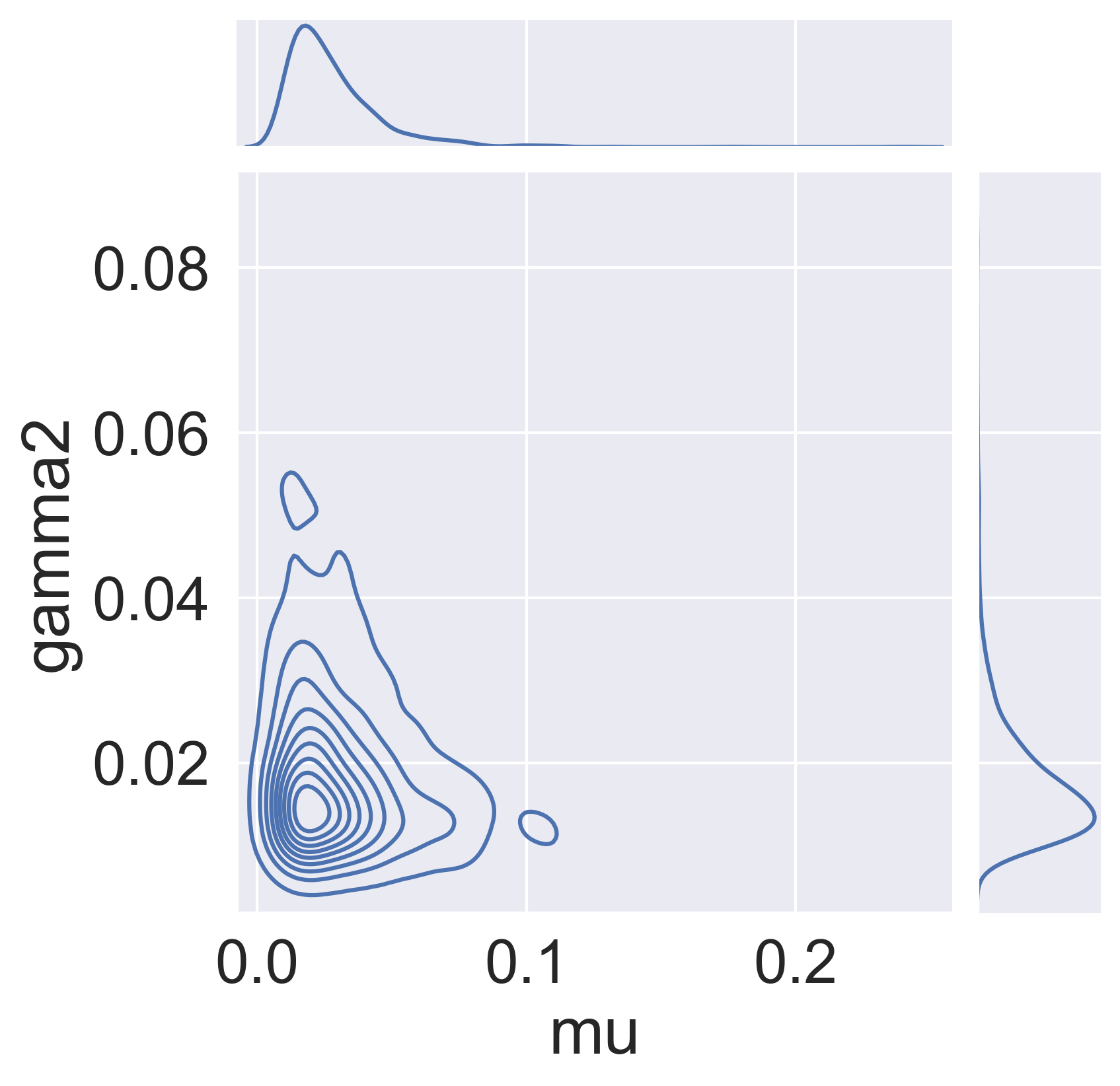}
         \caption{$\mu$ and $\gamma_2$}
     \end{subfigure}
        \caption{Joint posterior distribution for SIIDR parameters for wc\_1\_5s variant.}
        \label{fig:wc15_2_param}
\end{figure}

\section{Linearization of SIIDR as NLDS}
\label{sec:appB}

The Jacobian matrix $\mathcal{J}$ at the equilibrium point $P^*$ is defined as: 
\begin{align}
\mathcal{J}=\nabla g(P^*),
\end{align}
where $\mathcal{J}_{i,j}=[\nabla g(P^*)]_{i,j}=\frac{\partial g_i}{\partial p_j}|_{P=P^*}$. 

We calculate the partial first order derivatives of our equation system and obtain the Jacobian matrix:
\begin{equation}
\mathcal{J}= 
\begin{bmatrix}
\mathds{O}&-\mathds{I}&-(\alpha_{I_DI}+\alpha_{I_DI_D})\mathds{I}&-(\alpha_{II}+\alpha_{II_D})\mathds{I}\\
\mathds{O}&\mathds{I}&\mathds{O}&-x_S\tilde{\beta} \mathds{A}\\
\mathds{O}&\mathds{O}&\alpha_{I_DI_D}\mathds{I}&\alpha_{II_D}\mathds{I} \\
\mathds{O}&\mathds{O}&\alpha_{I_DI}\mathds{I}& x_S\tilde{\beta} \mathds{A} +\alpha_{II}\mathds{I}\\
\end{bmatrix}
\end{equation}

The size of the Jacobian matrix is $4N \times 4N$, where $N$ is the number of nodes in the graph. Every row has 4 elements of size $N \times N$. We use the following notation: $\mathds{I}$ is the identity matrix of size $N \times N$ and $\mathds{O}$ is a matrix of size $N \times N$ with all zeros. $\mathds{A}$ is the adjacency matrix of the network represented as a graph, of size $N \times N$.

The first row is a linear combination of the other rows, thus:
\begin{equation}
\mathcal{J}= 
\begin{bmatrix}
\mathds{I}&\mathds{O}&-x_S\tilde{\beta} \mathds{A}\\
\mathds{O}&\alpha_{I_DI_D}\mathds{I}&\alpha_{II_D}\mathds{I} \\
\mathds{O}&\alpha_{I_DI}\mathds{I}& x_S\tilde{\beta} \mathds{A} +\alpha_{II}\mathds{I}\\
\end{bmatrix}
\end{equation}

Let us represent the Jacobian matrix as follows:
\begin{equation}
\mathcal{J}= 
\begin{bmatrix}
Q_1&Q_2\\
O&Q_3\\
\end{bmatrix}
\end{equation}
where $Q_1$, $Q_2$, $Q_3$, $O$ are matrices of size $N \times N$, $2N \times N$, $2N \times 2N$, $2N \times N$ respectively:
\begin{equation}
\begin{split}
Q_1=\mathds{I},
\quad
Q_2= 
\begin{bmatrix}
\mathds{O}&-x_S\tilde{\beta} \mathds{A}
\end{bmatrix}, 
\quad
Q_3= 
\begin{bmatrix}
\alpha_{I_DI_D}\mathds{I}&\alpha_{II_D}\mathds{I}\\
\alpha_{I_DI}\mathds{I}& x_S\tilde{\beta} \mathds{A}
\end{bmatrix}
\end{split}
\end{equation}

Let $\Vec{v}$ of size $3N \times 1$ and $\lambda_J$ be the eigenvector and the eigenvalue of $J$ respectively. Then we can define $\Vec{v}$ to be composed of $\Vec{v_1}$ of size $N\times 1$ and $\Vec{v_2}$ of size $2N \times 1$:
\begin{align}
\Vec{v}= 
\begin{bmatrix}
\Vec{v_1}\\
\Vec{v_2}\\
\end{bmatrix}
\end{align} 
$\Vec{v}$ and $\lambda_J$ satisfy the following equation:
\begin{align}
J\Vec{v} = \lambda_J\Vec{v}
\end{align}
which results in:
\begin{align}
\begin{bmatrix}
Q_1&Q_2\\
O&Q_3\\
\end{bmatrix}
\begin{bmatrix}
\Vec{v_1}\\
\Vec{v_2}\\
\end{bmatrix}
= \lambda_J
\begin{bmatrix}
\Vec{v_1}\\
\Vec{v_2}\\
\end{bmatrix}
\label{eq:JB}
\end{align}

\noindent Eq.~\ref{eq:JB} implies that:
\begin{align}
Q_1 \cdot \Vec{v_1} + Q_2 \cdot \Vec{v_2} = \lambda_J \cdot \Vec{v_1} \label{eq1_b}\\
Q_3 \cdot \Vec{v_2} = \lambda_J \cdot \Vec{v_2} \label{eq2_b}
\end{align}

\noindent From Eq.~\ref{eq2_b} we have:
\begin{enumerate}
    \item $\Vec{v_2} = \Vec{0}$, or
    \item $\Vec{v_2}$ is the eigenvector of $Q_3$ and $\lambda_J$ is the eigenvalue of $Q_3$.
\end{enumerate}  

We look at the first case into more detail: if $v_2$ = 0, from Eq.~\ref{eq1_b}, we obtain that $Q_1 \cdot \Vec{v_1} = \lambda_J \cdot \Vec{v_1}$. That means either: (a) $\Vec{v_1} = 0$, which is not feasible, because in this case $\Vec{v}= \Vec{0}$, or (b) $\lambda_J$ is the eigenvalue of $Q_1$. 

Thus, the eigenvalues of the Jacobian matrix can be represented as eigenvalues of matrix $Q_1$ (when $\Vec{v_2}$ = 0) and eigenvalues of matrix $Q_3$. Given the structure of $Q_1$ (i.e., identity matrix of size $N\times N$), the eigenvalues of $Q_1$ are equal to $\Vec{1}$. Thus, we can conclude that the Jacobian matrix has at least one eigenvalue equal to 1.

\section{SIIDR Stability as the System of NLDS}
\label{sec:app2}

\begin{theorem}
The equilibrium points of SIIDR represented as NLDS of the form~(\ref{eq:sys_nlds})
are Lyapunov stable if:
\begin{align}
\lambda_A \frac{\tilde{\beta} }{\mu} \leq 1\label{eq:thres1}
\end{align}
where $\lambda_A$ is the largest eigenvalue of the adjacency matrix, $\tilde{\beta}$ and $\mu$ are probabilities of infection and recovery respectively. 
\end{theorem}
\begin{proof}System~(\ref{eq:sys_nlds}) can be reduced to the first three equations because of linear dependency of $P_{R,i,t+1}$ on other equations, and has the following representation in the matrix form:
\begin{align*}
g(P_{t+1}) = CP_t + \mathcal{P}^T_tBP_t
\end{align*}
where matrices $C$ and $\mathcal{P}^T_tBP_t$ of size $3N \times 3N $ correspond to the linear and non-linear part of the system, respectively. $
\mathcal{P^T} = \{\mathcal{P}_1^T,\mathcal{P}_2^T,\mathcal{P}_3^T\}$ is a $3N \times 9N$ matrix, where $\mathcal{P}_i^T$ is a $3N \times 3N$ matrix with non-zero $i_{th}$ row $P_S, P_I, P_{I_D}$:
\begin{equation*}
\mathcal{P}_i^T= 
\begin{bmatrix}
\mathds{O}&\mathds{O}&\mathds{O}\\
P_S&P_I&P_{I_D}\\
\mathds{O}&\mathds{O}&\mathds{O}\\
\end{bmatrix}
\end{equation*}
$B = \{B_i\}_{i = 1}^3$ is a $9N \times 3N$ matrix where $B_i = \{b_{kl}\}_{k, l = 1}^3$ has the size of $3N\times 3N$.
Based on our system representation~(\ref{eq:sys_nlds}) matrix $C$ is the following:
\begin{equation*}
C= 
\begin{bmatrix}
\mathds{I}&\mathds{O}&\mathds{O}\\
\mathds{O} & \alpha_{II}\mathds{I} & \alpha_{I_DI}\mathds{I}\\
\mathds{O}&\alpha_{II_D}\mathds{I} &\alpha_{I_DI_D}\mathds{I}  \\
\end{bmatrix}
\end{equation*}
and matrix $B$ is:
\begin{equation*}
B= 
\begin{bmatrix}
\mathds{O}&\mathds{O}&\mathds{O}\\
-\tilde{\beta} \mathds{A} &\mathds{O}&\mathds{O}\\
\mathds{O}&\mathds{O}&\mathds{O}\\
\mathds{O}&\tilde{\beta} \mathds{A}&\mathds{O}\\
\mathds{O}&\mathds{O}&\mathds{O}\\
\mathds{O}&\mathds{O}&\mathds{O}\\
\mathds{O}&\mathds{O}&\mathds{O}\\
\mathds{O}&\mathds{O}&\mathds{O}\\
\mathds{O}&\mathds{O}&\mathds{O}
\end{bmatrix}
\end{equation*}
where $\mathds{A}$ is the adjacency matrix of the corresponding graph.

Let $L$ be the continuous function equal to $P^T K$, where $K$ is the $3N \times 1$ matrix:
\begin{equation*}
K= 
\begin{bmatrix}
\mathds{O}\\
\mathds{1}\\
\mathds{1}
\end{bmatrix}
\end{equation*}
Then
\begin{equation*}
L(P)= 
\begin{bmatrix}
P_S&P_I&P_{I_D}
\end{bmatrix}
\begin{bmatrix}
\mathds{O}\\
\mathds{1}\\
\mathds{1}
\end{bmatrix}=
\end{equation*}
\begin{align*}
\sum_{i = 1}^N(P_I + P_{I_D})_i
\end{align*}
$L$ is positive definite because it is equal to the sum of probabilities of all nodes in the graph be infected or infected dormant. The finite difference~(\ref{eq:diff}) in this case is equal to:
\begin{align*}
L(P_{t+1}) - L(P_t) = P^T_{k+1} K - P^T_tK =\\
[CP_t + \mathcal{P}^T_tBP_t]^T K - P_t K = \\
P^T[C^TK + B^T\mathcal{P}^TK -K]
\end{align*}
where:
\begin{equation*}
C^TK= 
\begin{bmatrix}
\mathds{I}&\mathds{O}&\mathds{O}\\
\mathds{O} & \alpha_{II}\mathds{I} & \alpha_{II_D}\mathds{I}\\
\mathds{O}&\alpha_{I_DI}\mathds{I} &\alpha_{I_DI_D}\mathds{I}  \\
\end{bmatrix}
\begin{bmatrix}
\mathds{O}\\
\mathds{1}\\
\mathds{1}
\end{bmatrix} = \\
\begin{bmatrix}
\mathds{O}\\
\alpha_{II}\mathds{I} + \alpha_{II_D}\mathds{I}\\
\alpha_{I_DI}\mathds{I} + \alpha_{I_DI_D}\mathds{I}  
\end{bmatrix} = 
\end{equation*}
\begin{equation*}
\begin{bmatrix}
\mathds{O}\\
\mathds{1}-\mu \mathds{1}\\
\mathds{1}
\end{bmatrix}
\end{equation*}
and
\begin{equation*}
B^T\mathcal{P}^TK= 
\begin{bmatrix}
\mathds{O}&\mathds{O}&\mathds{O}\\
-\tilde{\beta} \mathds{A} &\mathds{O}&\mathds{O}\\
\mathds{O}&\mathds{O}&\mathds{O}\\
\mathds{O}&\tilde{\beta} \mathds{A}&\mathds{O}\\
\mathds{O}&\mathds{O}&\mathds{O}\\
\mathds{O}&\mathds{O}&\mathds{O}\\
\mathds{O}&\mathds{O}&\mathds{O}\\
\mathds{O}&\mathds{O}&\mathds{O}\\
\mathds{O}&\mathds{O}&\mathds{O}
\end{bmatrix}^T
\begin{bmatrix}
\mathds{O}&P_S&\mathds{O}\\
\mathds{O}&P_I&\mathds{O}\\
\mathds{O}&P_{I_D}&\mathds{O}\\
\end{bmatrix}
\begin{bmatrix}
\mathds{O}\\
\mathds{1}\\
\mathds{1}
\end{bmatrix}=
\end{equation*}
\begin{equation*}
\begin{bmatrix}
\mathds{O}\\
\tilde{\beta}P_S\mathds{A}\\
\mathds{O}
\end{bmatrix}
\end{equation*}
thus, 
\begin{equation*}
C^TK + B^T\mathcal{P}^TK -K = 
\begin{bmatrix}
\mathds{O}\\
-\mu \mathds{1}+ \tilde{\beta}P_S\mathds{A}\\
\mathds{O}
\end{bmatrix}
\end{equation*}
which results in the condition:
\begin{equation*}
\begin{bmatrix}
P_S&P_I&P_{I_D}
\end{bmatrix}
\begin{bmatrix}
\mathds{O}\\
-\mu \mathds{1} + \tilde{\beta}P_S\mathds{A}\\
\mathds{O}
\end{bmatrix} \leq 0
\end{equation*}
or 
\begin{align}
P_I [\tilde{\beta}P_S\mathds{A} -\mu \mathds{1}] \leq 0 \label{eq:condition}
\end{align}
where $P_I$ is the $1 \times N$ vector of node probabilities to be infected, $P_S$ is the $1 \times N$ vector of node probabilities to be susceptible, and
$\mathds{A}$ is the adjacency matrix of the corresponding graph. Expression~(\ref{eq:condition})
means that the sum of probabilities of nodes to recover should be greater than the sum of probabilities of nodes to become infected at each time step for the equilibrium points of the system~(\ref{eq:sys_nlds}) to be Lyapunov stable.

As long as the maximum value of probabilities in the vector $P_S$ is 1, it is true that:
\begin{align}
\tilde{\beta}P_S\mathds{A} \leq \tilde{\beta}\mathds{1}\mathds{A} 
\end{align} 
So if we prove that:
\begin{align}
P_I [\tilde{\beta}\mathds{1}\mathds{A} -\mu \mathds{1}] \leq 0 
\end{align}
the condition~(\ref{eq:condition}) will be satisfied.

This condition can also be formulated by incorporating the nodes' degrees as follows:
\begin{align}
P_I [\tilde{\beta}\mathds{D} -\mu \mathds{1}] \leq 0 \label{eq:cond_degrees}
\end{align}
where $\mathds{D}$ is the $1 \times N$ vector where each element $d_i$ is equal to the degree of the node $i$ in the graph. 

As long as the maximum value of probabilities in the vector $P_I$ is 1, it is true that:
\begin{align}
P_I\tilde{\beta}\mathds{D} \leq \mathds{1}\tilde{\beta}\mathds{D} 
\end{align} 
So if we prove that:
\begin{align}
\mathds{1}\tilde{\beta}\mathds{D}  -\mu \mathds{1}\leq 0 \label{eq:cond_simple}
\end{align}
the condition~(\ref{eq:condition}) will be satisfied.
Condition~\ref{eq:cond_simple} can be rewritten as follows:
\begin{align}
\frac{\tilde{\beta} \sum_i^Nd_i}{N\mu} \leq 1
\end{align}
or
\begin{align}
\frac{\tilde{\beta} d_{ave}}{\mu} \leq 1
\end{align}
It is known that the largest eigenvalue $\lambda_A$ has the following lower bound in the case of an arbitrary graph:
\begin{align}
d_{ave} \leq \lambda_A \label{eq:cond_eigs}
\end{align}
where $d_{ave}$ is the average degree of the graph. Therefore it is true that 
\begin{align}
\frac{\tilde{\beta} d_{ave}}{\mu} \leq \frac{\tilde{\beta} \lambda_A}{\mu}  
\end{align}
Hence if the following condition:
\begin{align}
\frac{\tilde{\beta}\lambda_A}{\mu} \leq 1
\end{align}
is satisfied, then the DFE equilibrium point will be Lyapunov stable on an arbitrary graph.
\end{proof}

\section{Model Fitting and Parameter Estimation}
\label{sec:fitting}

In this section, we present the methodology used to compare different epidemic models in reproducing real WannaCry attack traces. Our method leverages the Akaike Information Criterion (AIC)~\cite{akaike1974} to select the model that best fits the spreading caused by WannaCry malware. We also discuss how we estimate the posterior distribution of the SIIDR transition rates using an Approximate Bayesian Computation approach based on Sequential Monte Carlo (ABC-SMC)~\cite{filippi2013optimality,mckinley2018approximate,toni2009approximate}.

\subsection{Model Selection}\label{sec:fitting_aic}
We use the AIC as guiding criterion to compare SIIDR to other epidemiological models, namely SI, SIS, SIR. The AIC is calculated based on the number of free parameters $k$ and the maximum likelihood estimate of the model $L$ as follows:
\begin{align}
    \text{AIC} = 2 k - 2 \ln{L}\label{eq:aic_eq}
\end{align}
The first term introduces a penalty that increases with the number of parameters and thus discourages overfitting. The second term rewards the goodness of fit that is assessed by the likelihood function. For the likelihood function, we use the least squares estimation.
The best model is the one with the lowest AIC. In the case of the least squares estimation, the AIC can be expressed as:
\begin{align*}
    \text{AIC} &= 2 k + n \ln{\hat \sigma^2}
\end{align*}
where:
\begin{align}
    \hat \sigma^2 &= \frac{\sum_{t=1}^T \hat\epsilon_i^2}{T}
\end{align}
and $\hat \epsilon_i$ are the estimated residuals:
\begin{align*}
\hat \epsilon_t = I^{sim}_t - I_t^{real}
\end{align*}
\noindent with $I^{sim}_t$ being the cumulative number of infected nodes from model simulations, and $I_t^{real}$ the cumulative number of infected nodes from real-world observations, at time interval $t$.

We use stochastic simulations~\citep{higham2001} to obtain a numerical approximation of the propagation process described by the system of ODEs. Generally, statistical methods such as stochastic simulations are a good approximation for larger systems, while in the case of smaller systems stochastic fluctuations become more important. 
The transitions among compartments are implemented through chain binomial processes~\citep{Abbey1952AnEO}. At step $t$ the number of entities in compartment $X$ transiting to compartment $Y$ is sampled from a binomial distribution $Pr^{Bin}(X(t), p_{X \rightarrow Y}(t))$, where $p_{X \rightarrow Y}(t)$ is the transition probability. If multiple transitions can happen from $X$ (e.g., $X \rightarrow Y$, $X \rightarrow Z$), a multinomial distribution is used (e.g., $Pr^{Mult}(X(t), p_{X \rightarrow Y}(t), p_{X \rightarrow Z}(t))$). 

The model selection methodology is summarized in Algorithm~\ref{alg:model_selection}. We start by creating a uniform grid of possible parameter values (lines~\ref{algo:1}-\ref{algo:4}). For each model and each set of parameter values $p =(\beta,\mu,\gamma_1,\gamma_2)$ we perform several stochastic experiments simulating the model dynamics (the run\_stochastic\_avg procedure).
Each stochastic realization consists of a time series, where $S(t), I(t), I_D(t), R(t)$ represent the number of nodes in each state at time interval $t$ during the simulation. The cumulative infection $I_{sim}$  consists of the total number of nodes in states $I, I_D$, and $R$, and is also a time series across all time intervals $dt$. Next, we compute the AIC using equation~(\ref{eq:aic_eq}) by comparing the simulated to the actual dynamic. We select the minimum AIC score for each model; the best model is the one with the minimum AIC score overall. 
\begin{algorithm}
\caption{SPM model selection}\label{alg:model_selection}
\begin{algorithmic}[1]
\Procedure{model\_selection}{}
\State $\beta \gets $ 20 equidistant values in (0,1) \label{algo:1}
\State $\mu \gets $ 20 equidistant values in (0,1)
\State $\gamma_1 \gets $ 10 equidistant values in (0,1)
\State $\gamma_2 \gets $ 10 equidistant values in (0,1) 
\label{algo:4}
\For{each model $m \in \{\mbox{SI, SIS, SIR, SIIDR}\}$}
    \For{each set $p = (\beta,\mu,\gamma_1,\gamma_2)$}
        \State {$S_i, I_i, I_{Di}, R_i$ = run\_stochastic\_avg($p, h$)} \label{algo:stoc}
        \State {$I_{sim}=I_i + I_{Di} + R_i$} 
        \State {$\text{aic}$ = AIC($I_{sim}, I_{real}$)} 
    \EndFor 
    \State {$\text{aic}_{min}^m \gets \min_{i}\text{aic}$}
\EndFor
\State {$\text{aic}_{min} \gets \min_{m}\text{aic}$}
\State {$M \gets$ model that corresponds to $\text{aic}_{min}$}
\State \textbf{return} $M$
\EndProcedure
\end{algorithmic}
\end{algorithm}

\revision{
\subsection{SIIDR Parameters Associated with the Best AIC Score}
\label{sec:aic_params}
}
\revision{In Table~\ref{tab:wc_aic_params} we show the SIIDR parameters associated with the minimum AIC score for all WC variants.}
\begin{table*}[!t]
\centering
\caption{SIIDR parameters associated with the minimum AIC score.}
\begin{tabular}{|c||c|c|c|c|}
\hline
\textbf{WannaCry} & $\beta$ & $\mu$ & $\gamma_1$ & $\gamma_2$ \\ 
\hline
wc\_1\_500s&0.01&0.01&0.99&0.12\\
wc\_1\_1s&0.01&0.01&0.66&0.77\\
wc\_1\_5s&0.01&0.11&0.88&0.01\\
wc\_1\_10s&0.01&0.01&0.77&0.55\\
wc\_1\_20s&0.01&0.58&0.55&0.34\\
wc\_4\_500ms&0.01&0.01&0.23&0.34\\
wc\_4\_1s&0.11&0.01&0.89&0.01\\
wc\_4\_5s&0.01&0.01&0.77&0.99\\
wc\_4\_10s&0.01&0.01&0.66&0.45\\
wc\_4\_20s&0.01&0.06&0.66&0.55\\
wc\_8\_500ms&0.01&0.01&0.12&0.66\\
wc\_8\_1s&0.22&0.53&0.34&0.01\\
wc\_8\_5s&0.01&0.01&0.34&0.99\\
wc\_8\_10s&0.01&0.01&0.12&0.66\\
wc\_8\_20s&0.11&0.16&0.55&0.01\\
\hline
\end{tabular}
\label{tab:wc_aic_params}
\end{table*}

\subsection{Posterior Distribution of Transition Rates}
\label{sec:fitting_smc}
To find the best set of parameters for the SIIDR model we can approximate the posterior distribution of the parameters using Approximate Bayesian Computation (ABC) techniques~\citep{MINTER2019100368}. These techniques are based on the Bayes rule for determining the posterior distribution of parameters given the data:
\begin{align}
P(\theta|D) = \frac{P(D|\theta)P(\theta)}{P(D)}\propto p(D|\theta)P(\theta)\label{eq:poster}, 
\end{align}
where $P(\theta)$ is the prior distribution of parameters that represents our belief about them and $P(D|\theta)$ is the likelihood function, i.e., the probability density function of the data given the parameters. Marginal likelihood of the data $P(D)$ does not depend on $\theta$, and therefore the posterior distribution $P(\theta|D)$ is proportional to the numerator in~(\ref{eq:poster}). 

ABC methods are useful when the likelihood function is unknown or is not feasible to estimate analytically. The simplest version of ABC techniques is called rejection algorithm and is illustrated in Algorithm~\ref{alg:abc_rejection}. Despite it simplicity, the rejection algorithm is generally slow at converging. Indeed, each iteration is independent from the previous ones and the prior distribution from which parameters are sampled is never updated. Furthermore, it is often difficult to decide, a priori, a reasonable threshold value $\epsilon$ that guarantees both fast convergence and accurate results.

\begin{algorithm}
\caption{ABC-rejection algorithm}\label{alg:abc_rejection}
\begin{algorithmic}[1]
\State Sample $\theta^*$ from the prior distribution $P(\theta)$.
\State Simulate SPM model $D^*$ using $\theta^*$.
\State If $\sum_{t=1}^T(D_t - D^*_t)^2 \leq \epsilon$ accept $\theta^*$, reject otherwise.
\State Repeat until $N$ particles $\theta^* = \{\theta^*_j, j = 1, \dots, N\}$ are accepted.
\end{algorithmic}
\end{algorithm}

In alternative to the rejection algorithm, we use here a more advanced ABC technique that leverages Sequential Monte Carlo (ABC-SMC)~\citep{toni2009approximate,mckinley2018approximate}. The ABC-SMC approach iteratively constructs generations of prior distributions by decreasing the rejection threshold over time. At the first generation, a given number of parameter sets (i.e., particles) is accepted from the starting prior distribution, while each prior distribution used in following generations is obtained as a weighted sample from the previous generation $\theta^*$ perturbed through a kernel $K(\theta|\theta^*)$. Common choices for the kernel are the uniform and multivariate normal distributions. A kernel with a large variance will prevent the algorithm from being stuck in the local modes, but will result in a huge number of particles being rejected, which is inefficient. Therefore, we use the multivariate normal distribution, where the covariance matrix is calculated considering $M$ nearest neighbors (MNN) of the particles from the previous generation ~\citep{filippi2013optimality}. The ABC-SMC-MNN algorithm is illustrated in Algorithm~\ref{alg:parameter_estimation}.

\begin{algorithm}
\caption{SIIDR parameters estimation}\label{alg:parameter_estimation}
\begin{algorithmic}[1]
\Require $G$ - number of generations, $N$ - number of particles, $M$ - number of nearest neighbors, $\epsilon_1>\epsilon_2>\epsilon_3>\dots>\epsilon_G$ - sequence of decreasing tolerance values for each generation of the particles
\State Set $g = 0$
\State Set $j = 0$
\State If $g = 0$, sample particle $\theta^{**}$ from prior distribution $P(\theta)$. Otherwise, sample $\theta^*$ from the previous generation of particles $\{\theta_{g-1}\}$ with weights $\{w_{g-1}\}$ and perturb to obtain $\theta^{**}\sim K(\theta|\theta^*)$
\State Generate $n$ model simulations $D_l^{**}$ using $\theta^*$ and calculate $\hat{P}(D|D^{**}) = 1/n \sum_{l = 1}^n(d(D, D_l^{**})<\epsilon_g)$
\State If $\hat{P}(D|D^{**}) = 0$ return to step 4
\State Set $\theta_g^{j}$ and calculate weights for the particle: 
    \begin{center}
    $w_g^j =
    \begin{cases}
      \hat{P}(D|D^{**})P(\theta^{**}), g = 1\\
      \frac{\hat{P}(D|D^{**})P(\theta^{**})}{\sum_{l = 1}^Nw_{g-1}^lK(\theta_g^j|\theta_{g-1}^l)}, g > 1\\
    \end{cases} $      
    \end{center}
\State If $j < N$ increment $j$ and go to step 4
\State Normalize weights: $\sum_{j = 1}^Nw_g^j = 1$
\State If $g < G$ increment $g$ and go to step 3
\end{algorithmic}
\end{algorithm}




\end{appendices}


\bibliography{sn-bibliography}

\end{document}